\documentclass[11pt]{article}  
    
\usepackage{fncylab,enumerate,wrapfig,placeins}

\usepackage{graphicx}
\usepackage{amsmath,amsthm,amsfonts,amssymb}  
\usepackage{geometry}
\pdfoutput=1  

\usepackage{subcaption}
\usepackage{booktabs}

\usepackage{multirow}

\usepackage{algorithm}
\usepackage{algpseudocode}
\algnewcommand\algorithmicinput{\textbf{Input:}}
\algnewcommand\INPUT{\item[\algorithmicinput]}

\usepackage[usenames,dvipsnames]{color}
\usepackage{color}

\usepackage[colorlinks,%
linkcolor=BrickRed,%
filecolor =BrickRed,%
citecolor=RoyalPurple,%
]{hyperref}

\newlength{\noteWidth}
\long\def\notes#1{\ifinner
	{\footnotesize #1}
	\else 
	\marginpar{\parbox[t]{\noteWidth}{\raggedright\footnotesize#1}}
	\fi\typeout{#1}}

\def\notes#1{}

\def\spm#1{\notes{SPM:  #1}} 
 
\def\archive#1{\typeout{save this note:  #1}} 
\def\rd#1{{\color{red}#1}} 
\def\bl#1{{\color{blue}#1}} 

\def\mindex#1{\index{#1}}

\def\sq{\hbox{\rlap{$\sqcap$}$\sqcup$}}
\def\qed{\ifmmode\sq\else{\unskip\nobreak\hfil
\penalty50\hskip1em\null\nobreak\hfil\sq
\parfillskip=0pt\finalhyphendemerits=0\endgraf}\fi\medskip}

\long\def\defbox#1{\framebox[.9\hsize][c]{\parbox{.85\hsize}{%
\parindent=0pt
\baselineskip=12pt plus .1pt      
\parskip=6pt plus 1.5pt minus 1pt 
 #1}}}

\long\def\beginbox#1\endbox{\subsection*{}%
\hbox{\hspace{.05\hsize}\defbox{\medskip#1\bigskip}}%
\subsection*{}}

\def\endbox{}

\def\transpose{{\hbox{\it\tiny T}}}

\newsavebox{\junk}
\savebox{\junk}[1.6mm]{\hbox{$|\!|\!|$}}

\def\limsup{\mathop{\rm lim\ sup}}

\def\argmin{\mathop{\rm arg\, min}}

\def\U{{\sf U}}

\def\state{{\sf X}}

\def\bx{{{\cal B}(\state)}}

\newcommand{\field}[1]{\mathbb{#1}}

\def\Re{\field{R}}

\def\ind{\field{I}}

\def\Co{\field{C}}

\def\bfmath#1{{\mathchoice{\mbox{\boldmath$#1$}}%
{\mbox{\boldmath$#1$}}%
{\mbox{\boldmath$\scriptstyle#1$}}%
{\mbox{\boldmath$\scriptscriptstyle#1$}}}}

\def\bfmq{\bfmath{q}}

\def\bfmC{\bfmath{C}}

\def\bfmG{\bfmath{G}}

\def\bfmU{\bfmath{U}}

\def\bfmX{\bfmath{X}}

\def\bfmY{\bfmath{Y}}

\def\bfmhhaY{\bfmath{\hhaY}} 
\def\bfmhhaY{\hbox to 0pt{$\widehat{\bfmY}$\hss}\widehat{\phantom{\raise 1.25pt\hbox{$\bfmY$}}}}

\def\bfmZ{\bfmath{Z}}

\def\bfPhi{\bfmath{\Phi}}

\def\haf{{\hat f}}

\def\hagamma{{\hat\gamma}}

\def\til={{\widetilde =}}

\def\tiltheta{\widetilde \theta}

\def\tiltheta{{\tilde \theta}}

\def\clA{{\cal A}}
\def\clB{{\cal B}}
\def\clC{{\cal C}}

\def\clE{{\cal E}}
\def\clF{{\cal F}}
\def\clG{{\cal G}}
\def\clH{{\cal H}}

\def\clI{{\cal I}}

\def\clN{{\cal N}}

\def\clQ{{\cal Q}}

\def\clT{{\cal T}}

 \def\FRAC#1#2#3{\genfrac{}{}{}{#1}{#2}{#3}}

\def\ddt{{\mathchoice{\FRAC{1}{d}{dt}}%
{\FRAC{1}{d}{dt}}%
{\FRAC{3}{d}{dt}}%
{\FRAC{3}{d}{dt}}}}

\def\ddtp{{\mathchoice{\FRAC{1}{d^{\hbox to 2pt{\rm\tiny +\hss}}}{dt}}%
{\FRAC{1}{d^{\hbox to 2pt{\rm\tiny +\hss}}}{dt}}%
{\FRAC{3}{d^{\hbox to 2pt{\rm\tiny +\hss}}}{dt}}%
{\FRAC{3}{d^{\hbox to 2pt{\rm\tiny +\hss}}}{dt}}}}

\def\half{{\mathchoice{\FRAC{1}{1}{2}}%
{\FRAC{1}{1}{2}}%
{\FRAC{3}{1}{2}}%
{\FRAC{3}{1}{2}}}}

\def\eqdef{\mathbin{:=}}

\def\Prob{{\sf P}}

\def\Expect{{\sf E}}

\def\average#1,#2,{{1\over #2} \sum_{#1}^{#2}}

\def\eye(#1){{\bf(#1)}\quad}
\def\epsy{\varepsilon}

\def\varble{\,\cdot\,}

\newtheorem{theorem}{Theorem}[section]

\newtheorem{proposition}[theorem]{Proposition}
\newtheorem{lemma}[theorem]{Lemma}

\def\Lemma#1{Lemma~\ref{#1}}
\def\Proposition#1{Prop.~\ref{#1}}
\def\Theorem#1{Theorem~\ref{#1}}

\def\Section#1{Section~\ref{#1}}

\def\Appendix#1{Appendix~\ref{#1}}

\def\eq#1/{(\ref{e:#1})}

\newcommand{\beqn}[1]{\notes{#1}%
\begin{eqnarray} \elabel{#1}}

\newcommand{\eeqn}{\end{eqnarray} }

\newcommand{\beq}[1]{\notes{#1}%
\begin{equation}\elabel{#1}}

\newcommand{\eeq}{\end{equation}}

\def\bdes{\begin{description}}
\def\edes{\end{description}}

\def\barc{{\overline {c}}}

\def\barf{{\overline {f}}}

\def\barT{{\overline {T}}}

\def\bartheta{{\overline{\theta}}}

\def\barmu{{\overline{\mu}}}

\newcounter{rmnum}
\newenvironment{romannum}{\begin{list}{{\upshape (\roman{rmnum})}}{\usecounter{rmnum}
\setlength{\leftmargin}{14pt}
\setlength{\rightmargin}{8pt}
\setlength{\itemsep}{2pt}
\setlength{\itemindent}{-1pt}
}}{\end{list}}

\newcounter{anum}

{\end{list}}

\def\ass(#1:#2){(#1\ref{#1:#2})}

\def\ritem#1{
\item[{\sf \ass(\current_model:#1)}]
}

\newenvironment{recall-ass}[1]{%
\begin{description}
\def\current_model{#1}}{
\end{description}
}

\def\Ebox#1#2{%
\begin{center}
 \parbox{#1\hsize}{\epsfxsize=\hsize \epsfbox{#2}}
\end{center}}

\newcommand{\bd}{\begin{description}}
\newcommand{\ed}{\end{description}}
\newcommand{\bt}{\begin{theorem}}
\newcommand{\et}{\end{theorem}}
\newcommand{\ba}{\begin{array}{rcl}}
\newcommand{\ea}{\end{array}}

\def\assume#1{\smallbreak\noindent\textbf{#1}}

\def\nd{\ell}  
\def\ninp{\ell_u}  
\def\nphi{\ell_\phi}

\def\td{d}

\def\BE{{\cal B}}
\def\maxBE{\overline{\cal B}}

\def\Real{\text{Re}}

\def\upmf{\mu}

\def\qcost{c}
\def\bfqcost{\bfmath{c}}

\def\bfmG{\bfmath{G}}

\def\bfalpha{\bfmath{\alpha}}

\def\state{{\sf X}}

\def\eqdef{\mathbin{:=}}

\def\elig{\zeta}
\def\bfelig{\bfmath{\zeta}}

\def\uq{\underline{q}}
\def\uQ{\underline{Q}}
\def\utheta{\underline{\theta}}

\def\trace{\hbox{\rm trace\,}}

\def\Lemma#1{Lemma~\ref{#1}}
\def\Proposition#1{Prop.~\ref{#1}}
\def\Prop#1{Prop.~\ref{#1}}
\def\Theorem#1{Thm.~\ref{#1}}   

\def\hab{\widehat b}
\def\haA{\widehat A}
\def\haG{\widehat G}

\def\tilb{\tilde{b}}

\def\tilp{\widetilde{p}}

\def\tilA{\tilde{A}}

\def\interest{\Gamma}

\def\bfgamma{\bfmath{\gamma}}

\def\bfDelta{\bfmath{\Delta}}

\def\Cov{\text{\rm Cov}\,}

\graphicspath{{SFQFigures/}}

\def\Ebox#1#2{%
	\begin{center}
		\includegraphics[width= #1\hsize]{#2} 
\end{center}}

\def\Fig#1{Fig.~\ref{#1}}

\def\ind{\field{I}}

\def\Re{\field{R}}
 
\def\diagpie{\Pi}
\def\pie{\varpi}

\def\barGamma{\overline{\Gamma}}
\def\Qstar{\clQ}
\def\partialQstar{\partial\!\clQ}
\def\barqcost{\bar c}
\def\barq{\bar q}

\def\barclA{\bar{\cal A}}
\def\haC{\widehat C}

\def\bfmhC{{\widehat{\bfmC}}}

\title{Fastest Convergence for Q-Learning}

\author{ Adithya M. Devraj  and Sean P. Meyn
	\thanks{Research   supported by    the  National Science Foundation under grants CPS-0931416 and EPCN-1609131}
	\thanks{
		A.D. and S.M. are with the Department of Electrical and Computer
		Engg.\ at the University of Florida, Gainesville. A part of the research was conducted when the authors were visitors at the Simons Institute for the Theory of Computing at University of California, Berkeley.}%
}

\begin{document}
\maketitle

\begin{abstract}   

The  Zap~Q-learning algorithm introduced in this paper is an improvement of Watkins' original algorithm and recent competitors in several respects.   It is a matrix-gain algorithm designed so that its asymptotic variance is optimal.   Moreover, an ODE analysis suggests that the transient behavior is a close match to a deterministic Newton-Raphson implementation. This is made possible by a  two time-scale update equation for the matrix gain sequence. 

The analysis suggests that the approach will lead to stable and efficient computation even for non-ideal parameterized settings. Numerical experiments confirm the quick convergence, even in such non-ideal cases.     The comparison plot on this first page, taken from \Fig{6stateBEPlot} of this paper,  is an illustration of the amazing acceleration in convergence using the new algorithm. 

A secondary goal of this paper is tutorial.   The first half of the paper contains a survey on  reinforcement learning algorithms, with a focus on minimum variance algorithms.

\medskip

\smallskip

{\small
	
}

\vfill

			\Ebox{1}{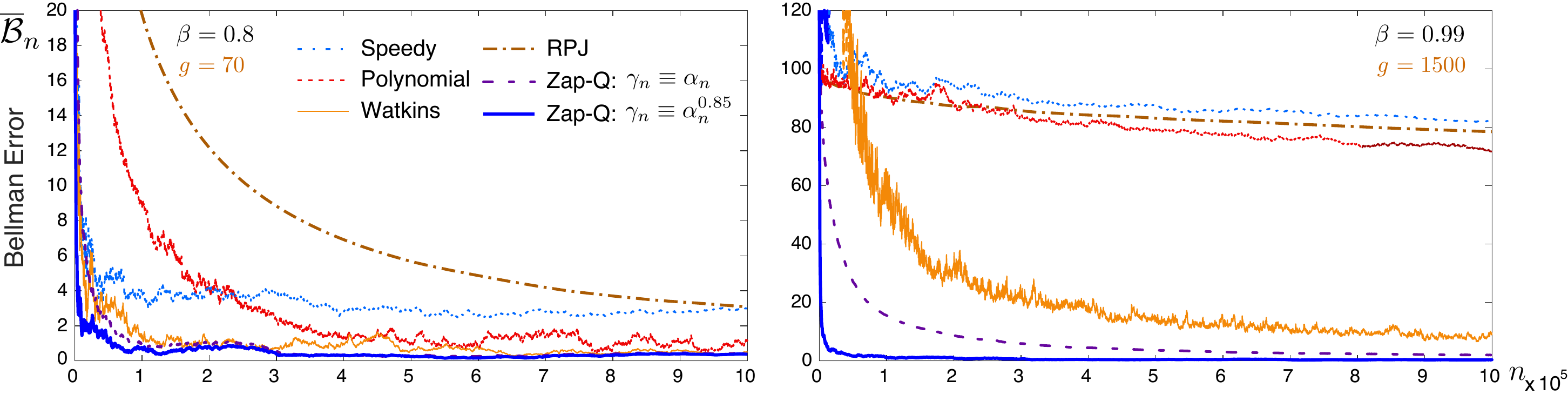}
\vfill
		
\end{abstract}
	  


\def\todo{{\color{red}To do:}\ }

\section{Introduction}

It is recognized that  algorithms for reinforcement learning such as TD- and Q-learning can be slow to converge.   The poor performance of Watkins' Q-learning algorithm was first quantified in \cite{sze97}, and since then many papers have appeared with proposed improvements, such as \cite{eveman03,azamunghakap11}.  
\archive{does has10 offer improved performance?  If not, let's leave it out}

An emphasis in much of the literature is computation of finite-time PAC (probably almost correct) bounds as a metric for performance.  Explicit bounds were obtained in \cite{sze97} for Watkins' algorithm,  and in \cite{azamunghakap11} for the  ``speedy'' Q-learning algorithm that was introduced by these authors.   A general theory is presented in  \cite{moubac11} for stochastic approximation algorithms.

\archive{Do we cite this \cite{bacmou13}?  Where does it fit?}
    
In each of the models considered in  prior work,  the update equation for the parameter estimates   can be expressed
\begin{equation}
\theta_{n+1} = \theta_n +\alpha_n [ \barf(\theta_n) + \Delta_{n+1} ]\,,\quad n\ge 0\, ,
\label{e:SAintro}
\end{equation}
in which $\{  \alpha_n \}$ is a positive gain sequence, and $\{  \Delta_n \}$ is a martingale difference sequence.    This representation is critical in analysis, but unfortunately is not typical in reinforcement learning applications outside of these versions of Q-learning.   For Markovian models, the usual transformation used to obtain a representation similar to \eqref{e:SAintro} results in an error sequence $\{  \Delta_n \}$ that is the sum of a  martingale difference sequence and a \textit{telescoping sequence} \cite{mamakshw90}.   It is the telescoping sequence that prevents easy analysis of  Markovian models.

This gap in the research literature carries over to the general theory of Markov chains.  Examples of concentration bounds for i.i.d.\ sequences or martingale-difference sequences include the finite-time bounds of Hoeffding and Bennett.   Extensions to Markovian models either offer very  crude bounds  \cite{meytwe94b}, or   restrictive assumptions \cite{lunmeytwe96a,glyorm02};  this remains an active area of research \cite{pau15}.

\archive{  \todo  (AD)
\\
Essay on new interest in Q-Learning (deep learning reference) 
cite{sutmaesze09} \rd{AD: Also cite Double Q?}  Please explain to me in person.
\\
Note that all this work claims that convergence rate must blow up with discount factor.   Crap!  \rd{Crap?  I thought crisis = opportunity?}}

In contrast, \textit{asymptotic} theory for stochastic approximation (as well as general state space Markov chains) is  mature.  Large Deviations or Central Limit Theorem (CLT) limits hold under very general assumptions \cite{benmetpri90,kusyin97,bor08a}.

The CLT will be a guide to algorithm design in the present paper.  For a typical stochastic approximation algorithm,  this takes the following form:   denoting $\{ \tiltheta_n\eqdef \theta_n-\theta^* : n\ge 0\}$ to be the error sequence,   under general conditions the scaled sequence   $\{\sqrt{n} \tiltheta_n : n\ge 1\}$	converges in distribution to a Gaussian distribution,  $\clN (0,\Sigma_\theta)$.   Typically,  the scaled covariance is also convergent:
\begin{equation}
\Sigma_\theta=\lim_{n\to\infty} n \Expect[\tiltheta_n\tiltheta_n^\transpose]\, .
\label{e:SAsigma}
\end{equation}  
The limit is known as the \textit{asymptotic covariance}.  

An asymptotic bound such as \eqref{e:SAsigma} may not be satisfying for practitioners of  stochastic optimization or reinforcement learning, given the success of finite-$n$ performance bounds in prior research.    There are however good reasons to apply this asymptotic theory in algorithm design:
\begin{romannum}
 
\item 
The   asymptotic covariance $\Sigma_\theta$ has a simple representation as the solution to a Lyapunov equation.  It is easily improved or optimized  by  design.   

\item 
As shown in examples in this paper,  the asymptotic covariance is often a good predictor of finite-time performance, since the CLT approximation is accurate for reasonable values of $n$.
\end{romannum}
Two approaches are known for optimizing the asymptotic covariance.   First is the remarkable averaging technique of Polyak and Juditsky \cite{pol90,poljud92} and Ruppert \cite{rup85}  (\cite{kontsi04} provides an accessible treatment in a simplified setting).  Second is what we will call \textit{Stochastic Newton-Raphson}, based on a special choice of matrix gain for the algorithm.   The second approach underlies the analysis of the averaging approach.

We are not aware of theory that distinguishes the performance of Polyak-Ruppert averaging as compared to the Stochastic Newton-Raphson method.    It is noted in \cite{moubac11} that the averaging approach often leads to very large transients, so that the algorithm should be modified (such as through projection of parameter updates). This may explain why averaging is not very popular in practice.    In our own numerical experiments it is observed  that the rate of convergence of CLT in this case is slow when compared to matrix gain methods.

\archive{I've removed some refs for now, such as cite{anb78,gol88,kiewol52}}

In addition to accelerating the convergence rate of standard algorithms for reinforcement learning,  it is hoped that this paper will lead to entirely new algorithms.    In particular, there is little theory to support Q-learning in non-ideal settings in which the optimal ``$Q$-function'' does not lie in the parameterized function class.  Convergence results have been obtained  for a class of optimal stopping problems \cite{yuber13}, and  for deterministic models \cite{mehmey09a}. There is now intense practical interest, despite an incomplete theory.  A stronger supporting theory will surely lead to more efficient algorithms.

\paragraph{Contributions}    

A new class of algorithms is proposed, designed to more accurately mimic the classical Newton-Raphson algorithm.    It is based on a two time-scale stochastic approximation algorithm,  constructed so that the matrix gain tracks the  gain that would be used in a deterministic Newton-Raphson method.   

The application of this approach to reinforcement learning results in the new \textit{Zap~Q-learning} algorithms.  A full analysis is presented for the special case of a complete parameterization (similar to the setting of Watkins' original algorithm).  It is found that the associated ODE has a remarkable and simple representation, which implies consistency under suitable assumptions.   Extensions to non-ideal parameterized settings are also proposed, and numerical experiments show dramatic variance reductions.   Moreover, results obtained from finite-$n$ experiments show close  solidarity with asymptotic theory.


The potential complexity introduced by the matrix gain is not of great concern in many cases,  because of the dramatically acceleration in the rate of convergence.   Moreover, the main contribution of this paper is not a single algorithm but  a class of algorithms, wherein the computational complexity can be dealt with separately. For example, in a parameterized setting, the basis functions can be intelligently pruned via random projection \cite{barbor08a}.

\smallbreak

The remainder of the paper is organized as follows.  
Background on computing and optimizing the asymptotic covariance is contained in 
\Section{s:SNR}.   Application to Q-learning, and theory surrounding the new Zap Q-learning algorithm is developed in \Section{s:Q}.  Numerical results are surveyed in
\Section{s:num},  and conclusions are contained in 
\Section{s:conc}.    The proofs of the main results are contained in the Appendix;  the final page contains Table~\ref{t:notation} containing a list of notation.

\section{Stochastic Newton Raphson and TD-Learning}
\label{s:SNR}

This first section is largely a tutorial on reinforcement learning.   
It is shown that the LSTD($\lambda$) learning algorithm  of  \cite{brabar96,boy02,nedber03a}
is an instance of the ``SNR algorithm'',  in which there is only one time-scale for the parameter and matrix-gain updates.   The original motivation for the LSTD($\lambda$) algorithm had no connection with asymptotic variance.  It was shown later in \cite{kon02} that the LSTD ($\lambda$) algorithm is the minimum asymptotic variance version of the TD ($\lambda$) algorithm of \cite{sut88}.

The focus is on fixed point equations associated with an uncontrolled Markov chain, denoted $\bfmX = \{ X_n : n=0,1,\dots\}$, on a measurable state space $(\state,\bx)$.  It is assumed to be $\psi$-irreducible and aperiodic  \cite{MT}.  In \Section{s:Q} we specialize to a finite state space.  

In control applications and analysis of learning algorithms, it is necessary to construct a Markov chain $\bfPhi$, of which $\bfmX$ is a component.    Other components may be an input process, or a sequence of ``eligibility vectors'' that arise in TD-learning.   It will be assumed throughout that there is a unique stationary realization of $\bfPhi$, with unique marginal distribution denoted $\pie$.

\archive{  $\bfmX$ that is a finite state space Markov chain; uni-chain and aperiodic, etc.   However ... 
 $\bfPhi$ does not have a finite state space for TD learning.   
 Note that $\bfPhi$ isn't even $V$-uniformly ergodic for TD learning!   }

\subsection{Motivation from SA \&\ ODE fundamentals}

The goal of stochastic approximation is to compute the solution   $\barf(\theta^*)=0$ for a function $\barf\colon\Re^d\to\Re^d$.    If the function is easily evaluated, then successive approximation can be used, and under stronger conditions the Newton-Raphson algorithm:
\begin{equation}
	\theta_{n+1} = \theta_n   +  G_n \barf(\theta_n)\,,\qquad G_n^{-1} = - \partial_\theta \barf\, (\theta_n)\, .
\label{e:NR}
\end{equation}
Under general conditions the convergence rate of \eqref{e:NR} is quadratic (much faster than geometric), which is not generally true of successive approximation.    

Stochastic approximation is itself an approximation of successive approximation.  It is assumed that 
$\barf(\theta) = \Expect[f(\theta,\Phi)]$,  where  $f\colon\Re^d\times\Re^m \to \Re^d$ and $\Phi$ is a random variable  with distribution $\pie$. 
The standard stochastic approximation algorithm  is defined by
\begin{equation} 
\theta_{n+1} = \theta_n + \alpha_n f(\theta_n,\Phi_{n+1} )\,,\quad n\ge 0\, .
\label{e:SAa}
\end{equation} 
For simplicity it is assumed that $\bfPhi$ is the stationary realization of the Markov chain. 
It is always assumed that the scalar gain sequence $\{\alpha_n\}$ is   non-negative, and satisfies:
\begin{equation}
\sum \alpha_n=\infty,\qquad\sum \alpha_n^2<\infty\, .
\label{e:stepsizeConditions}
\end{equation}

While convergent under general conditions,  the rate of convergence of \eqref{e:SAa} can often be improved dramatically through the introduction of a matrix gain.   This is explained first in a simple linear setting.

\subsection{Optimal covariance for linear stochastic approximation}

In many applications of reinforcement learning we arrive at a linear recursion of the form
\begin{equation}
\theta_{n+1} = \theta_n +  \alpha_{n+1} \bigl[A_{n+1}\theta_n - b_{n+1} \bigr]
\label{e:linearSA}
\end{equation} 
where $A_{n+1} = A(\Phi_{n+1})$ is a $d \times d$ matrix and  $b_{n+1} = b(\Phi_{n+1})$ is a $d \times 1$ vector,  $n\ge 0$.   
Let $A,b$ denote the respective steady-state means:  
\begin{equation}
A = \Expect[A(\Phi)]\,,\qquad 
b = \Expect[b(\Phi)]\,.
\label{e:Ab}
\end{equation} 
It is assumed throughout this section that $A$ is Hurwitz:   the real part of each eigenvalue is   negative.    Under this assumption, and subject to mild conditions on  $\bfPhi$, it is known that $\{ \theta_n\}$ converges with probability one to $\theta^*= A^{-1}b$  \cite{benmetpri90,kusyin97,bor08a}.  

Convergence of the recursion \eqref{e:linearSA} will be assumed henceforth.   
It is also assumed  that the gain sequence is given by $\alpha_n=1/n$, $n\ge 1$.

Under general conditions, the asymptotic covariance $\Sigma_\theta$ defined in  \eqref{e:SAsigma} is the non-negative semi-definite solution to the Lyapunov equation:
 \begin{equation}
(A+\half I)  \Sigma_\theta + \Sigma_\theta(A+\half I) ^\transpose +
\Sigma_\Delta  =0 .
\label{e:Lyap}
\end{equation} 
A  solution is  guaranteed only if  each eigenvalue of $A$ has   real part that is strictly less than $-1/2$.  If there exists an eigenvalue which does not satisfy this property, then under general conditions the asymptotic covariance is infinity  (see \Theorem{t:aCov}).
Hence the Hurwitz assumption must be strengthened to ensure that the asymptotic covariance is finite. 

\archive{Why $k$ and why not $n$? We still need to resolve the $k$, $n$ and $t$ issue.
	\\
	I'd love to keep $k$ and $t$ for time in a stochastic process.   Here $k$ is an index in something related to a process.   $n$ is good for iteration number in an algorithm, but 
	then we have also used $n$ for the time-horizon.   You come up with a plan and we shall follow it.
	\\
	\rd{AD: I like your idea; We'll change all the places where we have used $n$ for time-horizon to something else.}}
The matrix  $\Sigma_\Delta$ is obtained as follows: based on \eqref{e:linearSA}, the error sequence $\{\tiltheta_n = \theta_n-\theta^*\}$ evolves according to a deterministic linear system driven by ``noise'':  
\[
\tiltheta_{n+1} = \tiltheta_n +  \frac{1}{n+1}  [A \tiltheta_n +  \Delta_{n+1}   ]
\]
in which $\bf{\Delta}$ is the sum of three terms: 
\begin{equation}
\Delta_{n+1} =  \tilA_{n+1} \theta^* - \tilb_{n+1} +  \tilA_{n+1} \tiltheta_n,
\label{e:Dn}
\end{equation}  
with $\tilA_{n+1}=A_{n+1} - A$ , $ \tilb_{n+1} =  b_{n+1} - b $.  
The third term vanishes with probability one. 
The ``noise covariance matrix" $\Sigma_\Delta$ has the following two equivalent forms:
\begin{equation}
\Sigma_\Delta 
= \lim_{T\to\infty} \frac{1}{T} \Expect\bigl[ S_T S_T^\transpose  \bigr]
=  \sum_{k=-\infty}^\infty R(k)
\label{e:SigmaDelta}
\end{equation}
in which $S_T = \sum_{n=1}^T  \Delta_n$,  and 
\[
R(k)=R(-k)^\transpose = \Expect[ ( \tilA_k \theta^*  -  \tilb_k )( \tilA_0 \theta^*  -  \tilb_0)^\transpose]\,, \quad k\ge 0
\] 
where the expectation is in steady-state.   It is assumed that the CLT holds for sample-averages of the noise sequence:
\begin{equation}
\frac{1}{\sqrt{N}} \sum_{n=1}^N \Delta_n\to \clN (0, \Sigma_\Delta)\,,\quad N\to\infty\,,
\label{e:CLT-Delta}
\end{equation}
where the limit is in distribution.    This is a mild requirement when $\bfPhi$ is Markovian \cite{MT}.

A finite asymptotic covariance can be guaranteed by  increasing the gain:     choose $\alpha_n = g/n$ in
 \eqref{e:linearSA},  with $g>0$ sufficiently large so that the eigenvalues of $gA$ satisfy the required bound. More generally, a matrix gain can be introduced:
\begin{equation}
\theta_{n+1} = \theta_n +  \frac{1}{n+1} G\bigl[A_{n+1}\theta_n - b_{n+1} \bigr]
\label{e:GAINlinearSA}
\end{equation} 
in which $G$ is a $d\times d$ matrix.  Provided the matrix $G A$ satisfies the eigenvalue bound,  the corresponding asymptotic covariance  $\Sigma^G_\theta $ is finite and solves a modified Lyapunov equation:  
 \begin{equation}
(GA+\half I)  \Sigma^G_\theta + \Sigma^G_\theta(GA+\half I) ^\transpose +
G\Sigma_\Delta G^{\transpose}  =0
\label{e:GAINLyap}
\end{equation} 

The choice $G^*=-A^{-1}$ is analogous to the gain used in the Newton-Raphson algorithm 	\eqref{e:NR}.  With this choice, the asymptotic covariance is finite and given by 
\begin{equation}
\Sigma^* \eqdef A^{-1} \Sigma_\Delta {A^{-1}}^{\transpose}.  
\label{e:SigmaStar}
\end{equation} 
\archive{SPM is too lazy now
\\
Page 331: 10.2.2. Please check if this is what you want at the end of this page. And for your question in the next comment, is it Theorem 2.1 on page 330? It directly does not state this, but the extension seems to be direct.}
It is a remarkable fact that this choice   is optimal in the strongest possible statistical sense: For any other gain $G$,  the two asymptotic covariance matrices satisfy
\[
 \Sigma^G_\theta \ge \Sigma^*
\]
That is,  the difference $ \Sigma^G_\theta - \Sigma^*$ is positive semi-definite \cite{benmetpri90,kusyin97,bor08a}.  

The following theorem summarizes the results on the asymptotic covariance for the matrix-gain recursion  \eqref{e:GAINlinearSA}. The proof is contained in \Section{s:AsymCov} of the Appendix.
\begin{theorem}
\label{t:aCov}
Suppose that the eigenvalues of $GA$ lie in the strict left half plane,  and that the noise sequence satisfies the CLT \eqref{e:CLT-Delta} with finite   covariance $\Sigma_\Delta$. 
Then, the stochastic approximation recursion defined in \eqref{e:GAINlinearSA} is convergent, and the following also hold:
\begin{romannum}
	\item  Suppose that  $(\lambda,v)$ is an eigenvalue-eigenvector pair satisfying 
\[
\text{	$GAv =\lambda v$,  $\Real(\lambda)\ge -1/2$, \quad and $v^\dagger G \Sigma_\Delta G^\transpose v >0$,}
\]
where $v^\dagger$ denotes the conjugate transpose of the vector $v$.
Then
\[ 
 \lim_{n\to\infty} n\, v^\dagger  \Expect[\tiltheta_n\tiltheta_n^\transpose] v
=
 \lim_{n\to\infty }  n 	\Expect[ |v^\transpose \tiltheta_n|^2] =\infty\, ,
\]
and consequently,  the asymptotic covariance $\Sigma_\theta^G$ is not finite.

\item If all the eigenvalues of $GA $ satisfy  $\Real(\lambda) < -1/2$, then the corresponding asymptotic covariance $\Sigma_\theta^G$ is finite, and can be obtained as the solution to the Lyapunov equation \eqref{e:GAINLyap}
\item For any matrix gain $G$ the asymptotic covariance admits the lower bound
\[
	\Sigma^G_\theta \ge
	\Sigma^* \eqdef  A^{-1} \Sigma_\Delta {A^{-1}}^{\transpose}
\]
	This lower bound is achieved using $G^*\eqdef -A^{-1}$.   
	\qed
\end{romannum}
\label{t:asymcov}
\end{theorem}

\Theorem{t:aCov} inspires improved algorithms in many settings.  The first, which is essentially known, e.g.\ \cite[p.~331]{rup85,kusyin97}, will be called stochastic Newton-Raphson (SNR).

\paragraph{Stochastic Newton-Raphson}
 
This algorithm is obtained by estimating the mean $A$ simultaneously with the estimation of $\theta^*$: recursively define
\begin{equation}
\begin{aligned}
\theta_{n+1} &= \theta_n -  \alpha_{n+1} \haA_{n+1}^{-1} \bigl[A_{n+1}\theta_n - b_{n+1} \bigr]
\\
\haA_{n+1} &= \haA_{n} + \alpha_{n+1}  \bigl[  A_{n+1} - \haA_{n+1}  \bigr]\,,\qquad 
\textstyle
 \alpha_{n+1}= \frac{1}{n+1}\,,
\end{aligned}
\label{e:SNRlinearSA}
\end{equation}
where $\theta_0$ and $\haA_1$ are initial conditions.

If the steady-state mean $A$ (defined in \eqref{e:Ab})
 is invertible, then $ \haA_{n} $ is invertible for all $n$ sufficiently large.

The sequence $ \{n\haA_{n} \theta_n : n\ge 0\}$ admits a simple recursive representation that implies the following alternative representation of the SNR parameter estimates:
\begin{proposition}
\label{t:SAMC}
Suppose  $\haA_n$ is invertible for each $n \geq 1$. Then, the sequence of estimates $\{\theta_n\}$ obtained using \eqref{e:SNRlinearSA}  are identical to  the direct estimates:  
\[
\theta_{n} = \haA_{n}^{-1} \hab_{n}\,,
\qquad \text{
where}
\quad 
\haA_{n} = \frac{1}{n} \sum_{i=1}^{n} A_{i}\,,\quad \hab_{n} = \frac{1}{n} \sum_{i=1}^{n} b_{i}\,,\quad n\ge 1\,.
\]	
\label{t:LSAeqAinvb}
\qed
\end{proposition}

Based on the proposition, it is obvious that the SNR algorithm is consistent whenever the  Law of Large Numbers holds for the sequence $\{A_n, b_n\}$.  
Under the assumptions of
\Theorem{t:aCov}, the resulting asymptotic covariance is identical to what would be obtained with the constant matrix gain $G^*=-A^{-1}$.  
\archive{need that free lunch theorem here.  And, is this statement clear enough?  }

\medskip

Algorithm design in this linear setting is simplified in part because $\barf$ is an affine function of $\theta$, so that    the gain $G_n$ appearing in the standard 
Newton-Raphson algorithm \eqref{e:NR} does not depend upon the parameter estimates $\{\theta_k\}$.  However,  an ODE analysis of the SNR algorithm suggests that even in this linear setting, the dynamics are very different from its deterministic counterpart:
\begin{equation}
\begin{aligned}
\ddt x_t &= - \clA_t^{-1} \bigl[A x_t - b \bigr]
\\
\ddt \clA_t  &= - \clA_t + A
\end{aligned}
\label{e:SNRlinearSA-ODE}
\end{equation}
While evidently $\clA_t $ converges to $A$ exponentially fast in the linear model,  with a poor initial condition we might expect poor transient behavior.  
 
In extending the SNR algorithm to a nonlinear stochastic approximation algorithm,
an ODE approximation of the form \eqref{e:SNRlinearSA-ODE} will be possible under general conditions, but the matrix  $A$ will depend on $\theta$.    In addition to poor transient behavior,  the coupled equations  may be difficult to analyze. 
And, just as in the linear model,  the continuous time system looks very different from the deterministic Newton-Raphson recursion \eqref{e:NR}.
 
   The next class of algorithms are designed so that the associated ODE more closely matches the deterministic recursion.

\subsection{Zap Stochastic Newton-Raphson}

This is a two time-scales algorithm with a higher step-size for the matrix recursion. In the linear setting of this section, it is defined by  the variant of \eqref{e:SNRlinearSA}:
\begin{equation}
\begin{aligned}
\theta_{n+1} &= \theta_n -  \alpha_{n+1} \haA_{n+1}^{-1} \bigl[A_{n+1}\theta_n - b_{n+1} \bigr]
\\
\haA_{n+1} &= \haA_n + \gamma_{n+1}  \bigl[  A_{n+1} - \haA_n  \bigr]
\end{aligned}
\label{e:SNR2linearSA}
\end{equation}
It is different from the original Stochastic Newton-Raphson algorithm because of the two time-scale construction:
The second step-size sequence $\{\gamma_{n+1}\}$ is non-negative, satisfies \eqref{e:stepsizeConditions},  and also
\begin{equation}
\lim_{n\to\infty}\frac{\alpha_n}{\gamma_n} = 0\, .
\label{e:gammaalpha}
\end{equation}
We again take $\alpha_{n} = 1/n$,   $n\ge 1$.

The asymptotic covariance is again optimal.   
The ODE associated with the sequence $\{ \theta_n \}$ is far simpler,  and exactly matches the usual Newton-Raphson dynamics:
\begin{equation} 
\ddt x_t = - x_t + A^{-1} b  
\label{e:SNR2linearSA-ODE}
\end{equation}
This simplicity is also revealed in application to Q-learning, in which $A$ depends on the parameter. 

 A key point to note here is that the Zap version of the SNR algorithm plays a significant role in analysis as well as in performance improvement of general non-linear function approximation problems. We briefly discuss these in the following.

\subsubsection{Zap SNR for non-linear stochastic approximation}

Consider a stochastic approximation algorithm of the form \eqref{e:SAa} with $\barf(\theta) = \Expect [f(\theta,\Phi)]$, a non-linear function of the parameter vector $\theta$. The ODE of the two algorithms: SNR and Zap-SNR look significantly different in this case;  it is found that this difference is reflected in the rate of convergence of the stochastic recursion  (as we will see in the case of Q-learning). 
\spm{2018:  note smoothing here}
The SNR algorithm is essentially the same as \eqref{e:SNRlinearSA}:  
\begin{equation}
\begin{aligned}
\theta_{n+1} &= \theta_n -  \alpha_{n+1} \haA_{n+1}^{-1} f(\theta_n,\phi_{n+1})
\\
\haA_{n+1} &= \haA_{n} + \alpha_{n+1}  \bigl[  \nabla f(\theta_{n},\phi_{n+1}) - \haA_{n+1}  \bigr]\,,\qquad 
\textstyle
\alpha_{n+1}= \frac{1}{n+1}\,.
\end{aligned}
\label{e:SNRnonlinearSA}
\end{equation}
Note that the function $\nabla f(\theta_n,\phi_{n+1})$ may or may not be readily accessible, and this is application specific. In the case of Q-learning with linear function approximation, though the function $f$ is iteslf non-linear in $\theta$, $\nabla f$ is readily computable.

The ODE for the pair of recursions \eqref{e:SNRnonlinearSA} once again will be similar to \eqref{e:SNRlinearSA-ODE}:
\begin{equation}
\begin{aligned}
\ddt x_t &= - \clA_t^{-1} \barf(\theta_t)
\\
\ddt \clA_t  &= - \nabla \barf(\theta_t) + A
\end{aligned}
\label{e:SNRnonlinearSA-ODE}
\end{equation}

The Zap-SNR algorithm is a generalization	 of \eqref{e:SNR2linearSA}:
\begin{equation}
\begin{aligned}
\theta_{n+1} &= \theta_n -  \alpha_{n+1} \haA_{n+1}^{-1} f(\theta_n,\phi_{n+1})
\\
\haA_{n+1} &= \haA_n + \gamma_{n+1}  \bigl[  \nabla f(\theta_n,\phi_{n+1}) - \haA_n  \bigr]
\end{aligned}
\label{e:SNR2nonlinearSA}
\end{equation}
where once again the step-size sequence $\{\gamma_n\}$ satisfies \eqref{e:stepsizeConditions}, and \eqref{e:gammaalpha}. Similar to \eqref{e:SNR2linearSA-ODE}, the ODE of this algorithm is identical to the deterministic Newton-Raphson dynamics:
\begin{equation} 
\ddt x_t = - (\nabla \barf(x_t))^{-1} \barf(x_t).
\label{e:SNR2nonlinearSA-ODE}
\end{equation}

The general convergence and stability analysis of both \eqref{e:SNRnonlinearSA} and \eqref{e:SNR2nonlinearSA} is open. In \Section{s:Q} we show that when applied to Q-learning, the algorithms do converge under certain technical conditions.  However, the assumptions under which the single time-scale algorithm \eqref{e:SNRnonlinearSA} converges is far more restrictive than the assumptions under which the the two-time-scale algorithm \eqref{e:SNR2nonlinearSA} converges.

\subsubsection{Dealing with complexity: An $O$($d$) Zap-SNR algorithm}

\label{s:OdZapSNR}

It is common to discard the idea of second order methods because of their computational complexity. Before we move on to the specific applications in Reinforcement Learning, we propose an enhancement of the SNR algorithms that will result in complexity that is comparable to first order methods. 

We believe that we have convinced the readers that the two-timescale Zap-SNR algorithm \eqref{e:SNR2nonlinearSA} is of more interest to us (we will make this more precise in \Section{s:Q}), and hence restrict to extensions of this algorithm here.

It is assumed that there is no complexity in ``calculating" the gradient function $\nabla f(\cdot,\cdot)$, and that it is readily available. This is not be true in all applications, but holds in the applications of interest in this paper.  
Under these assumptions, computational complexity arises from the operations that are performed in manipulating these quantities.

The per-iteration complexity of the first order algorithm \eqref{e:SAintro} is $O(d)$, since $\theta \in \Re^d$. If the algorithm is run for $T$ iterations (assuming we have a data sequence of length $T$), the total complexity is $O(dT)$. The per iteration complexity in the case of the Zap-SNR algorithm \eqref{e:SNR2nonlinearSA} is $O(d^2)$, because it involves the product of a matrix inverse (of dimension $d \times d$) and a vector (of dimension $d \times 1$). The total complexity of the algorithm after running for $T$ iterations is $O(Td^2)$. 
\spm{2018: we will need a reference some day -- we can upload without it}

The essential idea behind the $O(d)$ Zap-SNR algorithm is to perform the $O(d^2)$ complexity steps only once every $N \geq d$ iterations, so that the total computational complexity   for a data sequence of length $T$ is $O(\frac{Td^2}{N})$; essentially resulting in the complexity of the first order method if $N = d$.
 This is done by ``batching" the data sequence into mini-sequences of length $N$, and applying recursions \eqref{e:SNR2nonlinearSA} for each batch as follows: For $i \geq 0$
\begin{equation}
\begin{aligned}
\theta_{(i+1)N} &= \theta_{iN} -  \alpha_{i+1} \haA_{(i+1)N}^{-1} \haf(\theta_{iN})
\\
\haA_{(i+1)N} &= \haA_{iN} + \hagamma_{i+1}  \bigl[  \nabla \haf(\theta_{iN}) - \haA_{iN}  \bigr],
\end{aligned}
\label{e:OdSNR2nonlinearSA}
\end{equation}
where,
\begin{equation}
\begin{aligned}
\haf(\theta_{iN}) &= N^{-1} \displaystyle \sum_{j = iN + 1}^{(i+1)N} f(\theta_{iN},\phi_j)
\\
\nabla \haf(\theta_{iN}) &= N^{-1} \displaystyle \sum_{j = iN + 1}^{(i+1)N} \nabla f(\theta_{iN},\phi_j)
\\
\hagamma_{i+1} &= 1 - \displaystyle \prod_{j=iN+1}^{(i+1)N} (1-\gamma_j).
\end{aligned}
\label{e:OdSNR2Defs}
\end{equation}

The first two definitions in \eqref{e:OdSNR2Defs} are straightforward; the expression for $\hagamma_{i+1,N}$ is obtained in such a way that the recursions in \eqref{e:OdSNR2nonlinearSA} very closely resemble the recursions in \eqref{e:SNR2nonlinearSA}\footnote{This deserves more explanation and we plan to provide one in a future version of the paper.}.

A remarkable (but almost obvious) property of the $O(d)$ Zap-SNR algorithm \eqref{e:OdSNR2nonlinearSA} is that it has the same asymptotic properties (specifically, the asymptotic covariance) as that of the original Zap-SNR algorithm \eqref{e:SNR2nonlinearSA}. This once again is made more precise in a future version of the paper. The specific application of this algorithm to Q-learning is discussed in \Section{s:OdZap}.


\subsection{Application to temporal-difference algorithms}

The general theory is illustrated here, through application to  TD($\lambda$)-learning algorithms.

Let $\{P^n\}$ denote the transition semigroup for the Markov chain $\bfmX$:
For each $n\ge 0$, $x\in\state$, and  $A\in \clB(\state)$,
\[
P^n(x,A):=\Prob_x\{X_n\in A\}:=\Pr\{X_n\in A\,|\,X_0=x\}.
\]
The standard operator-theoretic notation is used for conditional expectation:  
for any measurable function $f \colon \state \to \Re$,  
\[
P^n f\, (x) = 
\Expect_x [f(X_n)  ]  \eqdef
\Expect [f(X_n) \mid X_0 = x].
\]
In a finite state space setting, $P^n$ is the $n$-step transition probability matrix of the Markov chain, and the conditional expectation appears as matrix-vector multiplication:
\[
P^n f\, (x) = \sum_{x'\in\state} P^n(x,x') f(x'),\qquad x\in\state.
\]

Let $c\colon\state\to\Re_+$ denote a cost function, and $\beta\in (0,1)$ a discount factor.
The discounted-cost value function is defined as $h=\sum_{n=0}^\infty \beta^n P^n c$, which is the unique solution to the Bellman equation
\begin{equation}
	c+ \beta Ph= h
\label{e:DiscFish}
\end{equation}
TD-learning algorithms are designed to obtain approximations of $h$ within a finite-dimensional parameterized class.

Consider the case of a $d$-dimensional linear parameterization.   A   function $\psi \colon\state\to\Re^d$ is chosen,  which is viewed as a collection of $d$ basis functions.   Each vector $\theta\in\Re^d$ is associated with the approximate value function ${{{h}}}^\theta = \sum_i \theta_i \psi_i$.   There are two standard criteria for defining optimality of the parameter.  Most natural is the minimum norm approach:  
\begin{equation} 
 \theta^*= \argmin_\theta \|  {{{h}}}^\theta - h\|
\label{e:MetricMinNorm}
\end{equation}  
in which the choice of norm is part of the design of the algorithm.   Most common is   \begin{equation}
\|  {{{h}}}^\theta - h\|^2 = \Expect[ ({{{h}}}^\theta(X_n)-h(X_n))^2]
\label{e:pinorm}	
\end{equation} 
where the expectation is in steady-state.  
 
In the Galerkin approach, a $d$-dimensional stationary stochastic process $\bfelig$ is constructed that is adapted to a stationary realization of $\bfmX$.  An algorithm is designed to obtain the vector   $\theta^*\in\Re^d$ that satisfies
\begin{equation} 
 0 = \Expect\bigl[  \bigl( -{{{h}}}^{\theta^*} (X_n) + c(X_n)+ \beta {{{h}}}^{\theta^*} (X_{n+1}) \bigr)   \elig_n(i) \bigr] \, , \quad 1\le i\le d\, ,
\label{e:MetricGalTD} 
\end{equation}
in which the expectation is again in steady state.
The $d$-dimensional stochastic process $\bfelig$ is called the sequence of \textit{eligibility vectors}.

The motivation for the first criterion \eqref{e:MetricMinNorm} is clear, but algorithms that solve this problem often suffer from high variance. The Galerkin approach is used because it is simple and generally applicable.   Also,  if the basis functions are chosen such that $h={{{h}}}^{\theta^\bullet}$ for some $\theta^\bullet\in\Re^d$, and if the solution  to  \eqref{e:MetricGalTD} is unique,  then the Galerkin approach will yield the exact solution $h$.


%

The goal of the TD($\lambda$) learning algorithm is to solve the Galerkin relaxation  \eqref{e:MetricGalTD} in which the eligibility vectors are obtained by passing $\{\psi(X_n)\}$ through the corresponding first-order low-pass filter:  $\elig_{n+1}   = \lambda \beta \elig_n + \psi(X_{n+1})$, $n\ge 0$.  It is always assumed that $\lambda\in [0,1]$.    It is shown in \cite{tsiroy97a} that the solutions to the Galerkin fixed point equation  \eqref{e:MetricGalTD} and the minimum norm problem \eqref{e:MetricMinNorm} coincide if $\lambda=1$, with the norm defined by \eqref{e:pinorm}.

\paragraph{TD($\lambda$) algorithm:}
For initialization $\theta_0\,,\elig_0\in\Re^d  $, the sequence of estimates are defined recursively:  
\begin{equation}
\begin{aligned}
\theta_{n+1} & = \theta_n + \alpha_{n+1} \elig_n  d_{n+1} 
\\
d_{n+1} & =  c(X_n) +  \bigl[\beta \psi(X_{n+1}) -  \psi(X_n) \bigr]^\transpose  \theta_n
\\
\elig_{n+1} & = \lambda \beta \elig_n + \psi(X_{n+1})\, .
\label{e:TDlambda}
\end{aligned}
\end{equation}

The recursion  \eqref{e:TDlambda} can be placed in the form \eqref{e:linearSA} in which $\Phi_n= (X_n, X_{n-1},\elig_{n-1} )$, and
\begin{equation}
A_{n+1} =  \elig_n
 \bigl[\beta \psi(X_{n+1}) -  \psi(X_n) \bigr]^\transpose \,,\qquad b_{n+1} =   - \elig_n c(X_n)
\label{e:TDAb}
\end{equation}
Based on this representation, it can be shown that the TD($\lambda$) algorithm is consistent provided the basis vectors are linearly independent, in the sense that $\Expect_\pie[ \psi(X_n) \psi(X_n)^\transpose]>0$.  

It is also easy to construct an example for which the  asymptotic covariance is infinite:   Take any consistent example, and scale the basis vectors by a small constant $\epsy$.   Using the basis $\epsy \psi$, the resulting matrix $A$ is scaled by $\epsy^2$.  Hence, for sufficiently small $\epsy>0$,  \textit{each eigenvalue of $A$ will have   real part that is strictly  greater than $-1/2$.}

 An application of the SNR matrix gain algorithm \eqref{e:SNRlinearSA} results in an algorithm with optimal asymptotic covariance.   This results in the coupled recursions:
\begin{eqnarray}
\begin{aligned}
\theta_{n+1} & = \theta_n - \alpha_{n+1} \haA_{n+1}^{-1} \elig_n  d_{n+1} 
\\
d_{n+1} & =  c(X_n) +  \bigl[\beta \psi(X_{n+1}) -  \psi(X_n) \bigr]^\transpose  \theta_n
\\
\elig_{n+1} & = \lambda \beta \elig_n + \psi(X_{n+1})
\label{e:TDKlambda}
\end{aligned}
\\
\haA_{n+1} = \haA_n +   \alpha_{n+1}\bigl[\elig_n (\beta\psi(X_{n+1}) -  \psi(X_n) )^\transpose   -   \haA_n  \bigr] \,,
\label{e:TDK_A}
\end{eqnarray}
where $ \alpha_n \equiv 1/n$, for $n \geq 1$.

The following  proposition follows directly from  \Proposition{t:LSAeqAinvb}:
\begin{proposition}
\label{t:LSTDSNR}   
Suppose that $\haA_{n}$ is invertible for all $n \geq 1$. Then,  the sequence of parameters obtained using the SNR-TD($\lambda$) algorithm  (\ref{e:TDKlambda},\ref{e:TDK_A}) coincides with the direct estimates: 
\begin{equation}
\begin{aligned}
\theta_{n} &= \haA_{n}^{-1} \hab_{n}, 
\\
\haA_{n} = \frac{1}{n} \sum_{i=1}^{n} \elig_{i-1} \bigl[\beta \psi(X_{i}) -  \psi(X_{i-1}) \bigr]^\transpose 
& + \frac{1}{n}\clE_1  \,,\qquad \hab_{n} = \frac{1}{n} \sum_{i=1}^{n} \elig_{i-1} c(X_{i-1})\,,
\end{aligned}
\label{e:SNRTD}
\end{equation}
where $\clE_1 = \haA_{\text{IC}} - A_1 $, $\haA_{\text{IC}}$ denoting the matrix $\haA_1$ in (\ref{e:TDKlambda},\ref{e:TDK_A}), 
and the sequence of vectors $\{\elig_{n}\}$ are again defined by  
$\elig_{n+1}  = \lambda \beta \elig_n + \psi(X_{n+1})$.  \qed  
\end{proposition}
\archive{Readers will be confused by ``initial condition''  for hat-A.   We don't need to explain this issue -- no big deal.
\\
\rd{Is the connection between (\ref{e:TDKlambda},\ref{e:TDK_A}) and \eqref{e:SNRTD} clear?? Is it understood that the initial conditions are different now that we have removed some text?}}

It is a remarkable fact that this algorithm is essentially equivalent to the LSTD($\lambda$) algorithm of  \cite{brabar96,boy02,nedber03a}:
The LSTD($\lambda$) algorithm is defined to be \eqref{e:SNRTD} with   $\clE_1 =0$.

\archive{We still need to do this!
Essay:	
While the motivation for the algorithm relied on a complete basis, so that ${{{h}}}^{\theta^*} $ is a solution to \eqref{e:DiscFish} for a unique value of $\theta^*$,   we demonstrate that this}

\archive{Some day let's revisit "Why does Polyak fail?" in the Jan 2017 version of the paper.   This is important for our education, not for this paper.}

\section{Q-Learning}
\label{s:Q}

The class of algorithms considered next is designed for a controlled Markov model, whose input process is denoted $\bfmU$.  It is assumed that the state space $\state$ and the action space   $\U$ on which $\bfmU$ evolves are both finite.  Denote $\nd = |\state|$ and  $\ninp = |\U|$.

\subsection{Notation and assumptions}

It is convenient to maintain the operator-theoretic notation used in the uncontrolled setting.    
There is now a controlled transition matrix  that acts on functions $h\colon\state\to\Re$ via
\[
P_u h\, (x) \eqdef \sum_{x'\in\state} P_u(x,x') h(x'),\qquad x\in\state\,,\ u\in\U\, .
\]
For any non-anticipative input sequence $\bfmU$ we have $P_u h\, (x) = \Expect[h(X_{t+1}) \mid X_0^t, U_0^t]$ on the event $X_t=x$ and $U_t=u$.

There is a finite number of deterministic stationary policies that are enumerated as 
$\{ \phi^{(i)} : 1\le i\le \nphi\}$, with $\nphi = (\ninp)^\nd$.   A randomized stationary policy is defined by a pmf $\upmf$ on the integers $\{ 1\le i\le \nphi\}$ and such that for each $t$,  
\begin{equation}
 U_t =  \sum_{k=1}^{\nphi}  \iota_k(t)  \phi^{(k)}(X_t) 
\label{e:RandStatPolicy}
\end{equation}
where $\{ \iota(t) \}$ is an i.i.d.\ sequence on $\{0,1\}^{\nphi}$ satisfying $\sum_k \iota_k(t) = 1$,   and $\Prob\{\iota_k(t) = 1 \mid X_0^t\} = \upmf(k)$ for all $k$ and $t$.

For any deterministic stationary policy $\phi$,  let $S_\phi$ denote the substitution operator, defined for any function $q\colon \state\times\U\to\Re$ by $S_\phi q\, (x) = q(x,\phi(x))$. If the policy $\phi$ is randomized, of the form \eqref{e:RandStatPolicy},   then we denote
\[
S_\phi q\, (x) =  \sum_k \upmf(k) q(x,\phi^{(k)} (x)) 
\]  
With  $P$ viewed as  a single matrix with $\nd\cdot\ninp$ rows and $\nd$ columns, and $S_\phi$ viewed as a matrix with $\nd$ rows and $\nd\cdot\ninp$ columns,  the following interpretations hold:
\begin{lemma}
\label{t:PSphi}
Suppose that $\bfmU$ is defined using a   stationary policy $\phi$ (possibly randomized).  Then, both $\bfmX$ and the pair process $(\bfmX,\bfmU)$ are Markovian, and
\begin{romannum}
\item
$P_\phi \eqdef S_\phi P$ is the transition matrix for   $\bfmX$.
 
\item  $P S_\phi$ is the transition matrix for   $(\bfmX,\bfmU)$.
\qed
\end{romannum}
\end{lemma}

A cost function $c\colon\state\times\U\to\Re$ is given together with a discount factor $\beta\in (0,1)$.  For any (possibly randomized) stationary policy $\phi$, the resulting value function is denoted  
\begin{equation}
{{{h}}}_\phi(x) = \sum \beta^t P_\phi^t S_\phi c = \sum \beta^t \Expect[c(X_t,U_t) \mid X_0=x]\,,\qquad x\in\state\, .
\label{e:hphi}
\end{equation}
The minimal value function is denoted $h^*$, which is the unique solution to the discounted-cost optimality equation (DCOE):
\begin{equation}
h^*(x) = \min_u \bigg\{ c(x,u) +  \beta \sum_{x'\in\state}  P_u(x,x')  h^*(x')     \bigg\}    
\label{e:DCOE}
\end{equation}
The minimizer defines a stationary policy $\phi^*\colon\state\to\U$ that is optimal over all input sequences \cite{ber12a}.

The associated ``Q-function'' is defined to be the term within the brackets,  $Q^*(x,u) \eqdef  c(x,u) + \beta P_u h^*\, (x)$.
The DCOE implies a similar fixed point equation for the Q-function:
\begin{equation}
	Q^*(x,u) =  c(x,u) + \beta \sum_{x'\in\state}  P_u(x,x')  \uQ^*(x')   
\label{e:DCOE-Q}
\end{equation}
in which $\uQ(x)\eqdef\min_u Q(x,u)$ for any function $Q\colon\state\times\U\to\Re$.  

For any function $q\colon\state\times\U\to\Re$,   let    $\phi^q\colon\state\to\U$  denote  an associated policy satisfying
\begin{equation}
\phi^q(x) \in  \argmin_uq(x,u)
\label{e:phi_q}
\end{equation}
for each $x\in\state$. It is assumed to be specified \emph{uniquely} as follows: 
\begin{equation}
\begin{aligned}
\phi^q & \eqdef \phi^{(\kappa)}\,\,\text{such that}\,\, \kappa = \min \{ i : \phi^{(i)} (x) \in \argmin_uq(x,u), \,\, \text{for all }x \in \state  \}
\end{aligned}
\label{e:phi_q_def}
\end{equation}
The fixed point equation \eqref{e:DCOE-Q} becomes
\begin{equation}
Q^* =  c + \beta  P S_{\phi} Q^*\,,\qquad \text{with $\phi =  \phi^q $,   $q={Q^*}$}
\label{e:DCOE-Qb}
\end{equation}

In the analysis that follows it is necessary to consider the Q-function associated with all  possible cost functions simultaneously:  given any function $\varsigma\colon\state\times\U\to\Re$,  let $\Qstar(\varsigma)$ denote the corresponding solution to the fixed point equation \eqref{e:DCOE-Q}, with $c$ replaced by $\varsigma$.  That is,    the function $q=\Qstar(\varsigma)$ is the solution to the fixed point equation,
\archive{$\Qstar^{-1}$ is obviously piecewise linear and convex, but not monotone}
\begin{equation} 
q(x,u) = \varsigma(x,u)   + \beta \sum_{x'} P_u(x,x') \min_{u'} q(x',u')\,, \quad x\in\state,\ u\in\U.
\label{e:DCOE-Qkappa}
\end{equation}
For a pmf $\upmf$ defined on the set of policy indices $\{ 1\le i\le \nphi\}$, denote
\begin{equation}
\partialQstar_{\upmf} \eqdef \bigl( \sum\upmf(i)  [I-\beta P S_{\phi^{(i)}}]\bigr)^{-1}
\label{e:partiakQmu}
\end{equation}
so that $ \partialQstar_{\upmf} \varsigma$ is the ``$Q$-function'' obtained with the cost function $\varsigma$, and the randomized stationary policy defined by $\mu$ (see also discussion of the SARSA algorithm following the proof of \Lemma{t:Thetak}).  It follows that the functional $\Qstar$ can be expressed as the minimum over all pmfs $\upmf$:
\begin{equation}
\Qstar(\varsigma) = \min \partialQstar_{\upmf} \varsigma
\label{e:Qstar}
\end{equation}
There is a single degenerate pmf that attains the minimum for each $(x,u)$  (the optimal stationary policy is deterministic) \cite{ber12a}.  
\archive{spm: added some explanation.  In your dissertation you will need a chapter reviewing SARSA and how it relates to TD learning.  }

\begin{lemma}
\label{t:clQstar} 
The mapping $\Qstar$ is a bijection on the set of real-valued functions on $\state\times\U$.   It is also piecewise linear, concave and monotone.    
\end{lemma}

\begin{proof}  
The fixed point equation \eqref{e:DCOE-Qkappa}  defines the Q-function   with respect to the cost function $\varsigma$.   Concavity and monotonicity hold because $q=\Qstar(\varsigma)$ as defined in \eqref{e:Qstar} is the minimum of linear, monotone functions.  The existence of an inverse $q\mapsto \varsigma$ follows from \eqref{e:DCOE-Qkappa}.
\end{proof}

A Galerkin approach to approximating $Q^*$ is   formulated as follows:  Consider a linear parameterization $Q^\theta(x,u)=   \theta^\transpose \psi (x,u) $, with $\theta\in\Re^d$ and $\psi\colon\state\times\U\to\Re^d$,  and denote  $\uQ^\theta(x) = \min_u Q^\theta (x,u)$.    Obtain a 
$d$-dimensional stationary stochastic process $\bfelig$ that is adapted to $(\bfmX,\bfmU)$, and define
 $\theta^*$ to be a solution to
\begin{equation}
\Expect\bigl[ \bigl\{  c(X_n,U_n)   + \beta     \uQ^{\theta^*} (X_{n+1})  - Q^{\theta^*}(X_n,U_n)  \bigr\}   \elig_n(i)\bigr] =0, \quad 1\le i\le d
\label{e:eligQalg}
\end{equation}
where the expectation is in steady-state.
 
Similar to TD($\lambda$)-learning,  a possible approach to estimate $\theta^*$ is the following: 
\paragraph{Q($\lambda$) algorithm:}
For initialization $\theta_0\,,\elig_0\in\Re^d  $, the sequence of estimates are defined recursively:  
\begin{equation}
\begin{aligned}
\theta_{n+1} & = \theta_n + \alpha_{n+1} \elig_n  \td_{n+1} 
\\
\td_{n+1} & =   c(X_n,U_n)   + \beta     \uQ^{\theta_n} (X_{n+1})  - Q^{\theta_n}(X_n,U_n)  
\\
\elig_{n+1} & = \lambda \beta \elig_n +\psi(X_{n+1}, U_{n+1})\, .
\label{e:Qlambda}
\end{aligned}
\end{equation}
The success of this approach has been demonstrated in a few restricted settings, such as optimal stopping problems \cite{yuber13},   deterministic models \cite{mehmey09a},   and variations of Watkins algorithm that are discussed next.

\subsection{Watkins algorithm}

The basic Q-learning algorithm of \cite{watday92a,wat89} is a particular instance of the Galerkin approach with $\lambda=0$ in  \eqref{e:Qlambda}.   The  basis functions are taken to be indicator functions:  
\begin{equation}
\psi_k(x,u) = \ind\{x=x^k, u=u^k\}\,,\quad 1\le k\le d\,,  
\label{e:WatkinsBasis}
\end{equation}
where $\{(x^k,u^k) : 1\le k\le d\}$ is an enumeration of all state-input pairs.   The goal of this approach is to compute the function $Q^*$ exactly.

The parameter $\theta$ is identified with the estimate   $Q^\theta$, and hence $\theta\in\Re^d$ with $d=\nd \cdot \ninp$.   The basic stochastic approximation algorithm to solve \eqref{e:eligQalg} coincides with Watkins algorithm: 
\begin{equation}
\theta_{n+1} = \theta_n + \alpha_{n+1} \bigl\{  c(X_n,U_n)   + \beta     \utheta_n (X_{n+1})  - \theta_n(X_n,U_n)  \bigr\}   \psi(X_n,U_n)
\label{e:Watkin}
\end{equation}
Only one entry of the approximation is updated at each time point, corresponding to the previous state-input pair $(X_n,U_n)$ observed.

   
\assume{Assumption Q1:}  The input is defined by a randomized stationary policy of the form 
\eqref{e:RandStatPolicy}.    The  joint process $(\bfmX,\bfmU)$ is an irreducible Markov chain. That is, it has a unique invariant pmf $\pie$ satisfying  $\pie(x,u)>0$ for each $x,u$.
\qed

\archive{In  \cite{bormey00a} Vivek  references one of his papers for stability of the ODE!  There must be a better reference.  We'll forget about this for now} 
 
\assume{Assumption Q2:}  The optimal policy $\phi^*$ is unique.    \qed

The  ODE for stability analysis takes on the following simple form:
\begin{equation}
\ddt q_t(x,u) = \pie(x,u)  \Bigl\{  c(x,u)   + \beta   P_u\uq_t\, (x)   - q_t(x,u)  \Bigr\}   
\label{e:QODEW}
\end{equation}
in which $\uq_t(x)=\min_u q_t(x,u)$ as defined below \eqref{e:DCOE-Q}.   This ODE is stable under Assumption~Q1, which then implies that the parameter estimates converge to $Q^*$ a.s. \cite{bormey00a}.

Under Assumption Q2 there exists $\epsy>0$ such that  
\[
\phi^*(x) = \argmin_{u\in\U} \theta(x,u),\qquad x\in\state,\ \theta\in\Re^d,\ \|\theta-\theta^*\|<\epsy\, .
\]
This justifies a linearization of the ODE \eqref{e:QODEW}, in which $\uq_t$ is replaced by $S_{\phi^*} q_t$.  

Although the algorithm is consistent, it should   be clear that the asymptotic covariance of this algorithm is typically infinite. 

\begin{theorem}
\label{t:Qinfinite}
Suppose that Assumptions Q1 and Q2 hold. Then, the sequence of parameters $\{ \theta_n \}$ obtained using the Q-learning algorithm \eqref{e:Watkin} converges to $Q^*$ a.s..
Suppose moreover that the conditional variance of $h^* (X_t)$ is positive:
\begin{equation}
\sum_{x,x',u} \pie(x,u) P_u(x,x') [ h^*(x') -  P_u h^*\, (x) ]^2 >0
\label{e:hvar}
\end{equation} 
and $(1-\beta) \max_{x,u} \pie(x,u)  \le \half$.  Then, in the case $\alpha_n \equiv 1/n$,  
\[
\lim_{n\to\infty}  n \Expect[ \| \theta_{n} - \theta^*\|^2 ] = \infty\, .
\] 
\qed
\end{theorem}
The assumption $(1-\beta) \max_{x,u} \pie(x,u)  \le\half$ is satisfied whenever $\beta\ge \half$.

The proof of convergence can be found in \cite{watday92a,wat89}.
The proof of infinite asymptotic covariance is given in \Section{s:InfAsymVar}	 	
of the Appendix. An eigenvector for $A$ is constructed with strictly positive entries, and with real eigenvalue satisfying $\lambda \ge -1/2$. 
Interpreted as a function $v\colon\state\times\U\to\Co$,  this eigenvector satisfies
\begin{equation}
v^\dagger \Sigma_\Delta v = \beta^2 \sum_{x,x',u} \pie(x,u) |v(x,u)|^2 P_u(x,x') [  h^*(x') -  P_u h^*\, (x) ]^2 .
\label{e:SigmaDeltaQ}
\end{equation}
Assumption \eqref{e:hvar} ensures that the right hand side  is strictly positive, as required in  \Theorem{t:aCov}~(i).    


The recursion  \eqref{e:Watkin} for the Q-learning algorithm can be written in the form \eqref{e:linearSA} in which
\[ 
A_{n+1}  =        \psi(X_n,U_n) \{  \beta    \psi (X_{n+1}, \phi_{n}(X_{n+1} ) ) - \psi(X_n,U_n)  \bigr\}^\transpose  
\,,
\qquad
b_{n+1} = c(X_n,U_n)  \psi(X_n,U_n)\,.
\]
This motivates the introduction of stochastic Newton-Raphson algorithms that are considered next.

\subsection{SNR and Zap Q-Learning}
  	

For a  sequence of $d \times d$ matrices $\bfmG=\{G_n\}$ and $\lambda\in [0,1]$, the matrix-gain Q($\lambda$) algorithm is described as follows:

\paragraph{$\bfmG$-Q($\lambda$) algorithm:}
For initialization $\theta_0\,,\elig_0\in\Re^d  $, the sequence of estimates are defined recursively:  
\begin{equation}
\begin{aligned}
\theta_{n+1} & = \theta_n + \alpha_{n+1} G_{n+1} \elig_n  \td_{n+1} 
\\
\td_{n+1} & =   c(X_n,U_n)   + \beta     \uQ^{\theta_n} (X_{n+1})  - Q^{\theta_n}(X_n,U_n)  
\\
\elig_{n+1} & = \lambda \beta \elig_n +\psi(X_{n+1}, U_{n+1})\,
\label{e:HQlambda}
\end{aligned}
\end{equation}

The special case based on stochastic Newton-Raphson \eqref{e:SNR2linearSA}  is called the \textbf{\textit{Zap-Q($\lambda$) algorithm:}}
\begin{algorithm}[H]
\caption{\small  Zap-Q($\lambda$) algorithm}
\label{ZapQ}
\begin{algorithmic}[1]
\INPUT Initial $\theta_0 \in \Re^d$, $\elig_{0} = \psi(X_0,U_0)$, $\haA_0 \in \Re^{d \times d}$, $n = 0$, $T \in \bfmZ^+$
\hfill  \Comment Initialization
\Repeat
\State $\phi_n^{X_{n+1}} \eqdef \argmin_u Q^{\theta_n}(X_{n+1},u)$;
\State $\td_{n+1} \eqdef   c(X_n,U_n)   + \beta     Q^{\theta_n} (X_{n+1},\phi_n^{X_{n+1}})  - Q^{\theta_n}(X_n,U_n)$;\hfill \Comment Temporal difference term
\State $A_{n+1}   \eqdef  \elig_n \bigl[  \beta    \psi (X_{n+1}, \phi_{n}^{X_{n+1} })  - \psi(X_n,U_n)  \bigr]^\transpose$;
\State $\haA_{n+1} = \haA_n + \gamma_{n+1}  \bigl[  A_{n+1} - \haA_n  \bigr]$;\hfill \Comment Matrix gain update rule
\State $ \theta_{n+1}  = \theta_n - \alpha_{n+1} \haA_{n+1}^{-1} \elig_n  \td_{n+1} $;\hfill \Comment Zap-Q update rule
\State $\elig_{n+1} \eqdef \lambda \beta \elig_n +\psi(X_{n+1}, U_{n+1})$;\hfill \Comment Eligibility vector update rule
\State $n = n+1$
\Until{$n \geq T$}
\end{algorithmic}
\end{algorithm}  

\archive{\todo Resolve $n$,$t$, and $k$ issue\\ see my prior responses}

It is assumed that a projection is employed to ensure that $\{\haA_n^{-1}\}$ is a bounded sequence --- this is most easily achieved using the Matrix Inversion Lemma.    

The analysis that follows is specialized to $\lambda=0$ and the basis \eqref{e:WatkinsBasis} that is used in Watkins' algorithm.   The resulting \textit{Zap-Q algorithm} is defined as follows, 
after identifying $Q^\theta$ and $\theta$:
\begin{eqnarray} 
 &&
\theta_{n+1} 
= \theta_n + \alpha_{n+1}\haG_{n+1}^* \bigl\{  c(X_n,U_n)   + \beta     \utheta_n (X_{n+1})  - \theta_n(X_n,U_n)  \bigr\}   \psi(X_n,U_n)
\label{e:QSNR2def}
\\[.2cm]  
&&
\begin{aligned}
\haA_{n+1} &= \haA_n + \gamma_{n+1}  \bigl[  A_{n+1} - \haA_n  \bigr]
\\
A_{n+1}   &=  \psi(X_n,U_n) \bigl[  \beta    \psi (X_{n+1}, \phi_{n}(X_{n+1} ))  - \psi(X_n,U_n)  \bigr]^\transpose   
\\
  \phi_n &= \phi^{\theta_n}    
\end{aligned}
\label{e:QSNR2Adef}
\end{eqnarray} 
where $\haG_n^* = - [ \haA_n]^{-1}$, and $[\varble]$ denotes a projection,   chosen so that $\{\haG_n^*\}$ is a bounded sequence. In \Theorem{t:ZAP} it is established that the projection is required only for a finite number of iterations:  $\{\haA_n^{-1} : n\ge n_\bullet\}$ is a bounded sequence, where $n_\bullet<\infty $ a.s..

An equivalent representation for the parameter recursion \eqref{e:QSNR2def} is 
\begin{equation} 
\theta_{n+1} 
= \theta_n +  \alpha_{n+1} \haG_n^* \bigl\{  \Psi_n c   +  A_{n+1} \theta_n \bigr\}   
\label{e:QSNR2def-b} 
\end{equation}
in which $c$ and $\theta_n$ are treated as $d$-dimensional vectors rather than functions on $\state\times\U$,  and  
\begin{equation}
\label{e:Psi}
\Psi_n  =  \psi(X_n,U_n)    \psi(X_n,U_n) ^{\transpose}.
\end{equation}

It would seem that the analysis is complicated by the fact  that the  sequence $\{A_n\}$   depends upon $\{\theta_n\}$ through the policy sequence $\{\phi_n\}$.  Part of the analysis is simplified by obtaining a recursion for the following $d$-dimensional sequence:   
\begin{equation}
\haC_n = -\diagpie^{-1} \haA_n \theta_n  \,, \quad n\ge 1\,,
\label{e:Cn}
\end{equation}
where   $\diagpie$ is the $d\times d$ diagonal matrix with entries $\diagpie(k,k) \eqdef \pie(x^k,u^k)$.   
This admits a very simple recursion in the special case $\bfgamma\equiv \bfalpha$.    In the other case considered, wherein the step-size sequence $\bfgamma$ satisfies \eqref{e:gammaalpha},  the recursion for $\bfmhC$ is more complex, but the ODE analysis is simplified.

\subsection{Main results}
\label{s:main}

Conditions for convergence of the Zap-Q algorithm (\ref{e:QSNR2def},\ref{e:QSNR2Adef}) are summarized in \Theorem{t:ZAP}.  The following assumption is used to address the discontinuity in the recursion for $\{ \haA_n\}$  resulting from the dependence of
 $A_{n+1}$  on $\phi_n$. 

\assume{Assumption Q3:} The sequence of policies $\{\phi_n\}$ satisfies:
\begin{equation}
\sum_{n=1}^\infty \gamma_n \ind\{\phi_{n+1}\neq \phi_n\}  <\infty
\,,\quad a.s..
\label{e:phiChangeNull}
\end{equation} 
\qed
 
\begin{theorem}
\label{t:ZAP}
Suppose that Assumptions Q1--Q3 hold,  with the gain sequences $\bfalpha$ and $\bfgamma$ satisfying
\begin{equation}
\alpha_n =\frac{1}{n}\,,\quad  \gamma_n =\frac{1}{n^\rho}\,,\qquad n\ge 1\, ,
\label{e:GAINS}
\end{equation} 	
for some fixed $\rho\in (\half,1)$. Then,
\begin{romannum}
\item 
The parameter sequence  $\{\theta_n\}$ obtained using the Zap-Q algorithm (\ref{e:QSNR2def},\ref{e:QSNR2Adef}) converges to $Q^*$ a.s.. 
		
\item The asymptotic covariance \eqref{e:SAsigma} is minimized over all
$\bfmG$-Q($0$) matrix gain versions of Watkins' Q-learning algorithm.

\item
An ODE approximation holds for the sequence $\{ \theta_n ,  \haC_n\}$, by continuous functions $(\bfmq,\bfqcost)$ satisfying
\begin{equation}
q_t = \Qstar(\qcost_t) \, , \quad  \ddt \qcost_t = - \qcost_t + c
\label{e:ZapODE}
\end{equation}
This ODE approximation is exponentially asymptotically stable, with $\displaystyle	\lim_{t\to\infty} q_t = Q^*
$.  	\qed
\end{romannum}
\end{theorem}
See \Section{s:ODE} and standard references such as \cite{bor08a} for the precise meaning of the ODE approximation \eqref{e:ZapODE}.

\begin{proof}[Proof of  \Theorem{t:ZAP}]
Boundedness of the sequences $\{\theta_n, \haA_n : n\ge 0\}$ and $\{\haA_n^{-1} : n\ge n_\bullet\}$ is established  in Lemmas~\ref{t:haA_Bdd_And_ODE_sol} and \ref{t:haq_Bdd_And_ODE_sol},  where $n_\bullet<\infty$ a.s..   The ODE approximation is established in \Proposition{t:ODE_Gamma}.  These two results combined with standard arguments establishes (i) \cite{bor08a}.   

Result (ii) follows from convergence of the algorithm,  just as in the case of TD-learning.  Uniqueness of the optimal policy is needed so that the recursion for $\{\theta_n\}$ admits a linearization around $Q^*$.   
\end{proof}

In the case $\bfgamma \equiv \bfalpha$, the three consequences hold under a stronger assumption than Q3: 
\begin{proposition}
\label{t:ZAP1}
Suppose that Assumptions Q1--Q2 hold, $\bfgamma \equiv \bfalpha$, and the sequence of policies $\{\phi_n\}$ is convergent. Then, 
the parameter sequence  $\{\theta_n\}$ obtained using the Zap-Q algorithm (\ref{e:QSNR2def},\ref{e:QSNR2Adef}) converges to $Q^*$ a.s..
\end{proposition}
\qed

The convergence assumption in \Proposition{t:ZAP1} is far stronger than Q3:
Recall that the policies $\{\phi_n\}$ evolve in a finite set $\{ \phi^{(i)} : 1\le i\le \nphi\}$.   \textit{Convergence}   means that $\phi_n =  \phi^{(\tilde k)} $ for some integer-valued random variable $\tilde k$, and all $n$ sufficiently large. 

The proof of \Proposition{t:ZAP1} is based on a simple  inverse dynamic programming argument:   it is easily shown  that $\haC_n$ is convergent to $c$  in the case $\bfgamma\equiv\bfalpha$, and it is also easily established  that $\lim_{n\to\infty} \theta_n - \Qstar(\haC_n )  =0$ in this case.   The proof of \Theorem{t:ZAP} is more delicate, and is based on extensions of ODE arguments in   \cite{bormey00a}.  

The simplicity of the proof of \Proposition{t:ZAP1} suggests that this case would be preferred.  However, when $\gamma_n\equiv\alpha_n=1/n$  we do not know how to relax the assumption that  $\{\phi_n\}$ is convergent.   Analysis is complicated by the fact that $\haA_n$ is obtained as a uniform average of $\{A_n\}$. 

The ODE analysis in the proof of \Theorem{t:ZAP} suggests that the dynamics of the two time-scale algorithm closely matches the Newton-Raphson ideal.  Moreover, the two time-scale algorithm has the best performance in all of the numerical experiments surveyed in \Section{s:num}.

\subsection{ODE and Policy Iteration}

Recall the definition of $\partialQstar_{\upmf}$ in \eqref{e:partiakQmu}. 
The ODE approximation \eqref{e:ZapODE} can be expressed
\begin{equation} 
\ddt q_t   = -q_t + \partialQstar_{\upmf_t} c  
\label{e:QODE1}
\end{equation} 
where $\upmf_t$ is \textit{any} pmf satisfying $ \partialQstar_{\upmf_t} \qcost_t = q_t$, and the derivative exists for a.e.\ $t$ (see \Lemma{t:ChainRule} for full justification).
This has an interesting geometric interpretation.   Without loss of generality, assume that the cost function is non-negative, so that $\bfmq$ evolves in the positive orthant   $\Re^d_+$ whenever its initial condition lies in this domain.

\begin{figure}[h] 
	\Ebox{.7}{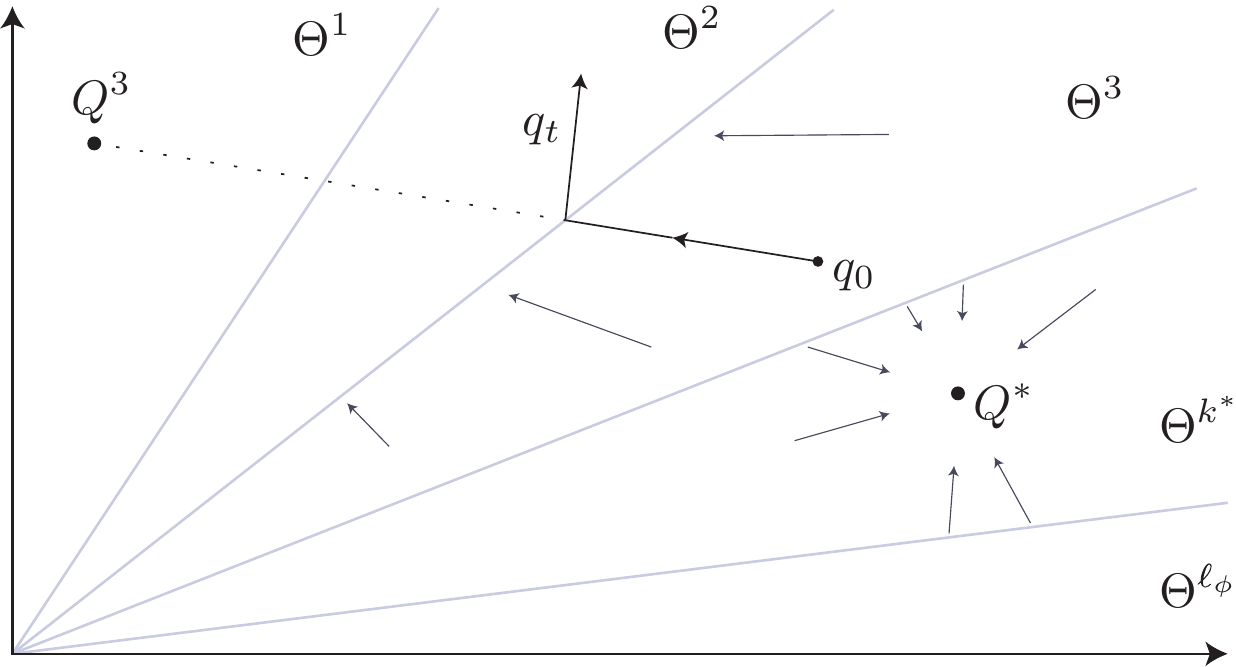} 
	\vspace{-.2cm}
	\caption{\small ODE for SNR2 Q-Learning.    The light arrows show typical vectors in the vector field that defines the ODE \eqref{e:QODE1}.  The solution starting at $q_0\in\Theta^3$ initially moves in a straight line towards $  Q^3$.}  
\label{f:Q-ODE} 
\end{figure} 

A typical solution to the ODE is shown in \Fig{f:Q-ODE}: the trajectory is piecewise linear, with changes in direction corresponding to changes in the policy $\phi^{q_t}$.   Each  set $\Theta^k$ shown in the figure corresponds to a deterministic policy:
\[ 
\Theta^k = \{ q\in\Re_+^d : \phi^q = \phi^{(k)}  \} 
\]

\begin{lemma}
\label{t:Thetak}
	For each $k$ the set $ \Theta^k$  is a   convex polyhedron, and also a positive cone.     
	When $q_t\in\text{interior}(\Theta^k)$ then 
\[
	\partialQstar_{\upmf_t} c  = Q^k \eqdef c+\beta P h_{\phi^{(k)}} \, .
\]
\qed
\end{lemma}

\begin{proof} 
	The power series expansion holds:
\[
	\partialQstar_\upmf c = 
	[I-\beta P S_\phi ]^{-1}c = c +\sum_{n=1}^\infty \beta^n [P S_{\phi}]^n c
\]
	For each $n\ge 1$ we have $ [P S_{\phi}]^n  =  P P_\phi^{n-1} S_{\phi}$, which together with \eqref{e:hphi}
	implies the desired result. 
\end{proof}

The function $ Q^k$  is the fixed-policy Q-function considered in the SARSA algorithm \cite{meysur11,rumnir94,sze10}.   While $q_t$ evolves in the interior of the set $\Theta^k$,  it moves in a straight line towards the function $Q^k$.   On reaching the boundary, it then moves in a straight line to the next Q-function.  This is something like a policy iteration recursion, since the policy $\phi^{q_t}$ is obtained as the argmin over $u$ of $q_t(\varble, u)$.

Of course, it is far easier to establish stability of the equivalent ODE   \eqref{e:ZapODE}.

\subsection{Overview of proofs}

This final subsection is dedicated to the proof of \Proposition{t:ZAP1},  and the main ideas in the proof of \Theorem{t:ZAP}. 
It is assumed throughout the remainder of this section that Assumptions Q1--Q3 hold.
Proofs of  technical lemmas are contained in  \Appendix{s:ZapQProof}.

We require the usual probabilistic foundations:  There is a probability space $(\Omega,\clF,\Prob)$ that supports all random variables under consideration.  The probability measure $\Prob$ may depend on an initialization of the Markov chain.    All stochastic processes under consideration are assumed adapted to a filtration denoted $\{\clF_n : n\ge 0\}$.   

We begin with the proof of the simpler  \Proposition{t:ZAP1}. 

\subsubsection{Inverse Dynamic Programming Analysis}
\label{s:IDP}

\Proposition{t:ZAP1} is   a quick consequence of  the following extension of \Proposition{t:SAMC}:
\begin{proposition}
\label{t:CnRecursion}
Suppose that Assumptions Q1--Q3 hold.  
Suppose moreover that each of the matrices $\{\haA_n : n\ge n_\bullet\}$ is invertible for some $n_\bullet \ge 1$ that is a.s.\ finite.
Then,  the following recursion  holds
for $  n\ge n_\bullet$:
\begin{equation}
\label{e:CnRecursion}
\begin{aligned}
\haC_{n+1}   = \haC_n 
&-
\gamma_{n+1} [ \diagpie^{-1} A_{n+1} \theta_n  + \haC_n    ]
\\[.2cm]
&    +  
\alpha_{n+1}  \diagpie^{-1}   [     \Psi_n c  +  A_{n+1}\theta_n]\,,
\end{aligned}
\end{equation}
where $\Psi_n$ is defined in \eqref{e:Psi}.
\qed
\end{proposition}

\begin{proof}[Proof of \Proposition{t:ZAP1}]

The assumption that the sequence of policies $\{\phi_n\}$ converges to a (possibly random) limit  $\phi_\infty$  has the following consequences:  
First, this 
implies that $\haA_{n}$ defined in \eqref{e:QSNR2Adef} converges: 
\begin{equation}
\lim_{n\to\infty} \haA_{n} = - \diagpie  [I - \beta P S_{\phi_\infty}]  \qquad a.s.\,.
\label{e:haAconverges}
\end{equation}
Second, for all  $n$  sufficiently large the following identities hold, 
by applying the definitions of $\phi_n$ and ${\Qstar}^{-1}$:
\begin{equation}
[I - \beta P S_{\phi_\infty}]  \theta_n = [I - \beta P S_{\phi_n}]  \theta_n  =   {\Qstar}^{-1} (\theta_n ) 
\label{e:haAnThetan}
\end{equation}

From \eqref{e:haAconverges}, since the limit on the right hand side is invertible and the set of all invertible matrices is open, it follows that there is an integer $n_\bullet$ that is finite a.s., and such that $\haA_n$ is invertible for  $n \geq n_\bullet$.


Now applying \Proposition{t:CnRecursion}, the recursion \eqref{e:CnRecursion} is reduced to the following in the case  that $\bfgamma \equiv \bfalpha$:
\[
\haC_{n+1}   = \haC_n + \alpha_{n+1} [ \diagpie^{-1}  \Psi_n c    -    \haC_n    ].
\]
This is essentially a Monte-Carlo average of $\{  \diagpie^{-1}  \Psi_n c : n\ge 0\}$.  
Since the steady state expectation of $\Psi_n$ is equal to $\diagpie$, convergence follows from the Law of Large Numbers:
\begin{equation}
\label{e:CnConv}
\lim_{n\to\infty}  \haC_n  =   c  \qquad a.s.\,.
\end{equation}
Combining equations \eqref{e:haAconverges}, \eqref{e:haAnThetan}, and \eqref{e:CnConv} implies
\[
\lim_{n\to\infty} {\Qstar}^{-1} (\theta_n )   = - \lim_{n\to\infty} \diagpie^{-1} \haA_n \theta_n = 
\lim_{n\to\infty}  \haC_n  = c
\]
\Lemma{t:clQstar} completes the proof:
$\displaystyle
\lim\limits_{n \to \infty} \theta_n = \Qstar (c) =Q^*$.
\end{proof}

\subsubsection{ODE Analysis}
\label{s:ODE}


The remainder of this section is devoted to a high-level view of the proof of the ODE approximation for the two time-scale algorithm, with $\bfalpha$ and $\bfgamma$ defined in \eqref{e:GAINS}.

The construction of an approximating ODE involves first defining a continuous time process.  Denote
\begin{equation}
t_n = \sum_{i=1}^n \alpha_i,\ \  n\ge 1, \quad t_0=0\,,
\label{e:tn}
\end{equation}
and define $\barq_{t_n} = \theta_n$ for these values, with the definition extended to $\Re_+$ via linear interpolation.  We say that the ODE approximation $\ddt q = f(q)$ holds if    we have the approximation,
\[
\barq_{T_0 +t} = \barq_{T_0} + \int_{T_0}^{T_0+t }  f( \barq_\tau)  \, d\tau   + \clE_{T_0,T_0 + t} \,, \quad t, T_0\ge 0\, ,
\]
where the error process satisfies, for each $T>0$,
\[
\lim_{T_0\to\infty }  \sup_{0\le t \le T}  \| \clE_{T_0,T_0 + t} \| =0\quad a.s.
\]
Such approximations will be represented using the more compact notation:
\begin{equation}
\barq_{T_0 +t} = \barq_{T_0} + \int_{T_0}^{T_0+t }  f( \barq_\tau)  \, d\tau   +o(1) \,, \quad  T_0\to\infty \, .
\label{e:ODEgeneral}
\end{equation}

An ODE approximation holds for Watkins algorithm, with $f(q_t)$ defined by the right hand side of  \eqref{e:QODEW}, or in more compact notation:   
\begin{equation}
\begin{aligned}
\ddt q_t  =   \diagpie  [  c    -  \qcost_t   ]\,,
\qquad
\qcost_t=[I - \beta   P S_{\phi^{q_t}} ]q_t 
\end{aligned}
\label{e:ODEW2}
\end{equation}
The significance of this representation  
is that $q_t$ is the Q-function associated with the ``cost function'' $\qcost_t$:   $q_t = \Qstar(\qcost_t)$.

The same notation will be used in the following treatment of Zap~Q-learning.  
Along with the piecewise linear continuous-time process   $\{ \barq_t : t\ge 0 \}$,   denote by $\{\barclA_t : t\ge 0\}$   the  piecewise linear continuous-time process defined similarly, with $\barclA_{t_n} = \haA_n$,  $n\ge 1$,
and $\barqcost_t = {\Qstar}^{-1}(\barq_t) $ for $t\ge 0$.
 


To construct an ODE, it is convenient first to obtain an alternative and suggestive representation for the pair of equations (\ref{e:QSNR2def},\ref{e:QSNR2Adef}).   A vector-valued sequence of random variables $\{\clE_k\}$ will be called \textit{ODE-friendly} if it admits the decomposition,
\begin{equation}
\clE_k   =  \Delta_k  +  \clT_k-\clT_{k-1} + \epsy_k
\quad k\ge 1
\label{e:bor08aNoise}
\end{equation}
in which 
$\{ \Delta_k : k\ge 1\}$ is a martingale-difference sequence satisfying  
$ \Expect[\|  \Delta_{k+1} \|^2\mid \clF_k] \le \bar\sigma^2_\Delta$ a.s.\ for some finite $\bar\sigma^2_\Delta$ and all $k$,
$\{ \clT_k : k\ge 1\}$ is a bounded sequence, and the final sequence  is bounded and satisfies
\begin{equation}
\sum_{k=1}^\infty \gamma_k \|\epsy_k\| <\infty \quad a.s.\, .
\label{e:Friendly3}
\end{equation}

\begin{lemma}
\label{t:pre-ODE}
The pair of equations (\ref{e:QSNR2def},\ref{e:QSNR2Adef}) can be expressed,
\begin{equation}
\begin{aligned}
\theta_{n+1} &= \theta_n + \alpha_{n+1} \haG_{n+1}^* \bigl[- \diagpie[I-\beta P S_{\phi^{\theta_n}} ]\theta_n +   \diagpie c
+\clE^A_{n+1} \theta_n + \clE^q_{n+1} \bigr] 
\\[.2cm]
\haG_{n+1}^* &= - [ \haA_{n+1}]^{-1}
\\[.2cm]
\qquad \haA_{n+1} &= \haA_n + \gamma_{n+1}  \bigl[ - \diagpie [I-\beta P S_{\phi^{\theta_n}} ]  - \haA_n +\clE^A_{n+1}   \bigr]
\end{aligned}
\label{e:SNR2linearSA_preODE}
\end{equation}
in which the sequence $\{ \clE^q_n : n\ge 1\}$  is ODE-friendly.   The   sequence $\{ \clE^A_n \}$ is ODE-friendly provided  Assumption~Q3 holds.
\qed
\end{lemma}
 
The assertion that  $\{ \clE^q_n , \clE^A_n \}$  are ODE-friendly follows from standard arguments based on solutions to Poisson's equation for zero-mean functions of the Markov chain $(\bfmX,\bfmU)$ \cite{shwmak91}. The proof of \Lemma{t:ODEfish} is based on an extension of this technique to the present setting. 
 
\begin{lemma}
\label{t:ODEfish}
	For each $n\ge 0$,
\[
	\Psi_n  P S_{\phi_{n}}  = \diagpie P S_{\phi_{n}}    +   \clT_{n+1} - \clT_n +  \Delta^\Psi_{n+1} +\epsy_{n+1}
\]
	where $\{\Delta^\Psi_k\}$ is a martingale difference sequence with uniformly bounded second moment, and the sequences $\{ \clT_k, \ \epsy_k : k\ge 0\}$ are also bounded.    If Assumption~Q3 holds then  $\{\epsy_k\}$ satisfies \eqref{e:Friendly3}.\qed
\end{lemma}

The representation in \Lemma{t:pre-ODE}
appears similar to an Euler approximation of the solution to an ODE: 
\[
{\theta_{n+1}
	\choose
	\haA_{n+1}
}
={\theta_n
	\choose
	\haA_n}
+ { \alpha_{n+1} f_\Theta(\theta_n, \haA_n) 
	\choose
	\gamma_{n+1} f_A(\theta_n, \haA_n) 
} + \clE_{n+1}
\]
It is  discontinuity of the function $f_A$ that presents the most significant challenge in analysis of the algorithm ---  this violates standard conditions for existence and uniqueness of solutions to the ODE without disturbance.   

%
%

Fortunately there is special structure that will allow the construction of an ODE approximation.   
Some of this structure is highlighted in the lemma that follows.   These approximations  are taken from Lemmas~\ref{t:haA_Bdd_And_ODE_sol}
 and \ref{t:haq_Bdd_And_ODE_sol}.

\begin{lemma}
\label{t:haA_ODE}
For each $t,T_0\ge 0$,
\begin{align}
\barq_{T_0 +t} &= \barq_{T_0}
-   \int_{T_0}^{T_0+t }  \barclA_\tau^{-1}  \diagpie  \bigl\{c - \barqcost_\tau    \bigr\}  \, d\tau   + o(1)
\label{e:preODE_Zap_q}
\\[.2cm]
\barclA_{T_0 +t} &= \barclA_{T_0}   
-   \int_{T_0}^{T_0+t }\bigl\{    \diagpie  [I-\beta P S_{\phi^{\barq_\tau}} ]  + \barclA_\tau    \bigr\} g_\tau  \, d\tau   +  o(1)\,, \quad  T_0\to\infty \, .
\label{e:preODE_Zap_A}
\end{align}
where   $g_t \eqdef \gamma_n/\alpha_n$ when $t=t_n$ for some $n$, and extended to all $t\in\Re_+$ by   linear interpolation.
	\qed 
\end{lemma}

The ``gain''  $g_t$ appearing in \eqref{e:preODE_Zap_A} converges to infinity rapidly as $t\to\infty$:     Based on the definitions   in \eqref{e:GAINS}, it follows from \eqref{e:tn} that $t_n\approx \log(n)$ for large $n$, and consequently $g_t \approx \exp((1-\rho)t)$ for large $t$.   This suggests that the integrand $     \diagpie  [I-\beta P S_{\phi^{\barq_t}} ]  + \barclA_t  $ should converge to zero rapidly with $t$.   This intuition is made precise in the Appendix.  
Through several subsequent transformations, these integral equations are shown to imply the ODE approximation in 
 \Theorem{t:ZAP}.

\subsection{An $O(d)$ Zap-Q learning algorithm}
\label{s:OdZap}

In this subsection, we introduce an $O(d)$ Zap-Q learning algorithm, which is basically the   $O(d)$ Zap-SNR algorithm described in \Section{s:OdZapSNR} specialized to   Q-learning. 

Based on the equations \eqref{e:OdSNR2nonlinearSA}, \eqref{e:OdSNR2Defs}, \eqref{e:QSNR2def}, and \eqref{e:QSNR2Adef}, the algorithm is defined as follows:  
\spm{2018  Why do you fix $N = d$?  You could say at the end that you choose N=d in simulations.}
 Fix $N = d$, and for $i \geq 0$,
\begin{equation} 
\begin{aligned}
\theta_{(i+1)N} 
&= \theta_{iN} + \alpha_{i+1}\haG_{(i+1)N}^* \haf(\theta_{iN})
\\
\haA_{(i+1)N} &= \haA_{iN} + \hagamma_{i+1}  \bigl[ \nabla \haf(\theta_{iN}) - \haA_n  \bigr]
\end{aligned}
\label{e:OdZapQ}
\end{equation}
where,
\begin{equation}
\begin{aligned}
\haf(\theta_{iN}) &= N^{-1} \displaystyle \sum_{j = iN + 1}^{(i+1)N} \bigl\{  c(X_j,U_j)   + \beta     \utheta_{iN} (X_{j+1})  - \theta_{iN}(X_j,U_j)  \bigr\}   \psi(X_j,U_j),
\\
\nabla \haf(\theta_{iN}) &= N^{-1} \displaystyle \sum_{j = iN + 1}^{(i+1)N} \psi(X_j,U_j) \bigl[  \beta    \psi (X_{j+1}, \phi^{\theta_{iN}}(X_{j+1} ))  - \psi(X_j,U_j)  \bigr]^\transpose  ,
\\
\hagamma_{i+1} &= 1 - \displaystyle \prod_{j=iN+1}^{(i+1)N} (1-\gamma_j) ,
\end{aligned}
\label{e:OdQSNR2Defs}
\end{equation}
$\haG_{(i+1)N}^* = - [ \haA_{(i+1)N}]^{-1}$, with $[\varble]$ denoting a projection, chosen so that $\{\haG_{(i+1)N}^*\}$ is a bounded sequence. For a given parameter vector $\theta$, the policy $\phi^{\theta}$ is defined in \eqref{e:phi_q_def}.

 Once again, we claim that the asymptotic properties of the above defined $O(d)$ Zap-Q learning algorithm is the same as that of the Zap-Q learning algorithm defined in \eqref{e:QSNR2def} and \eqref{e:QSNR2Adef}, with justification postponed to a future version of the paper.

%
%

\section{Numerical Results}
\label{s:num}

\archive{\rd{We are doing so much analysis assuming $A$, is it possible that the reviewer might feel like we are assuming a lot of things? Without understanding the fact that we are doing all this only for analysis?}}

Results from numerical experiments are surveyed here to illustrate the performance of the Zap Q-learning algorithm (\ref{e:QSNR2def},\ref{e:QSNR2Adef}). Comparisons are made with several existing algorithms, including Watkins Q-learning \eqref{e:Watkin}, Watkins Q-learning with Ruppert-Polyak-Juditsky (RPJ) averaging \cite{rup88,pol90,poljud92}, Watkins Q-learning with a ``polynomial learning rate" \cite{eveman03}, and the more recent \emph{Speedy Q-learning} algorithm \cite{azamunghakap11}.  

\archive{\bl{Are we giving Watkins too much credit??}}

In addition, the Watkins algorithm with a scalar gain $g$ is considered, with $g$ chosen so that the algorithm has finite asymptotic covariance.  When the value of $g$ is optimized and numerical conditions are favorable (e.g.,  the condition number of $A$ is not too large) it is found that the performance is nearly as good as the Zap-Q algorithm.  However, there is no free lunch:
\begin{romannum}
\item   Design of the scalar gain $g$ depends on approximation of $A$, and hence $\theta^*$.
While it is possible to estimate $A$ via Monte-Carlo in Zap Q-learning, it is not known how to efficiently update approximations for an optimal scalar gain.

\item  
A reasonable asymptotic covariance required a large value of $g$.   Consequently,
the scalar gain algorithm had massive transients, resulting in a poor performance in practice.

\item  Transient behavior could be tamed through projection to a bounded set.  However, this  again requires prior knowledge of the region in the parameter space to which $\theta^*$ belongs.

Projection of parameters was also necessary for RPJ averaging. 
\end{romannum}

The following batch mean method was used to estimate the asymptotic covariance.

\paragraph{Batch Mean Method}
At stage $n$ of the algorithm we will be interested in the distribution of a vector-valued random variable of the form $f_n(\theta_n)$,   where  $f_n \colon \Re^d \to \Re^m$ is possibly dependent on $n$.
The batch mean method is used to estimate its statistics:
For each algorithm, $N$ parallel simulations are run with $\theta_0$ initialized i.i.d.\ according to some distribution.
Denoting $\theta_n^i$ to be the vector $\theta_n$ corresponding to the $i^{th}$ simulation, the distribution of the random variable $f_n(\theta_n)$ is estimated based on the histogram of the independent samples $\{f_n(\theta_n^i) : 1\le i\le N\}$.  

An important special case in this paper is $f_n(\theta_n) = W_n\eqdef \sqrt{n} (\theta^n -\theta^*)$.   However, since the limit $\theta^*$ is not available, the empirical mean is substituted: 
\begin{equation}
	W_n^i  = \sqrt{n} [\theta_n^i  -  \bartheta_n]\, ,  \qquad     \bartheta_n \eqdef \frac{1}{N} \sum_{i=1}^N  \theta_n^i.
\label{e:BMMAsymCov}
\end{equation}
The estimate of the covariance of $f_n(\theta_n)$ is then obtained as the sample covariance of $\{W_n^i, \,1 \leq i \leq N\}$. This corresponds to the estimate of the asymptotic covariance $\Sigma_\theta$ defined in \eqref{e:SAsigma}.  

The value  $N = 10^3$ is used in all of the experiments surveyed here.

 \begin{wrapfigure}{l}{.25\hsize}
\vspace{-.5em}
\Ebox{.95}{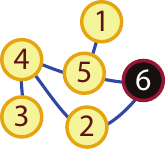}
\vspace{-1.15em}
	\caption{\small  
	Graph for MDP }\label{6StateGraph}
\vspace{-.3em}
\end{wrapfigure}

\subsection{Finite state-action MDP}

Consider first a simple stochastic-shortest-path problem.  The state space $\state = \{1,\ldots,6\}$ coincides with the  six nodes on the un-directed graph shown in \Fig{6StateGraph}. The action space $\U =\{ e_{x,x'} \}$,
$x,x' \in \state$,
 consists of all feasible edges along which an agent can travel, including each ``self-loop'', $u=e_{x,x} $.  
The number of state-action pairs for this example coincides with the number of nodes plus twice the number of edges:  $d = 18$.

The controlled transition matrix is defined as follows:    If $X_n=x\in\state$, and  $U_n=e_{x,x'} \in \U$, then  $X_{n+1} = x'$ with probability $0.8$, and with probability $0.2$, the next state is randomly chosen between all neighboring nodes.     
The goal is to reach the state $x^*=6$ and maximize the time spent there.   This is modeled through a discounted-reward optimality criterion with discount factor $\beta \in (0,1)$.  
The one-step reward  is defined as follows:  
\[
r(x,u) = \begin{cases}  	0 & u= e_{x,x} \,,\ x\neq 6  \quad 
\\
-100 &  u= e_{4,5}    
\\
100 &  u= e_{x,6}   
\\
-5 &  \text{otherwise}
\end{cases}
\]
The solution to the discounted-cost optimal control problem can be computed numerically for this model; the optimal policy is unique and independent of $\beta$.

Six different variants of Q-learning were tested: 
\archive{AD is correct to write this:  
{Power of asymptotic theory? Wouldn't a better motivation be to find the best algorithm?}}
\begin{enumerate} 
	\item Watkins' algorithm with  scalar gain $g$, so that $\alpha_n\equiv g/n$
	\item Watkins' algorithm using RPJ averaging, with $\gamma_{n} \equiv  (\alpha_n)^{0.6} \equiv n^{-0.6}$
	\item Watkins' algorithm with the polynomial learning rate $\alpha_n \equiv n^{-0.6}$
	\item Speedy Q-learning 
	\item Zap Q-learning with $\bfalpha \equiv \bfgamma$  
	\item Zap Q-learning with  $\gamma_{n} \equiv  (\alpha_n)^{0.85} \equiv n^{-0.85}$  
	 
\end{enumerate}
The basis was taken to be the same as in Watkins Q-learning algorithm.
In each case, the randomized policy was taken to be uniform:   feasible transitions were sampled uniformly   at each time.   

Discount factors $\beta=0.8$ and $\beta=0.99$ were considered.    In each case, the unique optimal parameter $\theta^*=Q^*$ was obtained numerically.

\paragraph{Asymptotic Covariance}

Speedy Q-learning cannot be represented as a standard stochastic approximation, so standard theory cannot be applied to obtain its asymptotic covariance.     The Watkins' algorithm with polynomial learning rate has infinite asymptotic covariance.

For the other four algorithms, the asymptotic covariance $\Sigma_{\theta}$ was computed by solving the Lyapunov equation \eqref{e:GAINLyap}  based on the matrix gain $G$ that is particular to each algorithm.   Recall that  $G=-A^{-1}$ in the case of either of the Zap-Q algorithms.

The  matrices $A$ and $\Sigma_{\Delta}$ appearing in \eqref{e:GAINLyap}  are defined with respect to Watkins' Q-learning algorithm with $\alpha_n=1/n$.  The   first matrix is 
$A = -\diagpie [I - \beta   P S_{\phi^*}]$
under the standing assumption that the optimal policy is unique.    The proof that this is a linearization comes first from the representation of  the ODE approximation \eqref{e:QODEW} in vector form: 
\begin{equation}
\ddt q_t =  \barf(q_t) =   A_{\phi_t} q_t  - b    \,,\qquad \textit{where}\quad A_{\phi} = -\diagpie [I - \beta   P S_{\phi}]\, \,,  \quad b = -\diagpie c\, .
\label{e:QODEWb}
\end{equation} 
Uniqueness of the optimal policy implies that $\barf$ is locally linear:  there exists $\epsy>0$ such that
\[
\barf(\theta) - \barf(\theta^*) =   A(\theta-\theta^*) ,\qquad \|\theta-\theta^*\|\le \epsy\, .
\] 
The matrix  $\Sigma_{\Delta}$ was also obtained numerically, without resorting to simulation.

\begin{figure}[h] 
	\Ebox{.55}{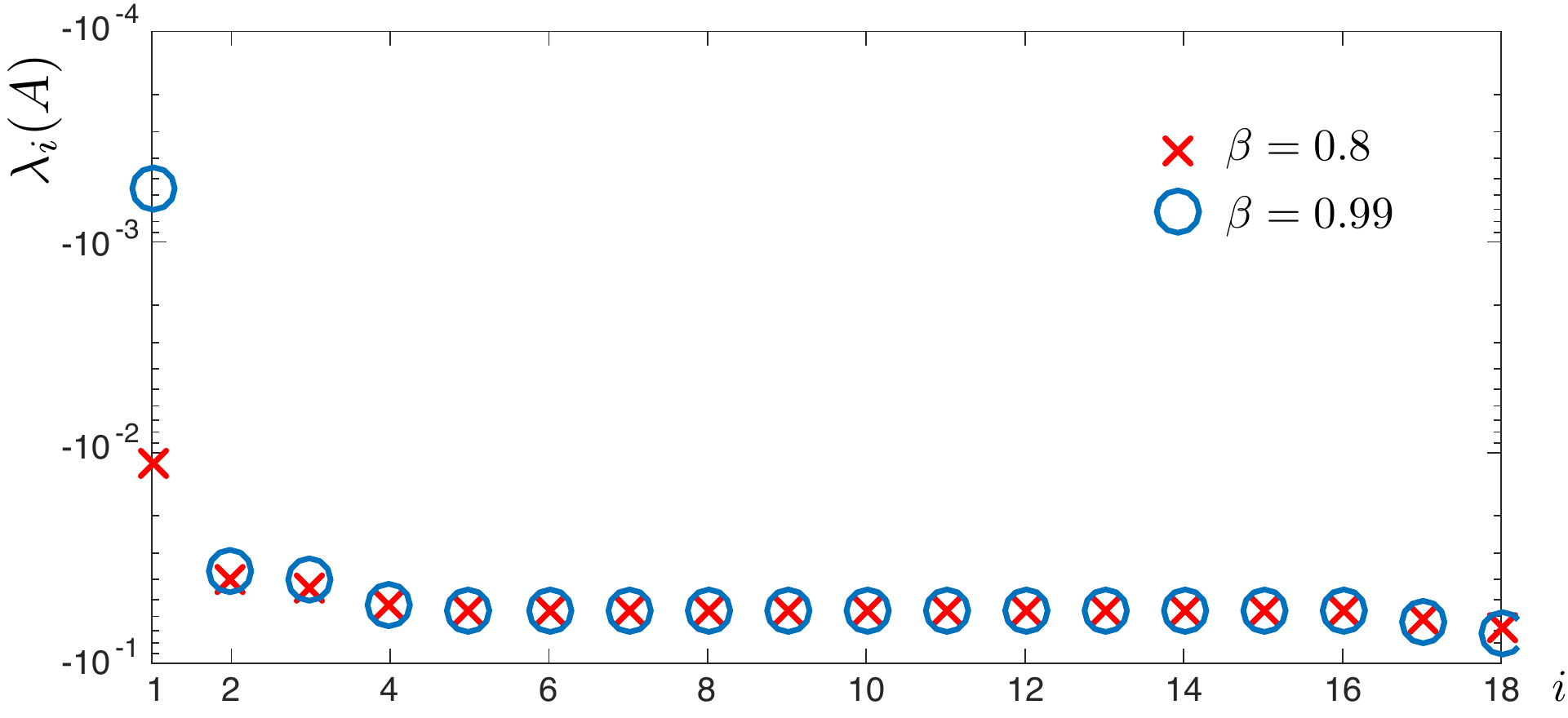}
	\caption{\small Eigenvalues of the matrix $A$ for the 6-state example}
\label{Eig6state}
\end{figure}

The eigenvalues of the $18\times 18$ matrix $ A$ are real in this example, as shown in \Fig{Eig6state} for both values of $\beta$.     To ensure that the eigenvalues of $gA$ are all strictly less than $-1/2$ in a scalar gain algorithm requires  the (approximate) lower bounds $g>45$ for $\beta=0.8$,
 and $g>900$ for $\beta=0.99$.   \Theorem{t:aCov}  implies that the asymptotic covariance $\Sigma_\theta(g)$ is finite for this range of $g$ in the Watkins algorithm with $\alpha_n \equiv g/n$.
 \Fig{TraceSigmaTgVsg} shows the normalized trace of the asymptotic covariance 
 as a function of $g>0$, and the significance of $g \approx 45$ and $g \approx 900$.

\begin{figure}[h]
	\vspace{0in}
	\Ebox{1}{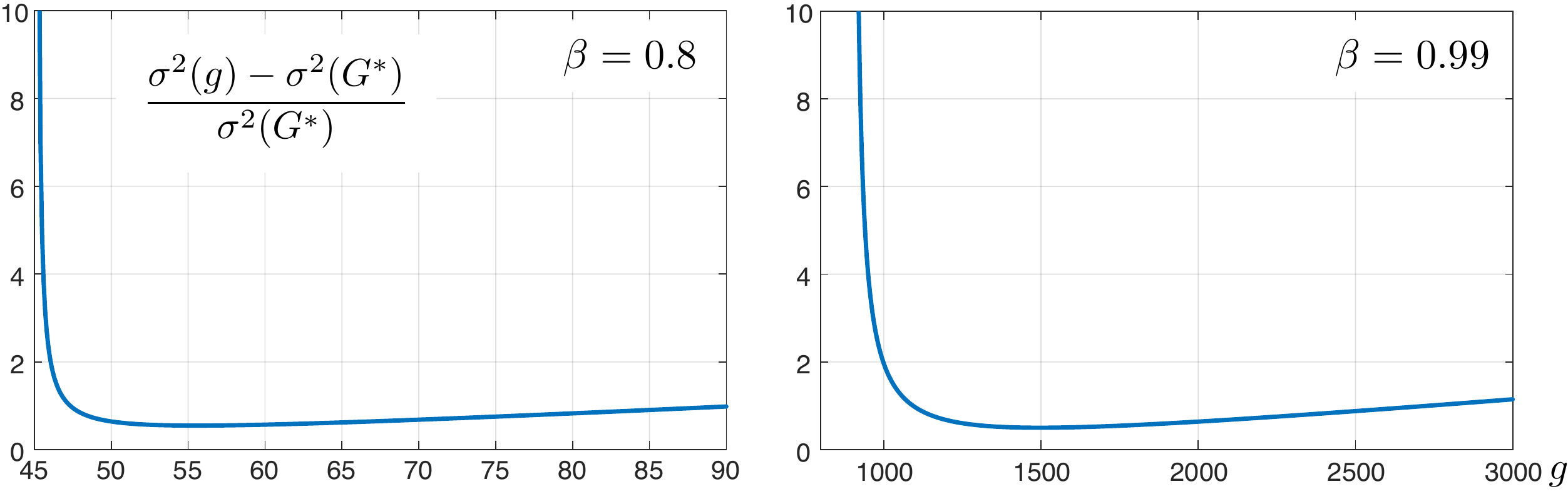}
	\vspace{0in}
	\caption{\small The normalized trace of the asymptotic covariance for the scaled Watkins algorithm with different scalar gains $g$, for the $6$-state example:  $\sigma^2(g)=\trace (\Sigma_\theta(g))$  and $\sigma^2(G^*)=\trace ( \Sigma^*) $.}
\label{TraceSigmaTgVsg}
\end{figure}

Based on this analysis or on \Theorem{t:Qinfinite}, it follows that the asymptotic covariance is not finite for the standard Watkins' algorithm with $\alpha_n \equiv 1/n$.  In simulations it was found that the parameter estimates are not close to $\theta^*$ even after many millions of samples. This is illustrated for the case $\beta = 0.8$ in \Fig{AsymCov_Watg1}, which shows a histogram of $10^3$ estimates of $\theta_n(15)$ with $n=10^6$  (other entries showed similar behavior).

\begin{figure}[h] 
	\Ebox{.52}{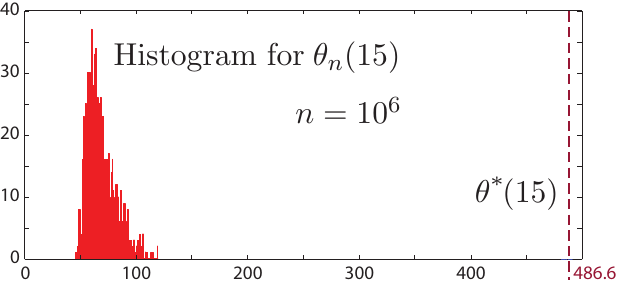}
	\vspace{-.15in}
	\caption{\small Histogram of $10^3$ estimates of $\theta_n(15)$, with $n=10^6$ for the Watkins algorithm applied to the 6-state example with discount factor $\beta = 0.8$}
\label{AsymCov_Watg1}
\end{figure}

It was found that the algorithm performed very poorly in practice for any scalar gain algorithm.  
For example, more than half of the $10^3$ experiments using  $\beta = 0.8$ and $g = 70$ resulted in values of $\theta_n(15)$ exceeding $\theta^*(15)$ by $10^4$ (with $\theta^*(15) \approx 500$), even with   $n=10^6$.  
The algorithm performed well with the introduction of projection in the case $\beta = 0.8$.  With $\beta = 0.99$, the  performance was unacceptable for any scalar gain, even with projection.

The results presented next used a gain of $g=70$  in the case $\beta = 0.8$,  and projection of each entry of the estimates to the interval $(-\infty , 1000]$.
\Fig{AsymCov_Wat2} shows normalized histograms of $\{W_n^i(k) : 1 \leq i \leq N\}$, as defined in 	\eqref{e:BMMAsymCov}, with $k=10, 18$. 

The Central Limit Theorem holds: $W_n$ is expected to be approximately normally distributed: $\clN (0,\Sigma_\theta (g))$, when $n$ is large.  Of the $d=18$ entries of the vector $W_n$, with $n\ge 10^4$,  it was found that the asymptotic variance matched the histogram nearly perfectly for $k=10$, while $k=18$ showed the worst fit.

\begin{figure}[h]
	\vspace{0in}
	\Ebox{1}{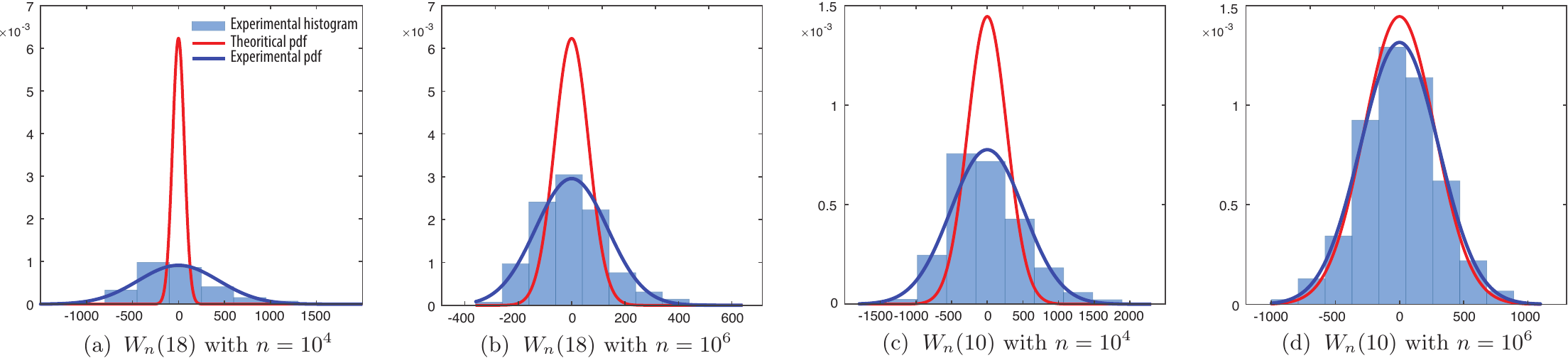}
	\vspace{0in}
	\caption{\small Comparison of theoretical and empirical asymptotic variance for  the scaled Watkins' algorithm, with  gain $g = 70$, applied to the 6-state example  with discount factor $\beta = 0.8$}
\label{AsymCov_Wat2}
\end{figure}


These experiments were repeated for each of  the Zap-Q algorithms, for which the asymptotic variance $\Sigma^*$ is obtained using the formula \eqref{e:SigmaStar}.  Plots are shown only for Case~2: the two time-scale algorithm,   with $\gamma_n = (\alpha_n)^{0.85}$.   Histograms in the case of $\beta = 0.8$ are shown in  \Fig{AsymCov_SNR2a}, and \Fig{AsymCov_SNR2a99} for $\beta = 0.99$. The covariance estimates and the Gaussian approximations match the theoretical predictions remarkably well for $n\ge 10^4$.    
\begin{figure}[h]
	\vspace{0in}
	\Ebox{1}{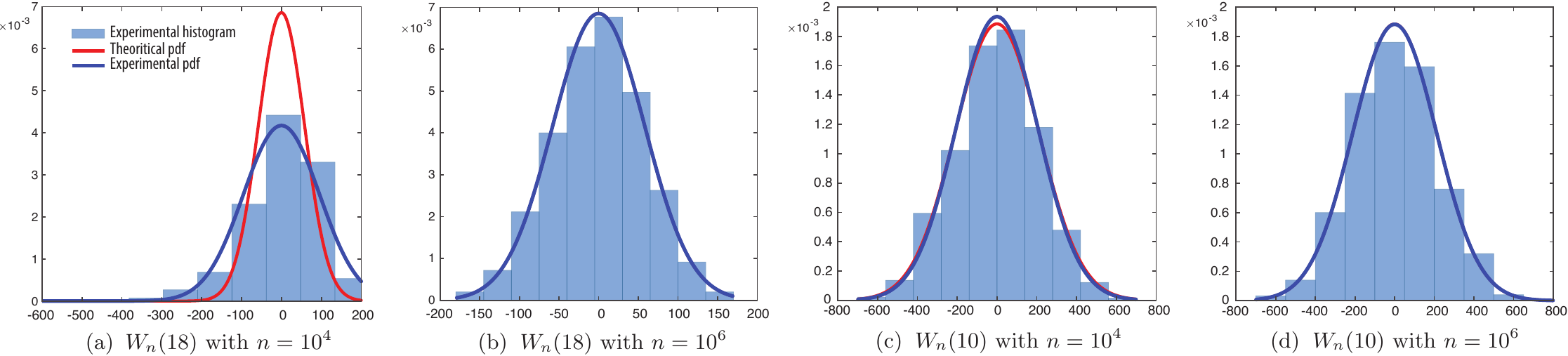}
	\vspace{0in}
	\caption{\small Comparison of theoretical and empirical asymptotic variance of the two time-scale Zap-Q algorithm applied to the 6-state example; $\beta = 0.8$}
\label{AsymCov_SNR2a}
\end{figure}

\begin{figure}[h]
	\vspace{0in}
	\Ebox{1}{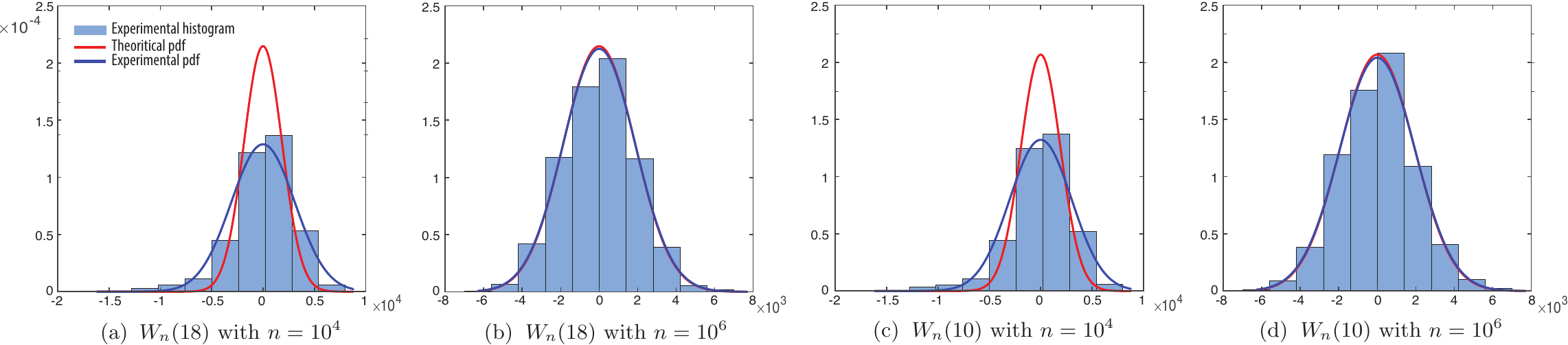}
	\vspace{0in}
	\caption{\small Comparison of theoretical and empirical asymptotic variance of the Zap-Q-learning algorithm applied to the 6-state example; $\beta = 0.99$}
\label{AsymCov_SNR2a99}
\end{figure}

\archive{\rd{Important: We are doing so much analysis assuming $A$, is it possible that the reviewer might feel like we are assuming a lot of things? Without understanding the fact that we are doing all this only for analysis?}}

\begin{figure}[h]
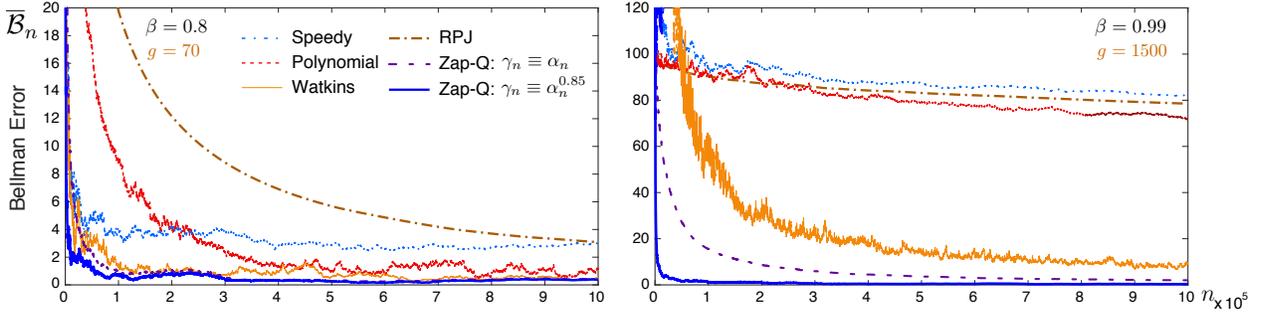

	\Ebox{1}{6State_BEPlot_Beta08099.pdf}
	\caption{\small Maximum Bellman error $\{ \maxBE_n : n\ge 0\}$ for the six   Q-learning algorithms}
\label{6stateBEPlot}
\end{figure}

\paragraph{Bellman Error}

\archive{{Important! Is the definition of $W_n^i$ fixed or changing? I mean, is it a variable? Didn't we define it to be something in the beginning? Using it for BE too?}
\\
Answer:  I want  $W_n^i$ fixed.}

The Bellman error at iteration $n$  is denoted:
\[
\BE_n(x,u) = \theta_n(x,u) -  r(x,u) - \beta \sum_{x'\in\state}  P_u(x,x')  \max_{u'} \theta_n(x',u')\, .
\]
This is identically zero if and only if $\theta_n = Q^*$.   If  $\{\theta_n\}$ converges to $Q^*$ then $\BE_n = \tiltheta_n - \beta P S_{\phi^*} \tiltheta_n$  for all sufficiently large $n$, and the CLT holds for $\{\BE_n\}$ whenever it holds for   $\{\theta_n\}$.  Moreover, on denoting the maximal error  
\begin{equation}
\maxBE_n = \max_{x,u}  | \BE_n(x,u) | \,,
\label{e:BError}
\end{equation}
the sequence $\{\sqrt{n} \,\maxBE_n\}$ also converges in distribution as $n\to\infty$.
\Fig{6stateBEPlot} contains   plots of $\{\maxBE_n \}$ for the six different Q-learning algorithms.

For large $n$, the two versions of Zap Q-learning exhibit similar behavior since $\haA_{n}$  converges to $A$ in both   algorithms.  Though all six algorithms perform reasonably well when $\beta = 0.8$, 
Zap Q-learning is the only one that achieves near zero Bellman error within $n = 10^6$ iterations in the case $\beta = 0.99$. Moreover, the performance of the two time-scale algorithm is clearly superior to the one time-scale algorithm. 

\begin{figure}[h]
	\vspace{0in}
	\Ebox{1}{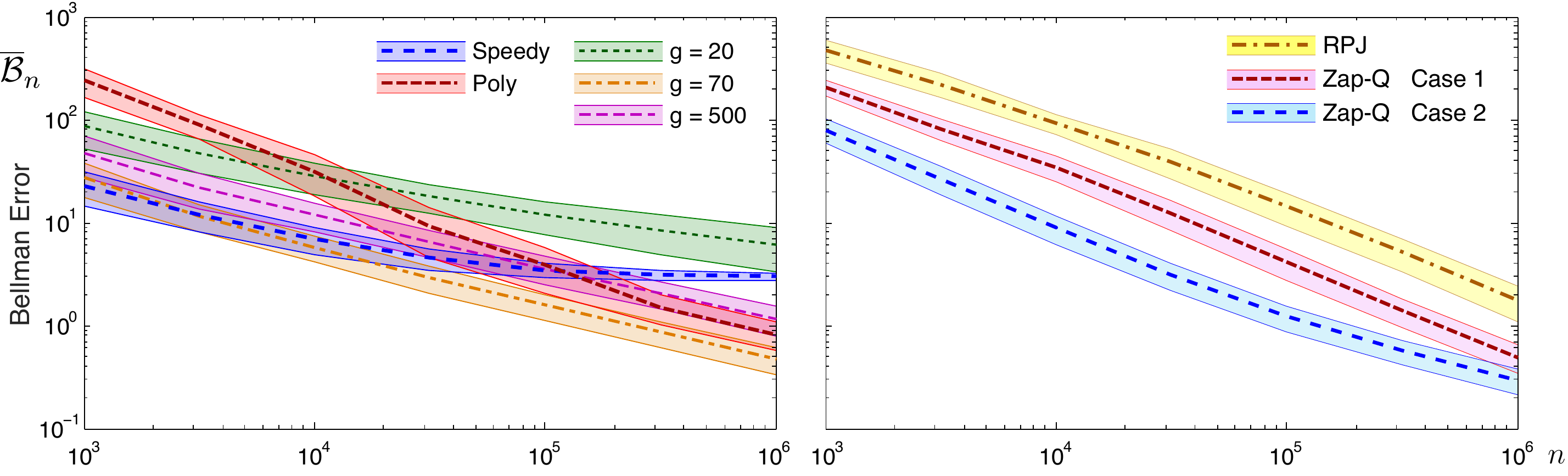}
	\vspace{0in}
	\caption{\small Simulation-based $2 \sigma$ confidence intervals for the six Q-learning algorithms with discount factor $\beta = 0.8$.}
\label{6StateConfIntBEPlot1}
\end{figure}

\archive{\todo Change the scalar gain plots with cone
\\
spm:   This is almost another paper!}

\begin{figure}[h]
	\vspace{0in}
	\Ebox{1}{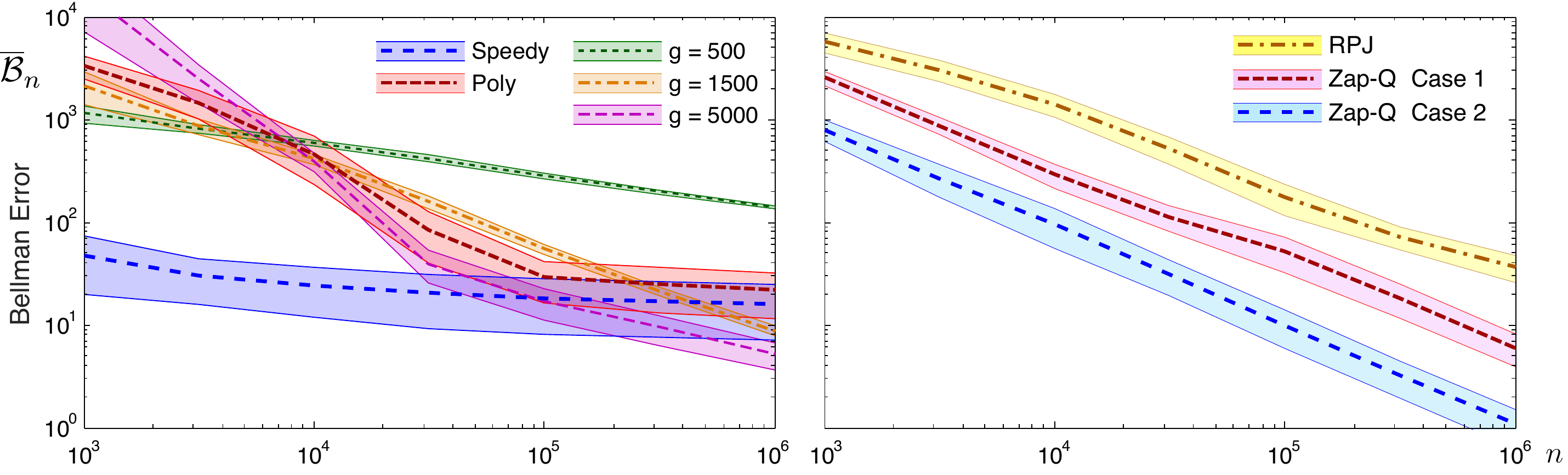}
	\vspace{0in}
	\caption{\small Simulation-based $2 \sigma$ confidence intervals for the six Q-learning algorithms with discount factor $\beta = 0.99$.}
	
\label{6StateConfIntBEPlot2}
\end{figure}

 \Fig{6stateBEPlot} shows only the typical behavior --- repeated trails were run to investigate the range of possible outcomes. 
For each algorithm, the outcomes of $N = 1000$ independent simulations resulted in samples   $\{\maxBE_n^i,\, 1 \leq i \leq N\}$,  with  $\theta_{0}$ uniformly distributed on the interval $[-10^3,10^3]$ for $\beta = 0.8$ and $[-10^4,10^4]$ for $\beta = 0.99$.

 The batch means method was  used to obtain estimates of the mean and variance of $ \maxBE_n$ for a range of values of $n$.   Plots of the mean and  $2 \sigma$ confidence intervals are shown in \Fig{6StateConfIntBEPlot1} for the case $\beta = 0.8$,   and  plots for $\beta = 0.99$ are shown in  \Fig{6StateConfIntBEPlot2}.

\Fig{6StateHistMaxBEPlot08} and \Fig{6StateHistMaxBEPlot99} shows   histograms of  $\{\maxBE_n^i,\, 1 \leq i \leq N\}$, $n = 10^6,$ for all the six algorithms;   this corresponds to the data shown in  \Fig{6StateConfIntBEPlot1} and \Fig{6StateConfIntBEPlot2} at $n = 10^6$.

\begin{figure}[H]
	\vspace{0in}
	\Ebox{0.85}{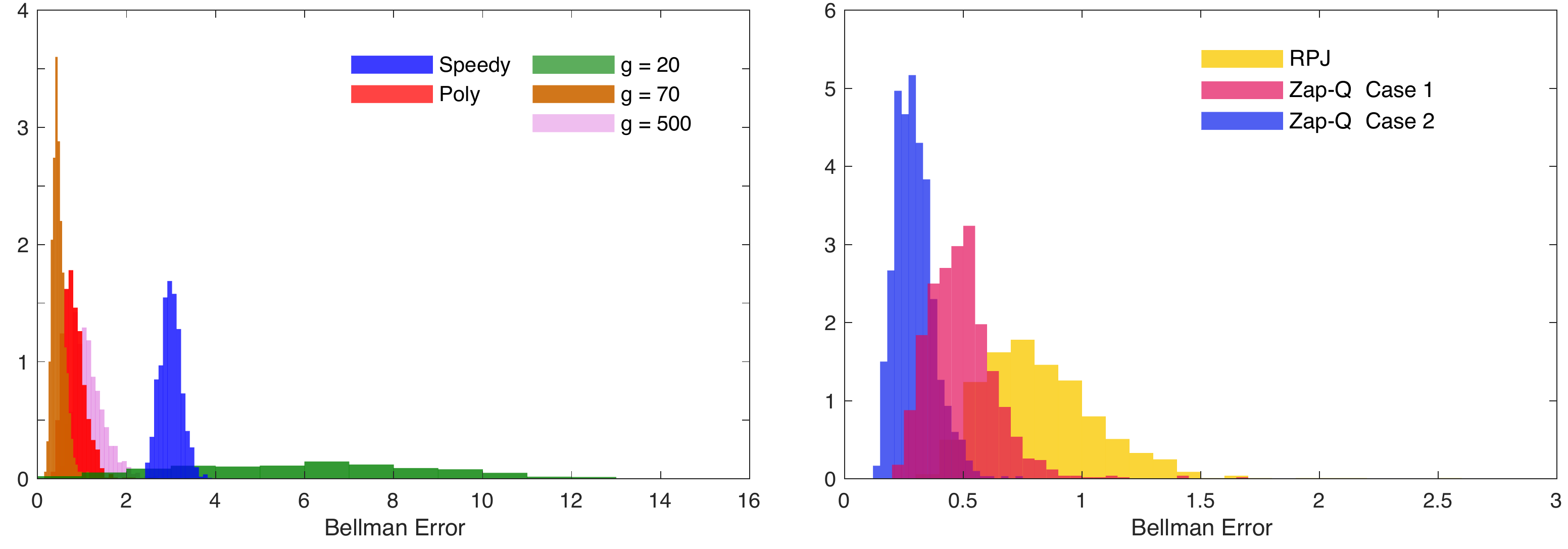}
	\vspace{0in}
	\caption{\small Histogram of the maximal Bellman error when discount factor $\beta = 0.8$ and number of iterations $n=10^6$.}
\label{6StateHistMaxBEPlot08}
\end{figure}

\begin{figure}[H]
	\vspace{0in}
	\Ebox{0.85}{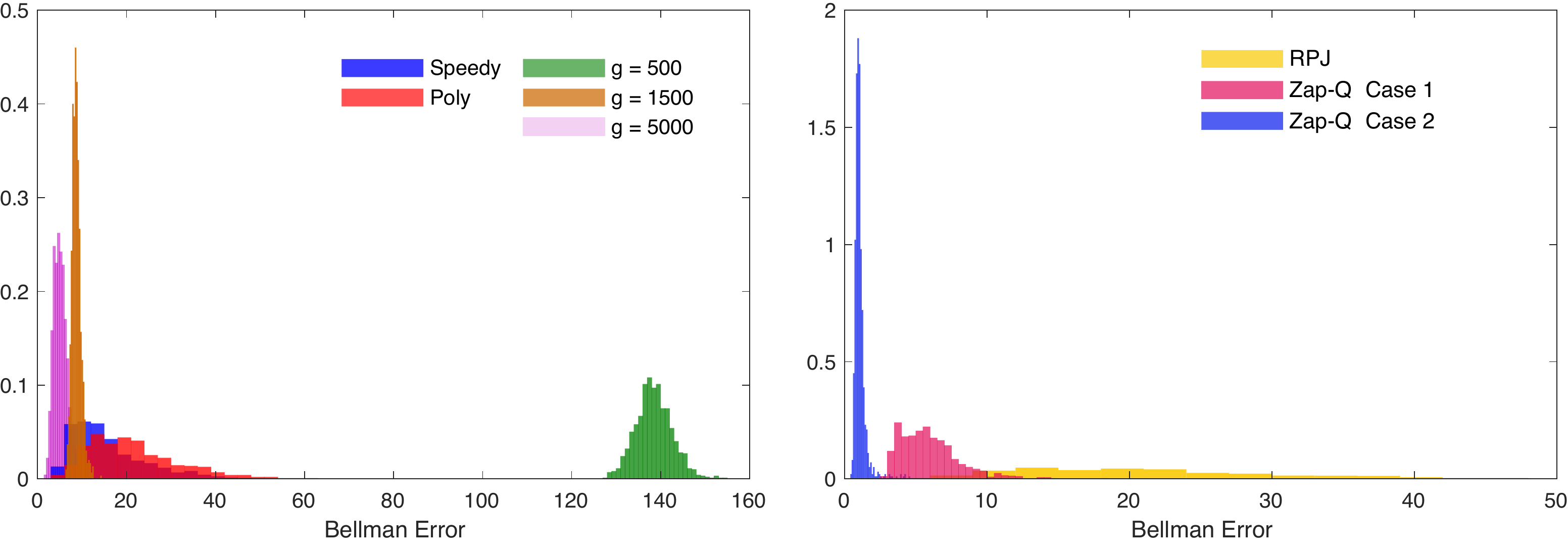}
	\vspace{0in}
	\caption{\small Histogram of the maximal Bellman error when discount factor $\beta = 0.99$ and number of iterations $n=10^6$.}
\label{6StateHistMaxBEPlot99}
\end{figure}

\subsubsection{Performance   of the $O(d)$ Zap-Q learning algorithm}

In this subsection, we test the performance of the $O(d)$~Zap-Q learning algorithm that was defined in equations \eqref{e:OdZapQ} and \eqref{e:OdQSNR2Defs} of \Section{s:OdZap} by applying it to the stochastic shortest path problem. We restrict to the comparison of the Bellman errors (defined in \eqref{e:BError}) of the different algorithms, and we consider the case $\beta = 0.99$. 

\Fig{ZapOdFig} contains   plots of $\{\maxBE_n \}$ for the different Q-learning algorithms. For the $O(d)$~Zap-Q learning algorithms, the batch size was set to $N = 100$ ($d = 18$ in this problem). 
\spm{2018   See,  d isn't N!}
We notice in the figure that the $O(d)$ algorithm performs nearly as well as the $O(d^2)$ algorithm when the step-sizes ($\hagamma_i$) are chosen appropriately. Furthermore, the naive batching technique applied to a single-time-scale Stochastic Newton-Raphson algorithm ($\gamma_n \equiv \alpha_n$ and therefore $\hagamma_i \approx \alpha_i$) performs extremely poorly.

\begin{figure}[H]
\Ebox{0.85}{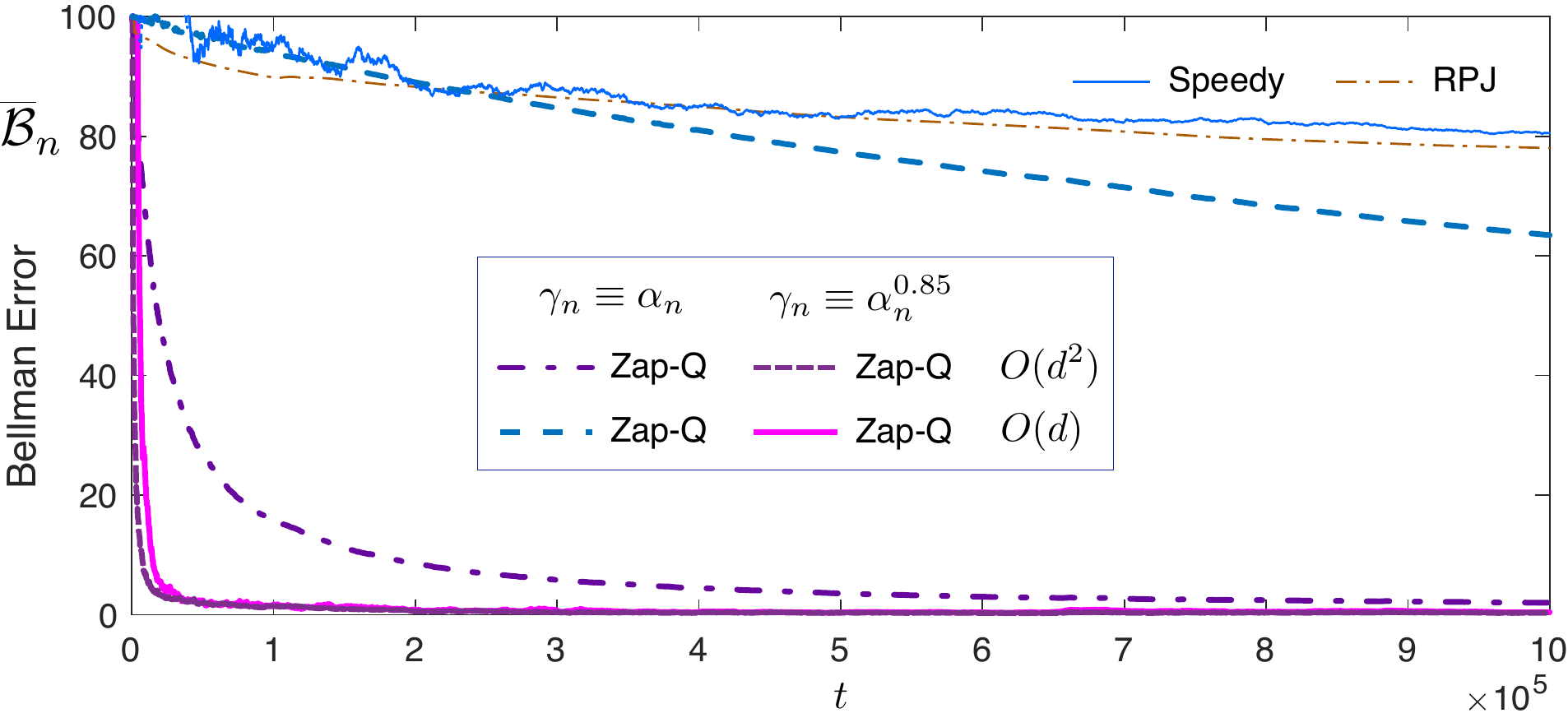}
\vspace{-.25em}
\caption{\small Maximum Bellman error $\{ \maxBE_n : n\ge 0\}$ for different Q-learning algorithms}
\label{ZapOdFig}
\end{figure}

\subsection{Finance model}

The next example is taken from \cite{tsiroy99,choroy06}.     The reader is referred to these references for complete details of the problem set-up and the reinforcement learning architecture used in this prior work.    
The example is of interest because it shows how the Zap Q-learning algorithm can be used with a more general basis, and also how the technique can be extended to optimal stopping time problems.  

The Markovian state process considered in \cite{tsiroy99,choroy06} is the vector of ratios: 
\[
X_n = (\tilp_{n-99},\tilp_{n-98}, \ldots,\tilp_{n} )^\transpose / \tilp_{n-100}\,\,, \quad n =0,1,2,\ldots\,
\]
in which $\{\tilp_t : t\in\Re\}$ is a geometric Brownian motion (derived from an exogenous price-process).
 This uncontrolled Markov chain is positive Harris recurrent on the state space $\state \equiv \Re^{100}$ \cite{MT}.

The ``time to exercise'' is modeled as a stopping time $\tau \in \bfmZ^+$.   The associated expected reward is defined as $
 \Expect [ \beta^{\tau} r(X_\tau) ]$,  with $r(X_n) \eqdef X_n(100) = \tilp_{n} / \tilp_{n-100}$ and $\beta \in (0,1)$ fixed.   The objective of finding a policy that maximizes the expected reward is modeled as an optimal stopping time problem. 

The value function is defined to be the supremum over all stopping times:
\begin{equation}
{{{h}}}^*(x) = \sup_{\tau > 0 } \Expect [ \beta^{\tau} r(X_\tau) \mid X_0 = x ].
\label{e:JstarFin}
\end{equation}
This solves the Bellman equation:  
\begin{equation}
{{{h}}}^*(x) = \max \big (r(x) , \beta \Expect [{{{h}}}^*(X_{n+1}) \mid X_n = x ] \big )\,\qquad x\in\state\,.
\label{e:FinBel}
\end{equation}
The associated Q-function is denoted 
$
Q^*(x) \eqdef \beta  \Expect [{{{h}}}^*(X_{n+1}) \mid X_n = x ]
$,
which solves a similar fixed point equation:
\[
Q^*(x) = \beta \Expect [ \max (r(X_{n+1}) , Q^*(X_{n+1}) )  \mid X_n = x  ].
\]

A stationary policy $\phi\colon \state \to \{0,1\}$ assigns an action for each state $x \in \state$ as
\[
\phi (x) = \begin{cases}  	0 & \quad \text{Do not exercise}
\\
1 &  \quad \text{Exercise}
\end{cases}
\]
Each policy $\phi$ defines a stopping time and   associated average reward, denoted
\[
\tau \eqdef \min \{n \colon \phi (X_n) = 1 \}\,,
\qquad
h_\phi(x) \eqdef \Expect [ \beta^{\tau} r(X_\tau) \mid X_0 = x ].
\]
The optimal policy is expressed as 
\archive{  I am so suspicious of this claim. We must discuss}
\[
\begin{aligned}
\phi^*(x) &= \ind \{ r(x) \ge Q^*(x) \}   
\end{aligned}
\]
The corresponding optimal stopping time that solves the supremum in \eqref{e:JstarFin} is achieved using this policy: $
\tau^* = \min \{n  \colon \phi^*(X_n) = 1 \}$~\cite{tsiroy99}.

The objective here is to find an approximation for $Q^*$ in a parameterized class  $\{Q^\theta \eqdef \theta^\transpose \psi\colon \theta \in \Re^d \}$,  where $\psi\colon \state \to \Re^d$ is a vector of basis functions.  
For a fixed parameter vector $\theta$, the associated value function is denoted  
\begin{equation}
\begin{aligned} 
h_{\phi^\theta}(x) &= \Expect [\beta^{\tau_\theta } r(X_{\tau_\theta }) \mid x_0 = x ]\,,
\\
 \text{\it where}\qquad
\tau_\theta = \min \{n \colon & \phi^\theta (X_n) = 1 \},  \qquad \phi^\theta(x) = \ind \{ r(x)  \ge Q^\theta(x) \}.
\end{aligned}
\label{e:Jtheta}
\end{equation}
The  function $h_{\phi^\theta}$ was estimated using Monte-Carlo in the numerical experiments surveyed below.

\paragraph{Approximations to the Optimal Stopping Time Problem}
To obtain the optimal parameter vector $\theta^*$, in \cite{tsiroy99} the authors apply the Q($0$)-learning algorithm:  
\begin{equation}
\theta_{n+1} = \theta_n + \alpha_{n+1} \psi(X_n) \Big[ \beta \max \big (X_{n+1}(100)  ,  Q^{\theta_n}(X_{n+1}) \big)  -  Q^{\theta_n}(X_{n} )  \Big ]\,,\quad n\ge 0\, .
\label{e:financeTD0}
\end{equation}
This is one of the few parameterized Q-learning settings for which convergence is guaranteed  \cite{tsiroy99}.

In  \cite{choroy06} the authors attempt to improve the performance of the Q($0$) algorithm through the use of the sequence of matrix gains and a special choice for the $\{\alpha_n\}$: 
\begin{equation}
G_n = \left(  \frac{1}{n} \sum_{k=1}^n \psi ( X_k ) \psi^\transpose( X_k )  \right)^{-1} g\,,\qquad \alpha_n = \frac{1}{b+n}\,,
\label{e:BVRG}
\end{equation}
where $g$ and $b$ are positive constants.
The resulting recursion is the $\bfmG$-Q($0$) algorithm:  
\[
\theta_{n+1} = \theta_n + \alpha_n  G_n \psi ( X_n ) \Big [ \beta \max \big (X_{n+1}(100)  ,  Q^{\theta_n}(X_{n+1}) \big)   -   Q^{\theta_n}(X_{n} )  \Big  ].
\]  
Through trial and error the authors find that  $g = 10^{2}$, $b = 10^{4}$ gives good performance. 
These values were also used in the experiments described in the following.   

The limiting matrix gain is given by 
\[
G = \Bigl (  \Expect [   \psi ( X_k ) \psi^\transpose( X_k )  ]  \Bigr )^{-1} g\,,
\]  
where the expectation is in steady-state.  
The asymptotic covariance $\Sigma_{\theta}^G$ is the unique positive semi-definite solution to the Lyapunov equation \eqref{e:GAINLyap}, 
provided all eigenvalues of $G A$ satisfy $\Real(\lambda) < -\half$.
	
The  Zap~Q-learning algorithm for this example is defined by the following recursion:  
\begin{eqnarray} 
	\theta_{n+1} &=& \theta_n - \alpha_{n+1}    \haA_{n+1}^{-1} \psi(X_n)  \Big [ \beta \max \big ( X_{n+1} (100) , Q^{\theta_n}( X_{n+1} )  \big )   -   Q^{\theta_n}(X_{n} )   \Big ], \nonumber
 \\
	\haA_{n+1} &=& \haA_{n}  +  \gamma_n[  A_{n+1} -  \haA_{n} ],
	\qquad  A_{n+1} =\psi ( X_n )  \varphi^\transpose(\theta_n,X_{n+1}) \,,
\label{e:A_SNR2_OS}
\\
	\varphi(\theta_n, X_{n+1}) &=& \beta \psi ( X_{n+1}) \ind{\{Q^{\theta_n}(X_{n+1} ) \geq X_{n+1} (100) \}} - \psi ( X_n ). \nonumber
\end{eqnarray}
It is conjectured that the  asymptotic covariance $\Sigma^*$ is obtained using \eqref{e:SigmaStar}, where the matrix $A$ is the limit of $\haA_{n}$:
\[
A =  \Expect \Bigl [   \psi ( X_n )  \big (\beta \psi ( X_{n+1}) \ind{\{ Q^{\theta^*}(X_{n+1} ) \geq X_{n+1}(100) \}} - \psi ( X_n ) \big)^\transpose  \Bigr ].   
\]

	 \begin{figure}[h] 
\Ebox{.9}{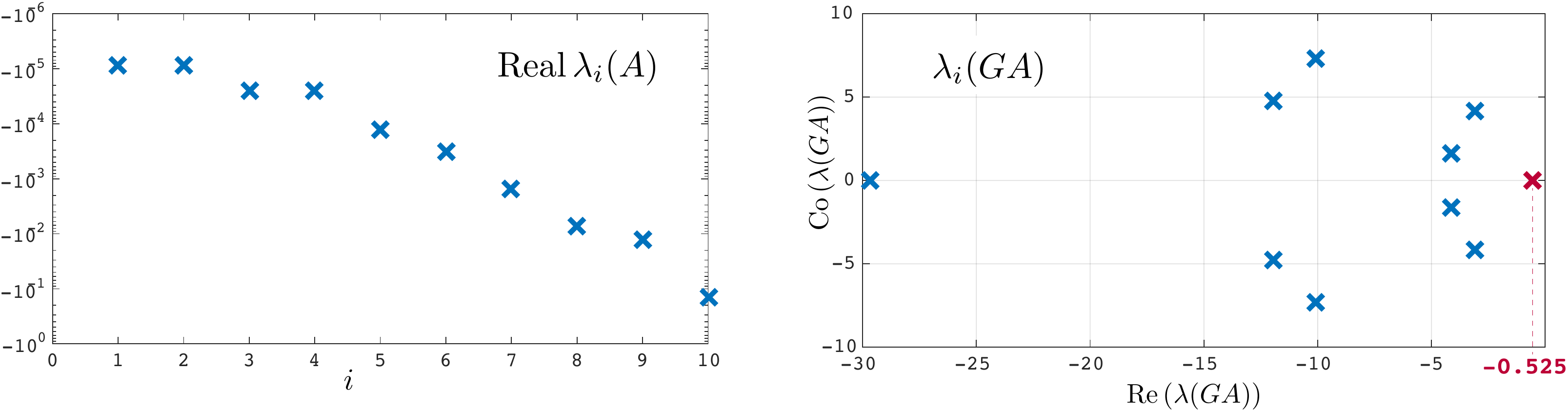}
 \vspace{-.75em}
	\caption{\small Eigenvalues of $A$ and $GA$ for the finance example}
 \vspace{-.57em}
\label{FinancelogEvalAAndGA}
\end{figure}

\paragraph{Experimental Results}

The experimental setting of \cite{tsiroy99,choroy06} is used to define the set of basis functions and other parameters. 
\archive{\today:  We will need to be sure there is no ambiguity about which parameters were used.   Did  
tsiroy99,choroy06 settle on these values?
\\
 were used, with volatility factor $\sigma = 0.02$ and the short term interest rate $\interest = 0.0004$. }
The dimension of the parameter vector $d$ was chosen to be $10$, with the basis functions defined in \cite{choroy06}. 
The objective here is to compare the performances of $\bfmG$-Q($0$) and the Zap-Q algorithms in terms of both parameter convergence,  and with respect to the resulting  average reward   \eqref{e:Jtheta}. 

The asymptotic covariance matrices $\Sigma^*$ and $\Sigma^G_\theta$ were estimated through the following steps:
The matrices $A$ and $G$ were estimated via Monte-Carlo.   Estimation of $A$ requires  an estimate of $\theta^*$; this was taken to be  $\theta_n$, with $n={2 \times 10^6}$,
obtained using the Zap-Q two timescale  algorithm with $\alpha_n \equiv 1/n$ and $\gamma_n \equiv \alpha_n^{0.85}$.  This estimate of $\theta^*$  was also used to estimate the  covariance matrix $\Sigma_\Delta$ defined in \eqref{e:SigmaDelta} using the batch means method.  
The matrices $\Sigma^G_\theta$ and $\Sigma^*$    were then obtained using \eqref{e:GAINLyap} and  \eqref{e:SigmaStar}, respectively.

It was found that the trace of $\Sigma^G_\theta$ was about $15$ times greater than that of $\Sigma^*$.

\paragraph{High performance despite ill-conditioned matrix gain  }


The real part of the eigenvalues of $A$ are shown on a logarithmic scale on the left-hand side of \Fig{FinancelogEvalAAndGA}.   The eigenvalues of the matrix $A$ have a wide spread:   The condition-number  is of the order $10^4$. 
This presents a challenge in applying any method.
In particular, it was found that the performance of any scalar-gain algorithm was extremely poor, even with projection of parameter estimates.

This is a consequence of the fact that the basis functions $\{\psi_i\}$ are nearly linearly dependent.    A better basis should be considered in future work, but the main objective here is to test the new methods in a challenging setting,  and to compare with prior approaches.

\begin{figure}[h]
	\vspace{0in}
	\Ebox{1}{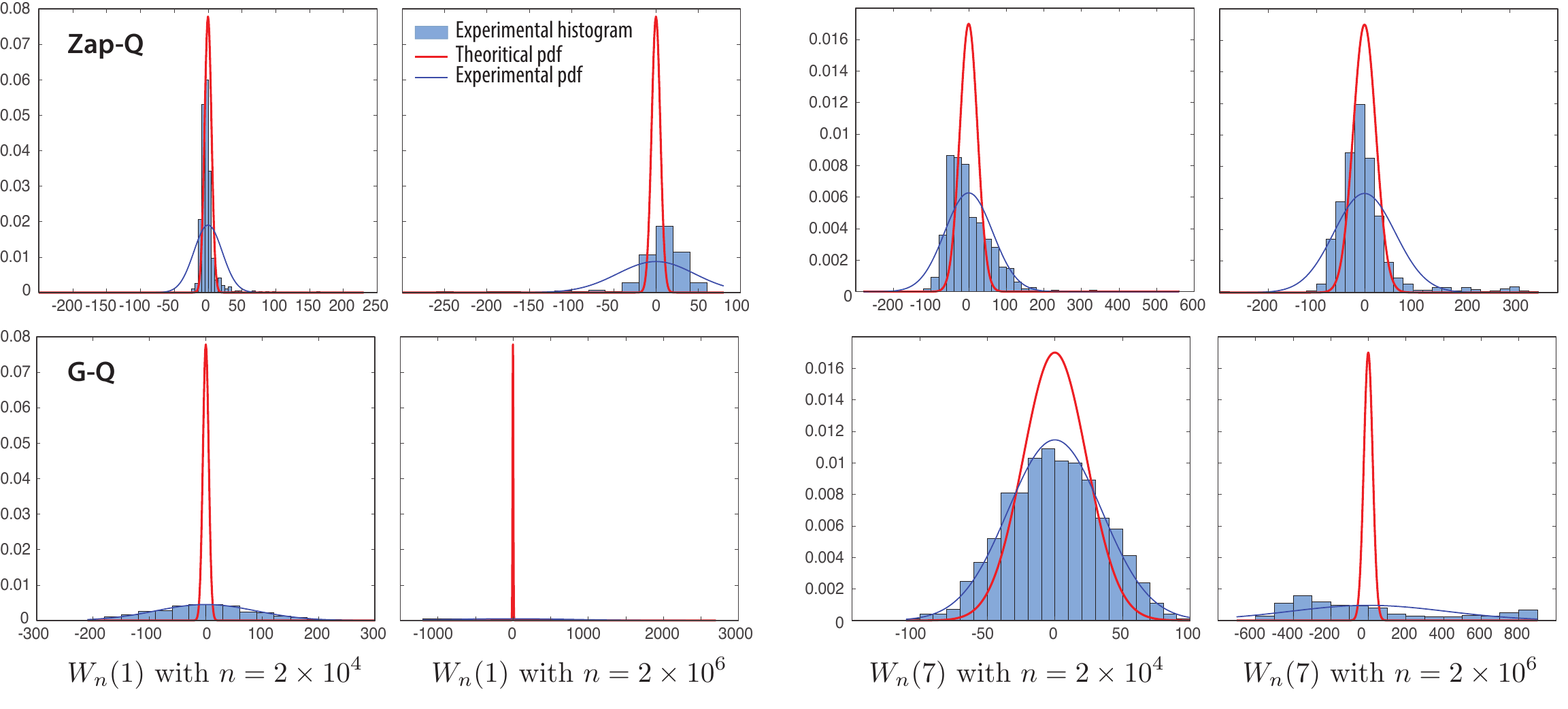}
	\vspace{-1.5em}
	\caption{\small Theoretical and empirical variance for the finance example}
\label{f:FinanceHistZap+BVR}
\end{figure}


In applying the Zap~Q-learning algorithm it was found that  the estimates $\{\haA_{n}\}$ in \eqref{e:A_SNR2_OS} are nearly singular.     Despite the unfavorable setting for this approach,  the performance of the algorithm was much better than any alternative that was tested.  
The upper row of 
\Fig{f:FinanceHistZap+BVR} contains normalized histograms of $\{W_n^i(k) = \sqrt{n} (\theta_n^i(k) - \bartheta_n(k)): 1 \leq i \leq N\}$ for the Zap-Q algorithm. 
The variance for finite $n$ is close to the theoretical predictions based on the asymptotic covariance $\Sigma^*$.
The histograms were generated for two values of $n$, and $k=1, 7$. Of the $d = 10$ possibilities, the histogram for $k = 1$ had the worst match with theoretical predictions, and $k=7$ was the closest.

The eigenvalues corresponding to the matrix $G A$ are shown on the right hand side of  \Fig{FinancelogEvalAAndGA}.  It is found that one of these  eigenvalues is very close to $-0.5$,
and the sufficient condition for $\trace(\Sigma_\theta^G)<\infty$ is barely satisfied. 
It is worth stressing that the finite asymptotic covariance was not a design goal in this prior work.   It is only now on revisiting this paper that we find that the sufficient condition $\lambda<-\half$ is satisfied. 

The lower row of \Fig{f:FinanceHistZap+BVR} contains the normalized histograms of $\{W_n^i(k) = \sqrt{n} (\theta_n^i(k) - \bartheta_n(k)): 1 \leq i \leq N\}$ for the $\bfmG$-Q($0$) algorithm for $n=2 \times 10^4$ and $2 \times 10^6$, and $k = 1,7$,   along with the theoretical predictions based on the asymptotic covariance $\Sigma^G_\theta$. 
 
\archive{\todo Specify the values of $g$ in the discussion and the figure labels}

\archive{\rd{We tried larger initial conditions and it did not work very well.. too many outliers!}We agreed to remove mention of starting point.}


\begin{figure}[h]
	\vspace{0in}
	\Ebox{1}{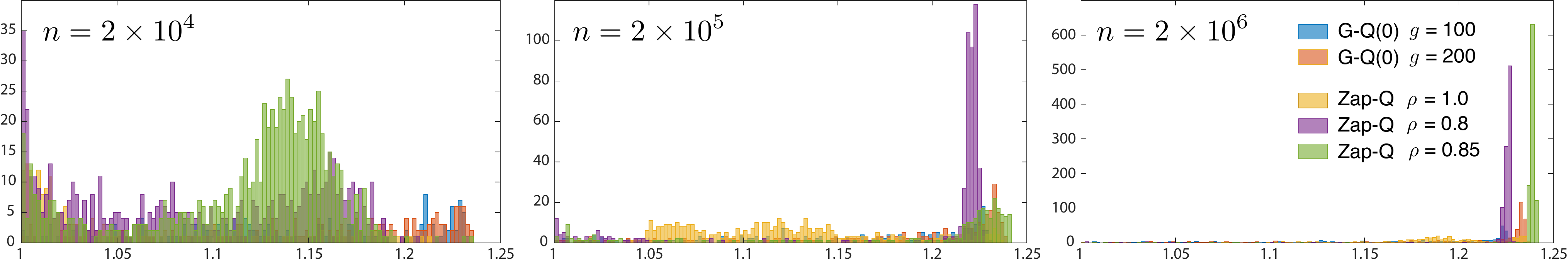}
	\vspace{0in}
	\caption{\small Histograms of the average reward obtained using the $\bfmG$-Q($0$) learning and the Zap-Q-learning, $\gamma_n \equiv \alpha_n^{-\rho} \equiv n{-\rho}$}	 
\label{FinanceVarAna}
\end{figure}

\paragraph{Asymptotic variance of the discounted reward}

Denote  $h_n= h_{\phi}$, with $\phi=\phi^{\theta_n}$.
Histograms of the average reward $h_n(x)$  were obtained for $x(i) = 1$, $1 \leq i \leq 100$, and various values of $n$, based on  $N = 1000$ independent simulations.  The plots shown in \Fig{FinanceVarAna} are based on $n = 2 \times 10^k$,  for $k=4,5,6$.
Omitted in this figure are \textit{outliers}:   values of the reward in the interval $[0,1)$.   
Table~\ref{tbl:financeoutliers} lists the number of outliers for each $n$ and each algorithm.

Recall that the asymptotic covariance of the $\bfmG$-Q($0$) algorithm was not  far from optimal (its trace was about  15 times larger than obtained using Zap Q-learning).  However,  it is observed that this algorithm suffers from much larger outliers. It can also be seen that doubling the scalar gain $g$ (causing the largest eigenvalue of $GA$ to be $\approx -1$) results in slightly better performance.

\begin{table}
\small{
\centering
\begin{subtable}[t]{0.45\linewidth}
\centering
\vspace{0pt}
\begin{tabular}{llll}
\toprule
$n$ & $2 \times 10^4$   & $2 \times 10^5$ & $2 \times 10^6$ \\
\midrule
$G$-Q($0$) $g = 100$ & 827 &  775 &  680   \\  
$G$-Q($0$) $g = 200$ & 824 &  725 &  559   \\
Zap-Q $\rho = 1$     & 820 &  541 &  625   \\
Zap-Q $\rho = 0.8$   & 236 &  737 &  61    \\
Zap-Q $\rho = 0.85$  & 386 &  516 &  74    \\		
\bottomrule
\end{tabular}
\vspace{0.05in}
\caption{$h_n(x) <1$}
\vspace{0.12in}
\label{tbl:three}
\end{subtable} }
\small{
\centering
\begin{subtable}[t]{0.45\linewidth}
\centering
\vspace{0pt}
\begin{tabular}{llll}
\toprule
$n$ & $2 \times 10^4$   & $2 \times 10^5$ & $2 \times 10^6$ \\
\midrule
$G$-Q($0$) $g = 100$ & 811  & 755  & 654   \\  
$G$-Q($0$) $g = 200$ & 806  & 706  & 537   \\
Zap-Q $\rho = 1$     & 55   &  0   &  0   \\
Zap-Q $\rho = 0.8$   & 0    &  0   &  0    \\
Zap-Q $\rho = 0.85$  & 0    &  0   &  0    \\		
\bottomrule
\end{tabular}
\vspace{0.05in}
\caption{$h_n(x) < 0.95$}\label{tbl:four}
\vspace{0.12in}
\end{subtable} }
\small{
\centering
\begin{subtable}[t]{0.45\linewidth}
\centering
\vspace{0pt}
\begin{tabular}{llll}
\toprule
$n$ & $2 \times 10^4$   & $2 \times 10^5$ & $2 \times 10^6$ \\
\midrule
$G$-Q($0$) $g = 100$ & 774  & 727  & 628   \\  
$G$-Q($0$) $g = 200$ & 789  & 688  & 525   \\
Zap-Q $\rho = 1$     & 4    &  0   &  0   \\
Zap-Q $\rho = 0.8$   & 0    &  0   &  0    \\
Zap-Q $\rho = 0.85$  & 0    &  0   &  0    \\		
\bottomrule
\end{tabular}
\vspace{0.05in}
\caption{$h_n(x) < 0.75$}\label{tbl:one}
\vspace{0.12in}
\end{subtable} } \hfill
\small{
\centering
\begin{subtable}[t]{0.45\linewidth}
\centering
\vspace{0pt}
\begin{tabular}{llll}
\toprule
$n$ & $2 \times 10^4$   & $2 \times 10^5$ & $2 \times 10^6$ \\
\midrule
$G$-Q($0$) $g = 100$ & 545  & 497  & 395   \\  
$G$-Q($0$) $g = 200$ & 641  & 518  & 390   \\
Zap-Q $\rho = 1$     & 0   &  0   &  0   \\
Zap-Q $\rho = 0.8$   & 0    &  0   &  0    \\
Zap-Q $\rho = 0.85$  & 0    &  0   &  0    \\		
\bottomrule
\end{tabular}
\vspace{0.05in}
\caption{$h_n(x) <0.5$}\label{tbl:two}
\vspace{0.12in}
\end{subtable} }
\caption{Outliers observed in $N = 1000$ runs. 
Each table represents the number of runs which resulted in an average reward below a certain value }\label{tbl:financeoutliers}
\end{table}

\section{Conclusions}
\label{s:conc}

Watkins' Q-learning algorithm is elegant, but subject to two common and valid constraints:  it can be very slow
to converge, and it is not obvious how to extend this approach to obtain a stable algorithm in non-trivial parameterized settings.   This paper addresses both concerns with the new Zap~Q($\lambda$) algorithms that are motivated by asymptotic theory of stochastic approximation.

There are many avenues for future research.   It would be valuable to find an alternative to Assumption Q3 that is readily verified.   Based on the ODE analysis, it seems likely that the conclusions of \Theorem{t:ZAP} hold without this additional assumption.  No theory has been presented here for  non-ideal parameterized settings.   It is conjectured that conditions for stability of Zap~Q($\lambda$)-learning will hold under general conditions.   Consistency is a more challenging problem and is a focus of current research.

In terms of algorithm design,  it is remarkable to see how well the scalar-gain algorithms perform, provided projection is employed and the condition number of $A$ is not too large.   It is possible to estimate the optimal scalar gain based on estimates of the matrix $A$ that is central to this paper.  How to do so without introducing high complexity is an open question.

On the other hand, the performance of RPJ averaging is unpredictable.  In many experiments it is found that  the asymptotic covariance is a poor indicator of finite-$n$ performance when using this approach.    There are many suggestions in the literature for improving this technique (see discussion after Theorem~3 of \cite{moubac11}) .   

The results in this paper suggest new approaches that we hope will simultaneously
\begin{romannum}
\item Reduce complexity and potential numerical instability of matrix inversion,
\item Improve transient performance, and 
\item Maintain optimality of the asymptotic covariance
\end{romannum}

\archive{
\rd{AD: New algorithm that eliminates the invertibility issue. Q($\lambda$) learning. }
\rd{SPM: Choosing a time varying scalar gain $g_t$ recursively }
\bl{AD: Scalar gain analysis. And then maybe concentration bounds. I am curious on how the existing dimension independent algorithms depend on the A matrix. Do they somehow implicitly manage to satisfy the eigenvalue bound? Or maybe they don't use the common $1/n$ type of step-sizes?}
}
	

\bibliographystyle{abbrv}
	
\def\cprime{$'$}\def\cprime{$'$}

\null  

\appendix

\section{Appendices}

\subsection{Asymptotic covariance for Markov chains} 
\label{s:AsymCov}

As an illustration of the Lyapunov equation \eqref{e:Lyap} that is solved by the asymptotic covariance, consider the error recursion for a one-dimensional version of \eqref{e:linearSA} in which $A_n\equiv 1$,  and the algorithm is scaled by a gain parameter $g>0$: 
\[
\tiltheta(n+1) = \tiltheta(n) + \frac{g}{n+1} \Bigl( -\tiltheta(n) + \Delta(n+1)\Bigr)
\]
When $g=1$ this is a standard Monte-Carlo average.    For general $g>\half$ the Lyapunov equation admits the solution:
\[
\Sigma_\theta =  \Bigl( \frac{g^2}{2g-1}  \Bigr)\sigma^2_\Delta
\]
This grows without bound as $g\to\infty$ or $g\downarrow \half$, as illustrated below: 
\Ebox{.65}{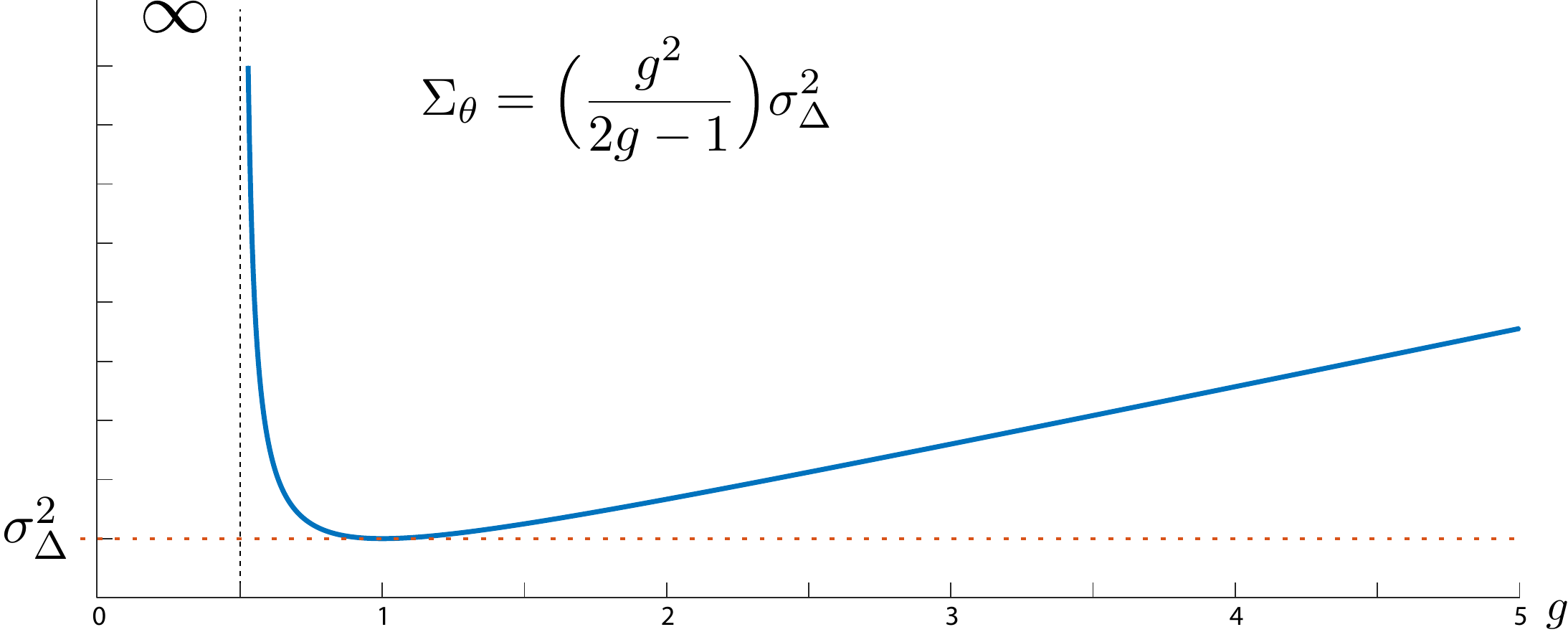}

Conditions to ensure that the covariance is infinite are presented in the following:

\begin{proposition}
\label{t:CovInfinite}
Consider the linear recursion 
\begin{equation}
\theta_{n+1} = \theta_n +  \frac{G}{n+1} \bigl[A \theta_n - b +\Delta_{n+1} \bigr]
\label{e:linearSAwhite}
\end{equation} 
in which $\{\Delta_n\}$ is a martingale difference sequence satisfying 
$
\Sigma_{\Delta_n} = \Cov(\Delta_n)\to\Sigma_\Delta$ as $n\to\infty$.   
\begin{romannum}
	\item Suppose that $(\lambda,v)$ is an eigenvalue-eigenvector pair satisfying 
	$GAv =\lambda v$,  $\Real(\lambda)\ge -1/2$, and $v^\dagger G \Sigma_\Delta G^\transpose v >0$.    Then,
\[
	\lim_{n\to\infty } \sigma^2_n \eqdef  \lim_{n\to\infty }  n 
	\Expect[ |v^\transpose \tiltheta_n|^2] =\infty\, . 
\]
	\item Suppose that all the eigenvalues of $GA$ satisfy $\Real(\lambda) < -1/2$, then the asymptotic covariance $\Sigma _\theta^G$ is finite, and is obtained as a solution to the Lyapunov equation \eqref{e:GAINLyap}.
\end{romannum}
\end{proposition}
\begin{proof}
Define $Z_n \eqdef \sqrt{n} \tiltheta_n$ and $\Sigma_n \eqdef \Expect [ Z_n Z_n^\transpose ]$. 
A standard Taylor-series approximation of $Z_n$ results in the following recursive definition of $\Sigma_n$:
\begin{equation}
\Sigma _n  =  \Sigma_{n-1} 
				+ \frac{1}{n} \Bigl\{  (GA+\half I) \Sigma (n-1)   
					+  \Sigma (n-1) (GA+\half I)^\transpose  +  \Sigma_{\Delta_n}\Bigr\} + o\Bigl( \frac{1}{n^2} \Bigr)
\label{e:Sigmanrec}
\end{equation}
The assumptions of part (i) of the proposition implies:
\[
\begin{aligned}
\sigma^2_n &= \sigma^2_{n-1} + \frac{1}{n} \Bigl\{  2(\Real(\lambda)+1/2)   \sigma^2_{n-1} +v^\dagger  \Sigma_{\Delta_n} v \Bigr\} + O\Bigl( \frac{1}{n^2} \Bigr) \\
& \geq \sigma^2_{n-1} +  \frac{1}{n} v^\dagger \Sigma_{\Delta_n} v  + O ( \frac{1}{n^2}  )
\end{aligned}
\]
 and therefore $\sigma^2_n \ge v^\dagger  \Sigma_\Delta  v \,\log n + O ( 1  )$, implying the result in part (i) of the proposition.
 Under the assumptions of part (ii), \eqref{e:Sigmanrec} implies that $\Sigma _n \to \Sigma _\theta^G$, where $\Sigma _\theta^G$ is obtained as a solution to the Lyapunov equation \eqref{e:GAINLyap}.
\end{proof}

\subsection{Proof of \Theorem{t:Qinfinite}}
\label{s:InfAsymVar}

\def\lambdaPF{\lambda_{\text{\tiny PF}}}

Recall that $\diagpie$ denotes the $d\times d$ diagonal matrix with entries $\diagpie(k,k)=\pie(x^k,u^k)$,  and   the matrix   $A$ is a function of $q$:
\[
A(q) = -\diagpie [I - \beta   P S_{\phi^{q}}]\,.
\] 
In a neighborhood of  $ \theta^* = Q^*$, the operator $S_{\phi^{q}}$ coincides with $S_{\phi^*}$, and we denote $A_\infty = A(\theta^*)$.

To prove the result we construct an eigenvector $v\in\Re^d$ for $A_\infty$ whose entries are strictly positive.  Next, under the assumptions of the proposition, we show that the corresponding eigenvalue satisfies $\lambda=\Real(\lambda)\ge -1/2$, and the result then follows from \Prop{t:CovInfinite} combined with \eqref{e:SigmaDeltaQ}.

Recall from \Lemma{t:PSphi} that $ P S_{\phi^*}$ is a $d\times d$ transition matrix, and so 
is the following 
\[
T \eqdef (1-\beta)
[I-\beta P S_{\phi^*} ]^{-1} 
= (1-\beta) \sum_{n=0}^\infty \beta^n [P S_{\phi^*}]^n 
\]
The construction of an eigenvector is via the representation $-A_\infty^{-1}  = (1-\beta)^{-1} T \diagpie^{-1}$.  Since this is a positive and irreducible matrix, we can apply Perron-Frobenius theory to conclude that there is a maximal eigenvalue $\lambdaPF>0$ and an everywhere positive eigenvector $v$ satisfying
\[
-A_\infty^{-1} v = \lambdaPF v
\]
The Perron-Frobenius eigenvalue coincides with the spectral radius of $A_\infty^{-1} $:
\[
\lambdaPF=
\lim_{n\to\infty} \| (A_\infty^{-1} )^n\|^{1/n} = 
\frac{1}{1-\beta} 
\lim_{n\to\infty} \| (T \diagpie^{-1} )^n\|^{1/n}   \ge     \frac{1}{1-\beta}   \min_{x,u} \frac{1}{\pie(x,u)} 
\]
The vector $v$ is also an eigenvector for $A_\infty$ with associated eigenvalue 
\[
\lambda = - \lambdaPF ^{-1}  \ge -(1-\beta) \max_{x,u}  \pie(x,u)\,.
\]
Thus, $\lambda \ge - \half$ under the assumptions of the proposition.  
\qed

\subsection{Proof of \Theorem{t:ZAP}}
\label{s:ZapQProof}

The remainder of the Appendix is devoted to the proof of \Theorem{t:ZAP}.

 \Lemma{t:ODEfish} is used to establish the ``ODE friendly'' property for the error sequences appearing in the ODE approximations.

\begin{proof}[Proof of \Lemma{t:ODEfish}]
Let  $H\colon\state\times\U\to\Re^{d\times d}$ solve Poisson's equation: 
\[
\Expect[H(X_{n+1}, U_{n+1})  -  H(X_{n}, U_{n})  \mid \clF_n]  = \diagpie - \Psi_n  
\]
with $\clF_n \eqdef \sigma(X_k,U_k : k\le n)$.	The following representation is immediate:
\[
\begin{aligned}
\Psi_n  P S_{\phi_{n}}  &=  \diagpie P S_{\phi_{n}}   - \Expect[H(X_{n+1}, U_{n+1})  -  H(X_{n}, U_{n})  \mid \clF_n]  P S_{\phi_{n}}
\\
&= \diagpie P S_{\phi_{n}}    +  \widehat{\clT}^\ominus_{n+1} -\widehat{\clT}_n +  \Delta^\Psi_{n+1} 
\end{aligned}
\]
where 
\[
\widehat{\clT}_n = -  H(X_{n}, U_{n})   P S_{\phi_{n}}   \,,\quad \widehat{\clT}^\ominus_{n+1}= -  H(X_{n+1}, U_{n+1})   P S_{\phi_{n}}  \,,
\]
and the final term is the martingale difference sequence:
\[ 
\Delta^\Psi_{n+1}  =  \Bigl( H(X_{n+1}, U_{n+1})  
- \Expect[H(X_{n+1}, U_{n+1})     \mid \clF_n]  \Bigr) P S_{\phi_{n}}\, .
\]
The telescoping sequence is thus,
\[
\clT_{n+1} - \clT_n =
- H(X_{n+1}, U_{n+1}) P S_{\phi_{n+1}}    +  H(X_{n}, U_{n}) P S_{\phi_{n}} 
\]
and
\[
\epsy_{n+1} = \widehat{\clT}^\ominus_{n+1} - \clT_{n+1} = -  H(X_{n+1}, U_{n+1})   P S_{\phi_{n}} + H(X_{n+1}, U_{n+1}) P S_{\phi_{n+1}},
\]
which satisfies \eqref{e:Friendly3} under Assumption~Q3.
\end{proof}

Recall the ``$o(1)$'' notation used in \eqref{e:ODEgeneral} is interpreted in a functional sense.  It is a function of two variables:   $o(1) = \clE_{T_0, T_0+t}$, satisfying for each $T>0$,
\[
\lim_{T_0\to\infty}  \sup_{0\le t\le T}   \| \clE_{T_0, T_0+t} \| =0
\]
The notation $ \clE_{T_0, T_0+t}=o(v_{T_0,t})$ for a  vector-valued function of $(T_0,t)$ has an analogous interpretation:
\[
\lim_{T_0\to\infty}  \sup_{0\le t\le T}   \frac{ \|  \clE_{T_0, T_0+t} \|}{ 1+\|v_{T_0,t}\|} =0
\] 
This notation will also be used in the standard setting in which the variable $t$ is absent.  In particular,  $ \clE_{T_0} = o(1)$ simply means that   $ \clE_{T_0} \to 0$ as $T_0\to\infty$.

Recall the definitions of  $\alpha_n$, $\gamma_n$, and  $t_n$ in \eqref{e:GAINS} and \eqref{e:tn};
$g_t$ was defined to be a piecewise linear function with $g_t = \gamma_n/\alpha_n$ when $t=t_n$ for some $n$, and extended to all $t\in\Re_+$ by   linear interpolation. Recall that  $g_t \approx \exp((1-\rho)t)$ for large $t$ under \eqref{e:GAINS}. For   fixed $T_0 > 0$, let $\clG_{T_0,(\varble)}$ denote the cumulative distribution function on the interval $[0,T_0]$:
\begin{equation}
\clG_{T_0,\tau} \eqdef \exp\bigl(-\int_\tau^{T_0} g_r\, dr\bigr)\,, \quad  \ 0\le \tau\le T_0\,
\label{e:clG}
\end{equation} 
This CDF defines a probability measure on the interval $[0,{T_0}]$ with the following properties:
\begin{lemma}
\label{t:clGProperty}
The probability measure associated with $\clG_{{T_0},(\varble)}$ has a density on $(0,{T_0}]$ and a single point mass at zero. 
Its total mass is concentrated near ${T_0}$: For any $\kappa>0$,  
\[
\begin{aligned}
\int_0^{T_0}  |{T_0}-\tau|  e^{  \kappa {T_0} }  d [\clG_{{T_0},\tau} ] & = o(1)
\\
\int_0^{T_0}  e^{  \kappa|{T_0}-\tau| }  d [\clG_{{T_0},\tau} ] & = 1 +o(1),\qquad {T_0}\to\infty
\end{aligned}
\]
\qed
\end{lemma}

An associated pmf on the set of policy indices is defined as follows:
\begin{equation}
\barmu_{T_0}(k) \eqdef \int_0^{T_0}   \ind\{  \phi^{\barq_\tau} =  \phi^{(k)}   \} \,  d [\clG_{{T_0},\tau} ] \,,
\quad 1\le k\le  \nphi
\label{e:barmu}
\end{equation}

\begin{lemma}
\label{t:haA_Bdd_And_ODE_sol}
The linear systems representation  \eqref{e:preODE_Zap_A} holds:
\[
\barclA_{T_0 +t} = \barclA_{T_0}   
-   \int_{T_0}^{T_0+t }\bigl\{    \diagpie  [I-\beta P S_{\phi^{\barq_\tau}} ]  + \barclA_\tau    \bigr\} g_\tau  \, d\tau   +  o(1)\,, \quad  T_0\to\infty \, .
\]
Furthermore,
\begin{romannum}
\item 
There exists $T_\bullet \geq 0$ satisfying $T_\bullet < \infty$ a.s., 
and for which
the processes
 $\{ \barclA_t: t \geq T_\bullet \}$ and $\{ \barclA_t ^{-1} : t \geq T_\bullet \}$ are bounded:
\begin{equation}
b_A  \eqdef   \sup_{t\ge T_\bullet}  \{ \| \barclA_t \| + \|\barclA_t ^{-1}\| \} <\infty\,,\quad a.s..
\label{e:bA}
\end{equation}
\item
For each $t,T_0 \ge 0$, 
\begin{equation*}
\barclA_{T_0 + t} = \clG_{T_0+t,T_0}\barclA_{T_0 }
- 
\int_{T_0}^{T_0 + t}      \diagpie  [I-\beta P S_{\phi^{\barq_\tau}} ] \, d[ \clG_{T_0+t,\tau}]  + o(1)
,\qquad T_0\to \infty\,
\end{equation*}
\item  With   $ \partialQstar_{\mu} $ defined in \eqref{e:partiakQmu}, and $\barmu_t$ defined in \eqref{e:barmu}, the following representation 
holds:
\begin{equation}
\partialQstar_{\barmu_t}^{-1} 
 = \int_0^t     [I-\beta P S_{\phi^{\barq_\tau}} ]  \,  d [\clG_{t,\tau} ]. 
\label{e:partialQstarbarmu}
\end{equation}
\item
The following approximations hold:
\begin{equation}
\begin{aligned}
\barclA_t  =  - \diagpie \partialQstar_{\barmu_t}^{-1}  + o(1) \quad \text{\it and}\quad
- \barclA_t^{-1} \diagpie  = \partialQstar_{\barmu_t}  + o(1)
,\qquad t\to \infty\,.
\end{aligned}
\label{e:AtpartialQstarbarmut}
\end{equation}
\end{romannum}
\end{lemma}
\begin{proof}
The representation  \eqref{e:preODE_Zap_A} directly follows from Lemmas~\ref{t:pre-ODE} and~\ref{t:ODEfish}.
\archive{\rd{to discuss $T_0 > 0$}
\\
I guess we need to discuss.  Remember to delete all these comments once this is settled [to me, it is all settled!]
\\
\rd{\Large Is this settled??
\\
\bl{AD: It is all fine.. I just wanted to know why $>0$.. I will archive it for now, since it doesn't really affect anything..}}
}
For $T_0>0$, the solution to this linear time varying system is the sum of three terms:
\begin{equation}
\begin{aligned}
\barclA_{T_0 + t} &=\clI_{T_0}^A(t)   -    \clI_{T_0}^B(t)  -  \clI_{T_0}^E(t) 
\\[.2cm]
&  \eqdef
\clG_{T_0+t,T_0} \barclA_{T_0} 
- 
\int_{T_0}^{T_0 + t} \diagpie \partialQstar_{k_\tau}^{-1}\, d[\clG_{T_0+t,\tau}]  
-
\int_{T_0}^{T_0 + t} \clG_{T_0+t,\tau}  \, d \clE_{T_0, T_0+\tau}^A\,,
\end{aligned}
\label{e:SolpreODE_Zap_A}
\end{equation}
in which  $\clE_{T_0, T_0+\tau}^A = o(1)$,  $T_0\to\infty$ and
\[
\partialQstar_{k} \eqdef [I-\beta P S_{\phi^{(k)}}]^{-1} \,, \quad 1 \leq k \leq \nphi
\]
and for each $t \geq 0$, $k_t$ is the integer satisfying $\phi^{(k_t)} = \phi^{\barq_t}$.    

We begin with a proof of boundedness of $\{ \barclA_t: t \geq 0 \}$,   considering the three terms in \eqref{e:SolpreODE_Zap_A} separately.  
The first term $\clI_{T_0}^A(t) $ vanishes as $t\to\infty$ for each fixed $T_0$.   The second term admits the uniform bound:  
\begin{equation}
\| \clI_{T_0}^B(t)  \|\le b_\clA \int_{T_0}^{T_0 + t}    \, d[\clG_{T_0+t,\tau} ]  \le  b_\clA
\label{e:clIBBd}
\end{equation}
in which $b_\clA\eqdef \max_k \|\diagpie  [I-\beta P S_{\phi^{(k)}} ]\|$ is an upper bound on    $\|\diagpie \partialQstar_{k_\tau}^{-1}\| $.  
It is shown next that the final term $ \clI_{T_0}^E(t) $ converges to zero as $T_0\to\infty$. 
Applying integration by parts:
\[
\begin{aligned}
\clI_{T_0}^E(t)
&   =  
\clG_{T_0+t,T_0+t} \clE_{T_0, T_0+t}^A  
-  \clG_{T_0+t,T_0} \clE_{T_0, T_0}^A  
- \int_{T_0}^{T_0 + t}  \clE_{T_0, \tau}^A  \, d[\clG_{T_0+t,\tau}    ]
\\
&=   \clE_{T_0, T_0+t}^A -  \int_{T_0}^{T_0 + t}  \clE_{T_0, \tau}^A  g_\tau \clG_{T_0+t,\tau}  \, d\tau
\end{aligned}
\]
where the second equation used $ \clG_{T_0+t,T_0+t}=1$ and  $\clE_{T_0, T_0}^A  =0$.  Using this identity and the same arguments used to bound $ \clI_{T_0}^B(t)$ then gives 
\begin{equation}  
\|  \clI_{T_0}^E(t) \| \le 2   \|\clE_{T_0, T_0+t}^A \|   =o(1)\,,\quad T_0\to\infty\, .
\label{e:clIEBd}
\end{equation} 
Using \eqref{e:clIBBd} and \eqref{e:clIEBd} in \eqref{e:SolpreODE_Zap_A} establishes boundedness of $\{ \barclA_t \}$:
\[
\limsup_{r\to\infty}\| \barclA_r\| = \limsup_{T_0\to\infty} \sup_{0\le t\le T} \| \barclA_{T_0 + t} \|   \le  b_\clA
\]
and hence also boundedness of the sequence $\{\haA_n: n \geq 0\}$.
\archive{\bl{Note to AD from Kauai:  happy now?
	\\
	\huge
VERY! :D}
}

The representation \eqref{e:SolpreODE_Zap_A}  along with \eqref{e:clIEBd} also implies the evolution equation in part (ii) of the lemma. It is shown next that \eqref{e:AtpartialQstarbarmut} holds.  This will imply that $\{\barclA_t^{-1}: t \geq T_\bullet\}$ is bounded, which will complete the proof of the lemma.

The approximation \eqref{e:SolpreODE_Zap_A} was obtained for $T_0>0$.   \Lemma{t:clGProperty} implies that we can let $T_0\downarrow 0$  in this bound to obtain,
\begin{equation}
\begin{aligned}
\barclA_t &= \clG_{t,0}\barclA_0
- 
\int_{0}^{t}      \diagpie  [I-\beta P S_{\phi^{\barq_\tau}} ] \, d[\clG_{t,\tau} ]  + o(1)
,\qquad t\to \infty
\\
&=  - 
\int_{0}^{t}      \diagpie  [I-\beta P S_{\phi^{\barq_\tau}} ] \, d[\clG_{t,\tau} ]  + o(1)
,\qquad t\to \infty\,.
%
\end{aligned}
\label{e:haA_ODE_sol}
\end{equation}
Next, based on the definition of $\barmu_t$ in \eqref{e:barmu}, the representation \eqref{e:partialQstarbarmu} for $\partialQstar_{\barmu_t}^{-1}$ is obtained:
\begin{equation*}
\begin{aligned}
\partialQstar_{\barmu_t}^{-1} 
& \eqdef \sum_{i = 1}^{\ell_\phi} \barmu_t(i)  [I-\beta P S_{\phi^{(i)}}]
\\
& = \int_0^t \sum_{i = 1}^{\ell_\phi}  \ind\{  \phi^{\barq_\tau} =  \phi^{(i)}   \} [I-\beta P S_{\phi^{(i)}}] \,  d [\clG_{t,\tau} ]
\\
& = \int_0^t     [I-\beta P S_{\phi^{\barq_\tau}} ]  \,  d [\clG_{t,\tau} ]. 
\end{aligned}
\end{equation*}
Combining the above result with \eqref{e:haA_ODE_sol},  \eqref{e:AtpartialQstarbarmut} is obtained.
This also implies that $\{\barclA^{-1}_t : t \geq T_\bullet \}$ is bounded.
\end{proof}

\begin{lemma}
\label{t:barqtexpbd}
The linear systems representation holds:
\begin{equation}
\barq_{T_0 +t} = \barq_{T_0}
-   \int_{T_0}^{T_0+t }  \barclA_\tau^{-1}  \diagpie  \bigl\{c - \barqcost_\tau    + o(\|\barq_{\tau} \|) \bigr\}  \, d\tau   + o(1)   \,, \quad  T_0\to\infty \,
\label{e:prepreODE_Zap_q}
\end{equation}
Furthermore, for a constant $b_q  < \infty$, and each $t,T_0 \ge 0$, 
\begin{subequations}
\begin{align}
\|\barq_{T_0 + t} - \barq_{T_0} \| &  \le b_q [1 + \| \barq_{T_0} \| ]  t e^{b_q t}  +  o(1)  
,\qquad T_0 \to \infty\,
\label{e:barqtexpbd}
\\
\| \barq_{t} \|   &  \le  b_q [1 + \| \barq_{0} \| ]  t e^{b_{q} t}   
,\qquad t \to \infty\,
\label{e:barqtexpbd2}
\end{align}
\end{subequations}
\end{lemma}
\begin{proof}
Boundedness of $\{\haA_n^{-1}: n \geq n_\bullet\}$ in \Lemma{t:haA_Bdd_And_ODE_sol} (i) implies that $\haG_n^* = - \haA_n^{-1}$ for $n \geq n_\bullet$ in recursion \eqref{e:SNR2linearSA_preODE}. Representation  \eqref{e:prepreODE_Zap_q} then follows from Lemmas~\ref{t:pre-ODE} and~\ref{t:ODEfish}.
The evolution equation \eqref{e:prepreODE_Zap_q} implies:
\[
\epsy_{T_0}^{\barq} (t)  \eqdef \barq_{T_0 + t} - \barq_{T_0}  =   -   \int_{T_0}^{T_0+t }  \barclA_\tau^{-1}  \diagpie  \bigl\{c - \barqcost_\tau +  o(\|\barq_{\tau} \|)  \bigr\}  \, d\tau   +  o(1)
,\qquad T_0 \to \infty\,
\]
Denoting
\begin{equation}
b_S =  \max_k  \| I - \beta P S_{\phi^{(k)}} \| \,,   \qquad b_c =  \max_k   c( x^k , u^k )
\label{e:bS}
\end{equation}
and applying \Lemma{t:haA_Bdd_And_ODE_sol} (i),
\[
\begin{aligned}
\|  \epsy_{T_0}^{\barq} (t)  \|   
& \le   \Big \| \int_{T_0}^{T_0+t }  \barclA_\tau^{-1}  \diagpie  \big \{c - [I - \beta P S_{\phi^{\barq_\tau}}] \barq_\tau  +  o(\|\barq_{\tau} \|)   \big \}  \, d\tau \Big \|  +  o(1)
\\
&
 \le   b_A b_c t   +  (b_Ab_S+1)  \int_{0}^{t }  \big\| -  \barq_{T_0 + \tau}   +  \barq_{T_0}   -  \barq_{T_0} \big\|  \, d\tau   +  o(1)
\\
& \le \big (    b_Ab_c   +     (b_A b_S +1) \| \barq_{T_0} \|  \big )  t  +   (b_A b_S +1)  \int_{0}^{t }   \|  \epsy_{T_0}^{\barq} (\tau)   \|  \, d\tau  +  o(1)
,\qquad T_0 \to \infty\,
\end{aligned}
\] 
The factor $1$ before the integral in the second inequality comes from choosing $T_0$ large enough, so that $o( \| \barq_{T_0 + \tau}  \| )     \leq   \| \barq_{T_0 + \tau}  \| $.
Equation \eqref{e:barqtexpbd} is then obtained by applying the Gr\"{o}nwall lemma \cite{bor08a}.

Applying  the triangle inequality to \eqref{e:barqtexpbd}, and choosing $\barT_0$ such that $\|o(1) \|\le 1$  for all $T_0 \geq \barT_0$  gives
\begin{equation}
\| \barq_{T_0 + t} \|   \le  \| \barq_{T_0} \|  +   b_q [1 + \| \barq_{T_0} \| ]  t e^{b_{q} t} + 1
, \qquad T_0 \geq \barT_0   
\label{e:qbarT0t}
\end{equation}
In particular, the above inequality is true for $T_0 = \barT_0$.
Furthermore, the following holds under the assumption that the sequence $\{\haA_n\}$ is projected prior to inversion:
 For a constant $b_{\barT_0} < \infty$,
\begin{equation}
\| \barq_{\barT_0} \|   \le  b_{\barT_0} \| \barq_{0} \|
\label{e:qbarT0}
\end{equation}
The bound \eqref{e:barqtexpbd2}
is obtained on
combining \eqref{e:qbarT0t} and \eqref{e:qbarT0}.
\end{proof}

\begin{lemma} 	
\label{t:Aq1}
For each $t\geq0$, the following holds:
\begin{subequations}
\begin{align}
\| {\partialQstar^{-1}_{\barmu_t}} \barq_t - \barqcost_t \|  & = o(1),
\label{e:MismatchExt1}
\\
\| \barq_t 	+  \barclA_t^{-1} \diagpie \barqcost_t \|  & = o(1+\|q_t\|) 
,\qquad t\to \infty\,
\label{e:MismatchExt2}
\end{align}
\end{subequations}
\end{lemma}
\begin{proof}
Equation \eqref{e:partialQstarbarmu} of \Lemma{t:haA_Bdd_And_ODE_sol} implies 
the following representation:
\[
\begin{aligned}
\partialQstar_{\barmu_t}^{-1} \barq_t
& = \int_0^t       [I-\beta P S_{\phi^{\barq_\tau}} ] \barq_t  \,  d [\clG_{t,\tau} ]
\\
& =  \int_0^t   \barqcost_\tau  \,  d [\clG_{t,\tau} ]    +   \int_0^t  [I-\beta P S_{\phi^{\barq_\tau}} ]  (\barq_t -  \barq_\tau )  \,  d [\clG_{t,\tau} ].
\end{aligned}
\]
Subtracting $ \barc_t$ from each side and taking norms gives the bound,
\begin{equation}
\| \partialQstar_{\barmu_t}^{-1} \barq_t  -  \barc_t \|
\le   \int_0^t   \| \barqcost_\tau  -  \barqcost_t \|  \,  d [\clG_{t,\tau} ]   
 +  b_S \int_0^t \| \barq_t -  \barq_\tau \| \,  d [\clG_{t,\tau} ]  
\label{e:LA5p1}
\end{equation}
\Lemma{t:clQstar} implies that the mappings $\Qstar$ and $\Qstar^{-1}$ are Lipschitz: 
\archive{spm:  read and then archive this:
\\
To do: 
Study the mapping $\Qstar^{-1}$.   Is it convex?  The max of smooth convex functions?  What do we know about it?
}
 for a constant $b_\Qstar$,
\begin{equation*}
\| \barqcost_\tau - \barqcost_t \| = \| \Qstar^{-1} (\barq_\tau  -  \barq_t) \| \leq b_\Qstar \| \barq_\tau - \barq_t \|
\label{e:barqcostLip}
\end{equation*}
Substituting into \eqref{e:LA5p1} gives
\[
\begin{aligned}
\| \partialQstar_{\barmu_t}^{-1} \barq_t  -  \barc_t \|
& \leq (  b_S  +   b_\Qstar   )   \int_0^t    \|  \barq_t -  \barq_\tau   \|  \,  d [\clG_{t,\tau} ]
\\
& \leq (  b_S  +   b_\Qstar   ) \int_0^t  b_q [ 1 + \| \barq_{\tau} \|  ]  (t - \tau)  e^{b_q (t - \tau)}  \,  d [\clG_{t,\tau} ]  +  o(1)
\\
& = (  b_S  +   b_\Qstar   ) b_q \int_0^t    \| \barq_{\tau} \|  (t - \tau)  e^{b_q (t - \tau)}  \,  d [\clG_{t,\tau} ]   +  o(1)
,\qquad t \to \infty\,
\end{aligned}
\]
where the second inequality uses equation \eqref{e:barqtexpbd} of \Lemma{t:barqtexpbd},  and  the last approximation follows from \Lemma{t:clGProperty}.   

Next, applying \eqref{e:barqtexpbd2} of \Lemma{t:barqtexpbd} and \Lemma{t:clGProperty}, \eqref{e:MismatchExt1} is obtained:
\[
\begin{aligned}
\| \partialQstar_{\barmu_t}^{-1} \barq_t  -  \barc_t \|
& \leq (  b_S  +   b_\Qstar   ) b_q^2 \int_0^t [1  +  \| \barq_{0} \|  ]  \tau (t - \tau)  e^{b_q t}  \,  d [\clG_{t,\tau} ]   +  o(1)
\\
& = o(1)  
,\qquad t \to \infty\,
\end{aligned}
\]
Using \eqref{e:MismatchExt1} and \eqref{e:AtpartialQstarbarmut}, \eqref{e:MismatchExt2} is obtained:
\[
\begin{aligned}
\barclA_t^{-1} \diagpie \barqcost_t 
& = \barclA_t^{-1} \diagpie \Bigl( \partialQstar_{\barmu_t}^{-1}\barq_t + o(1) \Bigr)  
\\
& = \barclA_t^{-1} \diagpie \Bigl( - \diagpie^{-1} \barclA_t \barq_t + o(\|\barq_t\|)    +    o(1)   \Bigr)  
\\
& = -\barq_t  +  o(1+\|q_t\|) \,,\qquad t\to \infty\, .
\end{aligned}
\]
\end{proof}

\begin{lemma}
\label{t:haq_Bdd_And_ODE_sol}
\begin{romannum}
\item 
$\displaystyle \sup_{t\ge 0}   \| \barq_t \|  <\infty\,,\quad a.s.$.
\item
For   $t\ge 0$, 
\begin{subequations}
\begin{align}
\barq_{T_0 +t}  - \barq_{T_0}   
& =
-   \int_{T_0}^{T_0+t }  \barclA_\tau^{-1} \diagpie \bigr\{c   -  \barqcost_\tau   \bigr\}  \, d \tau    + o(1)
\label{e:haq_ODE0}
\\
& =
-   \int_{T_0}^{T_0+t }   \bigr\{\barq_\tau   -  \partialQstar_{\barmu_\tau}   c   \bigr\}  \, d \tau    + o(1)
\label{e:haq_ODEa}
\\
& =
-   \int_{T_0}^{T_0+t }  \partialQstar_{\barmu_\tau}  \bigl\{  \barqcost_\tau    -     c   \bigr\}  \, d \tau    + o(1)
,\qquad T_0 \to \infty\,
\label{e:haq_ODEb}
\end{align}
\end{subequations} 
\end{romannum}
\end{lemma}

\begin{proof}
Equation \eqref{e:prepreODE_Zap_q} of \Lemma{t:barqtexpbd} implies,
\[
\barq_{T_0 +t} - \barq_{T_0} 
 = -   \int_{T_0}^{T_0+t }  \barclA_\tau^{-1}  \diagpie  \bigl\{c - \barqcost_\tau  + o(\|\barq_{\tau} \|)   \bigr\}  \, d\tau   +  o(1) ,\qquad T_0\to \infty\, .
\]
\Lemma{t:Aq1} along with the fact that $\{\barclA_t^{-1}: t \geq T_\bullet \}$ is bounded implies
\begin{equation}
\begin{aligned}
\barq_{T_0 +t} - \barq_{T_0}  
& = -   \int_{T_0}^{T_0+t }    \bigl\{  \barq_\tau   +   \barclA_\tau^{-1}  \diagpie c   + o(  1  +  \|\barq_\tau\|  )   \bigr\}  \, d\tau   +  o(1)
,\qquad T_0\to \infty\, .
\end{aligned}
\label{e:bartqbdd}
\end{equation}
Boundedness of $\{\barq_t \}$ is established by applying the Gr\"{o}nwall lemma \cite{bor08a}, and \eqref{e:haq_ODE0} immediately follows.

Next apply the approximation \eqref{e:AtpartialQstarbarmut}: substituting $- \barclA_{T}^{-1} \diagpie = \partialQstar_{\barmu_T}  + o(1)$ in \eqref{e:bartqbdd}, and using the fact that $\{ \barq_t\}$ is bounded, \eqref{e:haq_ODEa} is obtained.
Equation \eqref{e:haq_ODEb} then follows from \eqref{e:MismatchExt2} of \Lemma{t:Aq1}.
\end{proof}

With these results established, the  ODE approximation will quickly follow. 
For a fixed but arbitrary time-horizon $T>0$,  define a family of uniformly bounded and uniformly  Lipschitz continuous functions $\{\barGamma^{T_0} :  T_0\ge 0 \}$,  where  $\barGamma^{T_0} : [0,T] \to \Re^m$  for each $T_0\ge 0$ and some integer $m$.  The family of functions is constructed from   the following familiar components: for $t\in [0,T]$,
\[
\begin{aligned}
\barGamma^{T_0}_1(t)  =  \barq_{T_0+t}  
\,,\quad 
\barGamma^{T_0}_2(t)  =  \barqcost_{T_0+t}  
\,, \quad
\barGamma^{T_0}_3(t)  =  \int_{T_0}^{T_0+t}  \partialQstar^{-1}_{\barmu_r}\, dr
\,, \quad
\barGamma^{T_0}_4(t)  =  \int_{T_0}^{T_0+t}  \partialQstar_{\barmu_r}\, dr
\end{aligned}
\]
More precisely,   $\barGamma^{T_0}$   is a function of two variables,  $t\in [0,T]$ and $\omega\in\Omega$.  To say that  $\{\barGamma^{T_0}\}$ is  uniformly bounded and uniformly  Lipschitz continuous  means that there exists $\Omega_\bullet \in \clF$ with measure one such that for each $\omega \in \Omega_\bullet$, the family of functions $\{\barGamma^{T_0}(\omega,\varble) :  T_0\ge 0\}$,  is uniformly bounded and Lipschitz.  The bound and the Lipschitz constant may depend on $\omega$, but are independent of $T_0$.

Any sub-sequential limit of $\{\barGamma^{T_0} : T_0\ge 0\}$ will be denoted $\Gamma$.   
To maintain consistency with the notation in \Theorem{t:ZAP},  the first two components are denoted
\[
\Gamma_1(t) = q_t \,,\quad 
\Gamma_2(t) = c_t \,. 
\]
The ODE limit is recast as follows:
\begin{proposition}
\label{t:ODE_Gamma}
Any sub-sequential limit $\Gamma$ of  $\{\barGamma^{T_0} : T_0\ge 0\}$ must have the following form:   for $t\in[0,T]$,
\begin{romannum}
\item $
 q_t \eqdef\Gamma_1 (t)  =\Qstar(c_t) $.
  
\item $\displaystyle 
\ddt  c_t  = - c_t + c$.

\item  For a.e.\ $t\in [0,T]$, there is a pmf $\mu_t$ such that 
\[
\ddt \Gamma_3(t) =  \partialQstar^{-1}_{\mu_t}\qquad \text{\it and}\quad  
		\partialQstar^{-1}_{\mu_t}  q_t = c_t 
\]
\end{romannum}
\end{proposition}
The first relation, that $\Gamma_1 (t)  =\Qstar(\Gamma_2(t))$,  is obvious because the mapping $\Qstar$ is continuous.   The proofs of (ii) and (iii) are similar:  
prior results and a few results that follow are reinterpreted as properties of $\{\barGamma^{T_0}\}$ that are preserved in any sub-sequential limit.   For example, \eqref{e:haq_ODEb} admits the representation in terms of $\barGamma^{T_0}$:
\begin{equation}
\barGamma^{T_0}_1(t) -\barGamma^{T_0}_1(0)
+    \int_{0}^{t }   d \barGamma_4^{T_0}(\tau) \,  \bigl\{  \barGamma_2^{T_0}(\tau)       -     c   \bigr\}      = o(1)
,\qquad T_0 \to \infty\, .
\label{e:haq_ODEbGamma}
\end{equation}

The following result establishes that the left hand side represents a continuous functional of $\barGamma^{T_0}$:

\begin{lemma}
\label{t:ContinuousIntegration}  
For fixed $T>0$, $l>0$, and $b>0$, let $\clH_{l,b}$ denote the set of all functions  $h\colon[0,T] \to \Re$ satisfying $\|h\|\le b$,  and are Lipschitz continuous with Lipschitz constant $l$:
\[
| h(t) - h(s)| \le l |t-s|\,,\quad s,t\in[0,T],\ \ h\in\clH_{l,b}
\]
The set $\clH_{l,b}$ is compact as a subset of $C([0,T],\Re)$.  Moreover, the following real-valued functional is Lipschitz continuous on $\clH_{l,b}\times \clH_{l,b}$: 
\[
\clC_T(f,g) = \int_0^T  f(t) \, dg(t)
\]
\end{lemma}

\begin{proof}
Since $\clC_T$ is bilinear, it is sufficient to obtain Lipschitz constants in either variable.
 
For a fixed function $g \in \clH_{l,b}$, and any two functions $f_1,f_2 \in \clH_{l,b}$, 
\[
\begin{aligned}
| \clC_T(f_1,g) - \clC_T(f_2,g) |
& = \bigg | \int_0^T ( f_1(t)  -  f_2(t) ) \, dg(t) \bigg |
\\
& \leq  \int_0^T | f_1(t)  -  f_2(t) | \, dt 
\\
& \leq T l \| f_1  -  f_2 \|,
\end{aligned}
\]
which implies Lipschitz continuity of $\clC_T$ in its first variable, with Lipschitz constant $T l$.

A similar result is obtained for a fixed $f \in \clH_{l,b}$. Using integration by parts:
\[
\clC_T(f,g) = f(T)g(T) - f(0)g(0)  - \int_0^T  g(t) \, df(t).
\]
For any two functions $g_1,g_2 \in \clH_{l,b}$, 
\[
\begin{aligned}
| \clC_T(f,g_1) - \clC_T(f,g_2) | 
& \leq 2 \|f\| \cdot \| g_1 - g_2 \|
+ \bigg | \int_0^T ( g_1(t)  -  g_2(t) ) \, df(t) \bigg |
\\
& \leq (2b + Tl) \| g_1 - g_2 \|.
\end{aligned}
\]
This proves Lipschitz continuity of $\clC_T$ in its second variable, with Lipschitz constant $(2b + T l)$.
\end{proof}

  The next result implies another continuous relationship for $ \barGamma^{T_0}_3$:
\begin{lemma}
\label{t:limpartialQstar}
For each  $T_0\ge 0$ and $t>0$, there exists a pmf $\upmf$ such that
\[
\begin{aligned}
\frac{1}{t} \int_{T_0}^{T_0+t} \partialQstar_{\barmu_\tau}^{-1} \, d \tau\, = \partialQstar_{\upmf}^{-1},
\end{aligned}
\]
with $\barmu_\tau$ defined in \eqref{e:barmu}. 
That is,   the matrix  $ t^{-1}\barGamma^{T_0}_3(t)$ lies in the compact set   $ \{\partialQstar_{\nu}^{-1} : \nu \ \text{is a pmf}\}$.
\qed
\end{lemma}

\begin{lemma}
\label{t:ChainRule}
For any sub-sequential limit $\Gamma$, let $t_0$ denote a point  at which both  
$q_t$ and  $\qcost_t$
are differentiable.   
Let $\mu$ be any pmf satisfying $ \partialQstar_\mu\qcost_{t_0} = \Qstar(\qcost_{t_0}) = q_{t_0}$.   Then,
\[
 \ddt  q_t  =
 \partialQstar_\mu  \ddt \qcost_t ,\quad t=t_0\, .
\]
\end{lemma} 

\begin{proof}  
This is an instance of the chain rule.  A proof is provided since $\Qstar$ is not smooth.

These two functions are approximated by  a line at this time point:
\begin{equation}
\begin{aligned}
q_t  &=  q_{t_0} + (t-t_0) v^q  + o(|t-t_0|)
\\
  \qcost_t & =  \qcost_{t_0} + (t-t_0) v^c  + o(|t-t_0|)  
,\qquad t\sim t_0\, ,
\end{aligned}
\label{e:QODE2}
\end{equation}
where $v^q$, $v^c$ are the respective derivatives.  The lemma asserts that $  v^q  =  \partialQstar_\mu v^c    $.

Denote $L_t^q= q_{t_0} + (t-t_0) v^q$,  $L_t^c= \qcost_{t_0} + (t-t_0) v^c$,  $t\in\Re$.   Applying \eqref{e:Qstar} we have for each $t$: 
\[
\begin{aligned}
{\Qstar}(L_t^c)  &\le  \partialQstar_{\mu} L_t^c 
\\
&=   \partialQstar_{\mu} [c_{t_0} + (t-t_0) v^c]
\\
&=   {\Qstar}(c_{t_0}) + (t-t_0)  \partialQstar_{\mu}  v^c
\end{aligned}
\]
Differentiability of ${\Qstar}(c_t)  $ at $t_0$ then implies the desired conclusion: 
$\displaystyle
v^q = \ddt   {\Qstar}( L_t^c )\Big|_{t=t_0} =   \partialQstar_{\mu} v^c
$.
\end{proof}

\begin{proof}[Proof of \Prop{t:ODE_Gamma}] 
Recall that (i) has been established.  Result (iii) is established next, which will quickly lead to (ii).

\Lemma{t:limpartialQstar} implies the following   relationship for $ \barGamma^{T_0}_3$: For each  $T_0\ge 0$ and $t>0$, there exists a pmf $\upmf$ such that
\[
\begin{aligned}
\frac{1}{t} \int_{0}^{t} \barGamma^{T_0}_3(\tau) \, d \tau
= \partialQstar_{\upmf}^{-1},
\end{aligned}
\]
It follows that the same is true for any sub-sequential limit $\Gamma$:  There is a parameterized family of pmfs $\{\mu_\tau\}$ such that 
\[
\Gamma_3(t) = \int_0^t \partialQstar_{\mu_\tau}^{-1}\, d\tau\,,\quad 0\le t\le T\, .
\]
In the pre-limit we have $\ddt \barGamma^{T_0}_3(t) \times \ddt \barGamma^{T_0}_4(t) = I$ for each $t$ and $T_0$. It can be shown using Laplace transform arguments that the same must be true in the limit for a.e.\ $t$, giving the first half of (iii).
\archive{Insert Laplace transform proof in thesis}

Next, we prove the second half of (iii): $ \partialQstar^{-1}_{\mu_t}  q_t = c_t $ for a.e.\ $t$.
From equation \eqref{e:MismatchExt1} of \Lemma{t:Aq1}, 
\[
\int_{T_0}^{T_0 + t} {\partialQstar^{-1}_{\barmu_\tau}} \barq_\tau  \, d \tau - \int_{T_0}^{T_0 + t} \barqcost_\tau \, d \tau = o(1)
,\qquad T_0 \to \infty\,
\]
In $\barGamma^{T_0}$ notation:
\[
\int_{T_0}^{T_0 + t} d \barGamma_3^{T_0}(\tau) \, \barGamma_1^{T_0} (\tau) - \int_{T_0}^{T_0 + t} \barGamma_2^{T_0}(\tau) \, d \tau = o(1)
,\qquad T_0 \to \infty\,
\]
\Lemma{t:ContinuousIntegration} asserts that the left hand side of the above equation defines a continuous functional of $\barGamma^{T_0}$, and therefore the relationship also holds in the limit:
\[
\int_{0}^{t} d \Gamma_3(\tau) \, q_\tau = \int_{0}^{t} c_\tau \, d \tau.
\]
This establishes the second half of part (iii) of the lemma.

Part (ii) is obtained using similar arguments:  it is established that the left hand side of \eqref{e:haq_ODEbGamma} is a continuous mapping of $\barGamma^{T_0}$, so the relation is true for any sub-sequential limit $\Gamma$:
\begin{equation}
q_t -q_0
= - \int_{0}^{t }      \, d \Gamma_4(\tau) [  c_\tau      -     c   ]   
= - \int_{0}^{t }      \,   \partialQstar_{\mu_\tau} [  c_\tau      -     c   ] \,  d\tau \,.
\label{e:haq_ODEGamma}
\end{equation}
Combining \Lemma{t:ChainRule} with (iii) gives $ \partialQstar_{\mu_t}  \ddt c_t =   \ddt q_t $ at points of differentiability,   and thence \eqref{e:haq_ODEGamma} implies that for a.e.\ $t$
\[
\partialQstar_{\mu_t}  \ddt c_t 
=
\ddt q_t =  - \partialQstar_{\mu_t} [  c_t      -     c   ] 
\]
\end{proof} 

\clearpage

%
%
%
%
\archive{\bl{Is this right?? You said in the barmu notes that something about the existence of the partial derivative should be announced at the top}
	\\
	arx: We have to explain what we mean by an ODE.  What you have below is simply notation.     Let's discuss in person. }

\archive{It is the compactness that gives uniformly Lipschitz right?
	\\
	arx:   You tell me!  Look up the AA theorem
	\\
	arx: Uniformly Lipschitz is what we have just explained...}
\archive{\rd{I don't think we need to say equicontinuous.. because there is a version of the theorem for Lipschitz continuous and uniform boundedness.. It also says things about the limit! Not sure if that is part of the theorem or what... Wikipedia.. }
	\\
	arx:
	There is no limit at this stage!  the AA Theorem does not claim there is a limit.  Read the theorem!!!}

\clearpage


\clearpage

%
%
%
%


\small

\begin{table}[h]
	\footnotesize
	\caption{Notation}
\label{t:notation}
	\centering
	\begin{tabular}{lll}
		\toprule
		Symbol   & Type            & Description \\  
		\midrule
		$\state$          & set; $x$ is a component        & state space of the Markov chain     \\
		$\U$              & set; $u$ is a component        & action space of the Markov chain     \\
		$c$       		  & function $: \state \times \U \to \Re$ & cost function \\ 		
		$\beta$           & scalar $\in (0,1)$     & discount factor      \\
		$\theta$          & vector; $\in \Re^d$        & parameter vector  \\
		\multirow{2}{*}{$\psi$}& function $: \state \to \Re^d$ & basis functions for TD-learning  \\   &function $: \state \times \U \to \Re^d$ & basis functions for Q-learning\\
		$\alpha_n$        & scalar; $\in (0,1]$ & step-size sequence \\ 
		$\gamma_n$        & scalar; $\in (0,1]$ & step-size sequence \\ 		
		$h$       		  & function $: \state \to \Re$ & value function \\ 		
		$h_{\phi^\theta}$        & function $: \state \to \Re$  & a linear approximation to $h$: $h_{\phi^\theta} = \theta^\transpose \psi$ \\ 		
		$Q$       		  & function $: \state \times \U \to \Re$ &  {SARSA Q-function for an uncontrolled Markov chain} \\
		$Q^\theta$        & function $: \state \times \U \to \Re$  & a linear approximation to $Q$: $Q^\theta = \theta^\transpose \psi$ \\ 		
		$\theta^*$          & vector; $\in \Re^d$        & optimal parameter vector satisfying: $h = h^{\theta^*}$ or $Q = Q^{\theta^*}$ \\
		$\tiltheta$          & vector; $\in \Re^d$        & error in the parameter vector: $\tiltheta = \theta - \theta^*$ \\
		$\uq$       		  & operator & given $q: \state \times \U \to \Re$, $\uq (x) = \min_u q(x,u)$ \\
		\multirow{2}{*}{$\elig$}& function $: \state \to \Re^d$ & eligibility vector for TD-learning  \\   &function $: \state \times \U \to \Re^d$ & eligibility vector for Q-learning\\
		$\bfmX$          & sequence $\{X_n: n \geq 0 \}$        & uncontrolled Markov chain evolving on $\state$    \\
		$\bfPhi$          & sequence $\{\Phi_n: n \geq 0 \}$        &  {Markov chain of interest when applying stochastic approximation}    \\
		$\pie$          & function $: \state \times \U \to \Re$        & steady state distribution / probability mass function of $\bfPhi$ \\ 
		$\bfDelta$          & sequence $\{\Delta_n: n \geq 0 \}$        & error sequence    \\
		$\phi$        & function $: \state \to \U$  & policy \\ 	
		$\ell$        & scalar  & number of elements in $\state$: $\ell = |\state|$ \\ 		
		$\ell_u$      & scalar  & number of elements in $\U$: $\ell_u = |\U|$ \\ 
		$\nphi$       & scalar  & number of possible policies: $\nphi = (\ninp)^\nd$ \\ 		
		$P_u$        & operator   & tr.\ kernel (controlled): $P_u f\, (x) = \Expect [f(X_{n+1}) \mid X_n = x, \, U_n = u]$ \\
		$P_\phi$        & operator   & tr.\ kernel (uncontrolled): $P_\phi f\, (x) = \Expect [f(X_{n+1}) \mid X_n = x, \, U_n = \phi(x)]$ \\
		$S_\phi$        & operator   & given $q \colon \state\times\U\to\Re$, $S_\phi q\, (x) = q(x,\phi(x))$ \\ 		
		$h_\phi$        & function $: \state \to \Re$  & value function for a given policy $\phi$ \\ 
		$\phi^q$        & function $: \state \to \U$  & $q$-optimal policy: $\phi^q(x) \in  \argmin_uq(x,u)$ \\ 	
		$h_\phi$        & function $: \state \to \Re$  & value function for an uncontrolled Markov chain with policy $\phi$ \\ 		
		$h^*$        & function $: \state \to \Re$  & optimal value function \\ 
		$\Qstar$        & operator  & given $\qcost: \state \times \U \to \Re$, $\Qstar(\qcost)$ is the optimal $Q$-function for cost $\qcost$  \\ 		
		$Q^*$        & function $: \state \times \U \to \Re$  & optimal $Q$-function for cost $c$ \\ 
		$\BE$        & function $: \state \times \U \to \Re$  & Bellman error for the $Q$-function approximation \\ 		
		$\diagpie$        & matrix $\in \Re^{d \times d}$  & diagonal matrix: $\diagpie(k,k) = \pie(x^{(k)},u^{(k)})$ \\ 		
		$\Psi$        &  matrix $\in \Re^{d \times d}$   & outer product of the basis functions: $\Psi = \psi \times \psi^\transpose$ \\ 		
		$\bfmG$        & sequence $\{G_n: n \geq 0 \}$  & arbitrary sequence of matrix gains \\ 
		$G$        & matrix $\in \Re^{d \times d}$  & steady state mean of $\bfmG$  \\ 	
		$A_n$; $A$      & matrix $\in \Re^{d \times d}$  &  linearization in stochastic approximation;  $A$ s.s.\ mean \\ 
		$b_n$; $b$       & vector $\in \Re^d$  &  vector sequence;  $b$ s.s.\ mean  \\ 
		$\haA_n$    & matrix $\in \Re^{d \times d}$  & $n$-step Monte-Carlo estimate of $A$ \\
		$\hab_n$    & vector $\in \Re^d$  & $n$-step Monte-Carlo estimate of $b$ \\ 
		\bottomrule
	\end{tabular}
\end{table}

\end{document}